\documentclass[11pt]{article}
\usepackage[a4paper,margin=1in]{geometry}
\usepackage[T1]{fontenc}
\usepackage{lmodern,microtype}
\usepackage{amsmath,amssymb,amsthm,mathtools}
\usepackage{booktabs,array,enumitem,needspace}
\usepackage{xcolor}
\usepackage[round,authoryear]{natbib}
\usepackage{xurl}
\usepackage{hyperref}
\usepackage[capitalise,noabbrev]{cleveref}
\hypersetup{colorlinks=true,linkcolor=blue!40!black,
  citecolor=blue!40!black,urlcolor=blue!40!black,
  pdftitle={Quantum Submodular Maximization},
  pdfauthor={Yonggang Jiang; Xiaoming Sun; Penghui Yao; Zekun Ye; Jialin Zhang; Zhijie Zhang}}
\allowdisplaybreaks
\numberwithin{equation}{section}
\setlist[enumerate]{leftmargin=2em,itemsep=0.3em,topsep=0.4em}
\setlist[itemize]{leftmargin=2em,itemsep=0.3em,topsep=0.4em}
\newtheorem{theorem}{Theorem}[section]
\newtheorem{lemma}[theorem]{Lemma}
\newtheorem{corollary}[theorem]{Corollary}
\theoremstyle{definition}

\newcounter{algorithm}
\renewcommand{\thealgorithm}{\arabic{algorithm}}
\newsavebox{\algorithmframe}
\newenvironment{pseudocode}[1]{%
  \par\addvspace{0.6\baselineskip}%
  \begin{lrbox}{\algorithmframe}%
  \begin{minipage}{\linewidth}%
  \small\setlength{\parindent}{0pt}%
  \refstepcounter{algorithm}\hrule\smallskip
  \textbf{Algorithm \thealgorithm.\ #1}\par\smallskip
}{%
  \smallskip\hrule\end{minipage}\end{lrbox}%
  \Needspace{\dimexpr\ht\algorithmframe+\dp\algorithmframe+\baselineskip\relax}%
  \noindent\usebox{\algorithmframe}\par\addvspace{0.6\baselineskip}}
\newcommand{\Comment}[1]{\quad{\footnotesize\(\triangleright\) \textit{#1}}}
\newcommand{\E}{\mathbb E}
\newcommand{\Prb}{\mathbb P}
\newcommand{\R}{\mathbb R}
\newcommand{\OPT}{\mathrm{OPT}}
\newcommand{\one}{\mathbf1}
\newcommand{\cF}{\mathcal F}
\newcommand{\enc}{\operatorname{enc}}
\newcommand{\clip}{\operatorname{clip}}
\newcommand{\Tr}{\operatorname{Tr}}
\newcommand{\ii}{\mathrm i}
\newcommand{\GEst}{\mathsf{Gradient}}
\newcommand{\QDG}{\mathsf{DoubleGreedy}}
\newcommand{\QUSM}{\mathsf{QUSM}}
\newcommand{\Resp}{\mathsf R}
\newcommand{\Repair}{\mathsf{Repair}}

\newcommand{\eps}{\varepsilon}

\title{\bfseries Quantum Submodular Maximization}

\author{
 Yonggang Jiang\thanks{MPI-INF, D-66123 Saarbr\"ucken, Germany. Email: yjiang@mpi-inf.mpg.de}
\and
Xiaoming Sun
\thanks{Institute of Computing Technology, Chinese Academy of Sciences, Beijing 100190, China. Email: sunxiaoming@ict.ac.cn}
\and
Penghui Yao\thanks{State Key Laboratory for Novel Software Technology, Nanjing University, Nanjing 210023,
China;
Hefei National Laboratory, Hefei 230088, China. Email: phyao1985@gmail.com}
\and
Zekun Ye
\thanks{College of Computer and Data Science, Fuzhou University, Fuzhou 350108, China. Email: yezekun@fzu.edu.cn}
\and
Jialin Zhang
\thanks{Institute of Computing Technology, Chinese Academy of Sciences, Beijing 100190, China. Email: zhangjialin@ict.ac.cn}
\and
Zhijie Zhang
\thanks{School of Mathematics and Statistics, Fuzhou University, Fuzhou 350108, China.
Email: zzhang@fzu.edu.cn
}}
\date{}
\begin{document}
\maketitle
\begin{abstract}
We study the quantum query complexity of maximizing a non-negative
submodular function, considering both the unconstrained setting and,
for monotone functions, a cardinality constraint \(k\) on an
\(n\)-element ground set. In the exact reversible digital value-oracle
model, our unconstrained algorithm achieves an expected
\((1/2-\varepsilon)\)-approximation using only
\(O_\varepsilon(\log n)\) queries. In contrast, any classical randomized
algorithm that attains a fixed expected ratio above \(1/4\) requires
\(\Omega(n/\log n)\) queries \citep*{LiEtAl2022}, establishing an
exponential separation in query complexity. 
For cardinality-constrained maximization, we give a bounded-error quantum algorithm that achieves a $(1-1/e-\varepsilon)$-approximation using $\widetilde O_\varepsilon(\min\{\sqrt n,n/k\})$ queries. When $k=o(n)$, our algorithm achieves at least a quadratic speedup up to logarithmic factors over classical randomized algorithms \citep*{MirzasoleimanEtAl2015, PengRubinstein2025}. Moreover, when $k=cn$ for any fixed rational $c<1-1/e-\varepsilon$, the query complexity reduces to $O_{\varepsilon,c}(\log n)$, yielding an exponential separation from the classical $\Omega(n/\log n)$ lower bound \citep*{LiEtAl2022}.
%For the cardinality constraint, our improved algorithm returns a \((1-1/e-\varepsilon)\)-approximation with high probability using \(\widetilde O_\varepsilon(\min\{\sqrt n,n/k\})\) queries, yielding a quadratic speedup over the classical Stochastic Greedy algorithm \citep*{MirzasoleimanEtAl2015}, which uses \(O(n\log(1/\varepsilon))\)queries. 
We further prove quantum lower bounds of
\(\exp(\Omega(\varepsilon^2n))\) queries for achieving a ratio beyond
\(1/2+\varepsilon\) without constraints, and
\(\exp(\Omega(\varepsilon^2k))\) queries for exceeding
\(1-1/e+\varepsilon\) when \(k/n\le\varepsilon\).
These barriers demonstrate that quantum computation offers no
exponential speedup at these approximation thresholds.

The algorithms accelerate marginal evaluation by encoding many
random-context marginal gains into the local coherences of a single
constant-query phase experiment, allowing all coordinates to be
sampled together. Fourier reconstruction supplies approximate
gradients for unconstrained double greedy. Under a cardinality bound,
simultaneous marginal estimates guide pool filtering, while amplitude
amplification finds valuable batches; charging each search to the
gain it produces avoids paying separately for every greedy round.
\end{abstract}
\clearpage
\setcounter{tocdepth}{3}
\begingroup
\small
\tableofcontents
\endgroup
\clearpage

\section{Introduction}
Submodular maximization is a general model for selecting a collection
of useful objects whose contributions overlap. Its defining property
is diminishing returns: the benefit of adding an object is no larger
when more objects have already been selected. For example, a sensor
may be informative when few nearby locations are observed but
redundant after those locations have been instrumented. A useful
optimization algorithm must therefore account for interactions among
the selected objects, rather than simply rank their individual values.

This structure appears in influence maximization, where one chooses
seed users to maximize the expected spread of information in a
network \citep*{KempeKleinbergTardos2003}; in sensor placement, where
one chooses observations to reduce uncertainty
\citep*{KrauseSinghGuestrin2008}; and in document summarization,
where one selects representative sentences while controlling
redundancy \citep*{LinBilmes2011}. Coverage and graph-cut objectives
are basic combinatorial examples. An objective is \emph{monotone}
when every additional object can only help. In other applications
the objective is non-monotone: additional selections can have a
cost or create undesirable interactions. A cut function, for instance,
is submodular but need not increase when a vertex is added to one side.

We consider two fundamental problems on a ground set \(N\) with
\(n\) elements. In unconstrained maximization, the objective is a
nonnegative submodular function, and any subset is feasible. In
cardinality-constrained maximization, the function is also monotone,
and the output must contain at most \(k\) elements. For a set function
\(f\), denote the respective optimal values by
\[
 \OPT_{\mathrm u}(f)=\max_{S\subseteq N}f(S),\qquad
 \OPT_k(f)=\max_{\substack{S\subseteq N\\|S|\le k}}f(S).
\]
We abbreviate \(\OPT_{\mathrm u}(f)\) to \(\OPT_{\mathrm u}\) in the
unconstrained algorithm. When the objective is available only through
function-value queries, two questions arise: what approximation ratio
can be achieved, and how many queries are needed to achieve it?
These questions can have very different answers.

\paragraph{Classical approximation thresholds.}
For monotone cardinality-constrained maximization,
\citet*{NemhauserWolseyFisher1978} proved that greedy selection
achieves \(1-(1-1/k)^k\ge1-1/e\). The factor \(1-1/e\) is optimal
for polynomial-query algorithms in the general classical
value-oracle model \citep*{NemhauserWolsey1978}. For explicitly
represented instances, maximum coverage already excludes a
polynomial-time improvement by a fixed constant unless \(P=NP\)
\citep*{Feige1998}. These statements concern worst-case general
instances; they do not imply hardness for every individual budget,
such as the trivial monotone case \(k=n\).

The unconstrained problem has a different threshold. A uniformly
random subset achieves expected ratio \(1/4\), without making any
queries. \citet*{FeigeMirrokniVondrak2011}, whose conference version
appeared in 2007, obtained ratio \(1/3\) by deterministic local
search and \(2/5\) by randomized local search. They also proved an
exponential classical oracle barrier above \(1/2\). Subsequent
approximation algorithms reached \(0.41\) through simulated annealing
\citep*{OveisGharanVondrak2011} and \(0.42\) through structural
continuous greedy \citep*{FeldmanNaorSchwartz2011Structural}.
Randomized double greedy finally attained the optimal expected
ratio \(1/2\) \citep*{BuchbinderEtAl2015}.

\paragraph{Reducing classical query complexity.}
A direct implementation of cardinality greedy evaluates the remaining
elements in each of \(k\) rounds and uses \(O(nk)\) value queries.
Lazy greedy avoids many evaluations in practice
\citep*{Minoux1978}, but retains this general worst-case bound.
Decreasing-threshold greedy achieves \(1-1/e-\varepsilon\) using
\(O((n/\varepsilon)\log(n/\varepsilon))\) queries
\citep*{BadanidiyuruVondrak2014}. Stochastic Greedy further reduces
the count to \(O(n\log(1/\varepsilon))\), with the approximation
guarantee holding in expectation \citep*{MirzasoleimanEtAl2015}.
Clean linear-time algorithms give a deterministic
\(1-1/e-\varepsilon\) guarantee with \(O(n/\varepsilon)\) queries
\citep*{LiEtAl2022}. Thus the fastest established general classical
algorithms have \(O_\varepsilon(n)\) query complexity for fixed
accuracy, although their confidence and accuracy dependence differs.

For unconstrained maximization, randomized double greedy already uses
only \(O(n)\) queries \citep*{BuchbinderEtAl2015}. Simultaneous
fractional double greedy also gives \(1/2-\varepsilon\) with
\(O_\varepsilon(n)\) queries and a number of adaptive rounds
depending only on \(\varepsilon\)
\citep*{ChenFeldmanKarbasi2019}. On the lower-bound side,
\citet*[Theorem~4.3]{LiEtAl2022} show that any fixed expected ratio
above \(1/4\) requires \(\Omega(n/\log n)\) classical randomized
queries. Their cardinality lower bounds require \(\Omega(n)\)
queries for constant \(k\); for \(k=cn\) with fixed rational
\(c\), improving on the expected ratio \(c\) of a uniformly
random feasible set by a fixed amount requires
\(\Omega(n/\log n)\) queries. More recently, \citet*{PengRubinstein2025} prove an
\(\Omega(n/\log^2 n)\) query lower bound for every fixed constant
approximation ratio when \(k=o(n)\), even for additive objectives.
They also prove a near-linear lower bound for estimating the optimum
of a monotone submodular function to within a fixed factor in this
regime. Thus near-linear classical query complexity remains necessary
well beyond constant budgets. The distinction between these
parameter regimes will be important in our quantum comparison.

\paragraph{From classical to quantum value queries.}
Quantum computation permits reversible queries to coherent
superpositions of inputs. It can use interference to combine
information about many query addresses before measurement.
Grover search reduces an unstructured search over \(n\) candidates
from linear to \(O(\sqrt n)\) queries \citep*{Grover1996}, and
amplitude amplification extends this improvement to general
randomized sources \citep*{BHMT}. These operations suggest faster
submodular algorithms, but simply applying quantum search in every
greedy round would still pay separately for many rounds.
Moreover, one measurement of a superposition does not return all
its function values.

Our approach uses the value oracle to prepare a phase state.
Single-qubit measurements then sample marginal information for all
elements together. The main task is to turn those joint samples
into reliable optimization decisions, without assuming that the
multilinear extension or its gradient is directly queryable.
\citet*{LSZ} explicitly posed the quantum query complexity of
submodular maximization as an open question.
Quantum submodular minimization had already been studied by
\citet*{HRRS}, and \citet*{QuantumDR2026} give quantum algorithms
for continuous DR-submodular maximization with direct access to a
smooth continuous objective. The latter oracle is different from
the reversible set-value oracle studied here. To our knowledge,
this work gives the first quantum value-oracle approximation
algorithms for the two general discrete maximization problems
considered in this paper.

\subsection{Our results}
All theorems use an exact reversible digital value oracle. A query
can address any subset; only the returned solution must obey the
cardinality constraint. We count oracle calls, including calls
inside coherent computations and their inverses. The model is
specified in Section~\ref{sec:prelim}; the query bounds do not
also bound gate complexity, classical processing, or output length.

\begin{theorem}\label{usm:thm:main-upper}
Let \(f:2^N\to\R_{\ge0}\) be submodular. For every
\(0<\varepsilon<1/10\), a quantum algorithm returns \(S\subseteq N\)
with
\[
 \E f(S)\ge(1/2-\varepsilon)\OPT_{\mathrm u},
\]
using at most
\[
 O\!\left(\varepsilon^{-4}\log\frac{2n}{\varepsilon}
                              \log\frac2\varepsilon\right)
\]
queries on every execution path.
\end{theorem}
The algorithm combines the simultaneous double-greedy framework of
\citet*{ChenFeldmanKarbasi2019} with a noisy finite-step analysis
and a quantum gradient estimator. The estimator's query count
depends logarithmically on \(n\), while its expected sum of
coordinate errors is controlled by endpoint marginals.
For fixed \(\varepsilon\), the resulting \(O_\varepsilon(\log n)\)
bound gives an exponential separation from the classical
\(\Omega(n/\log n)\) lower bound above ratio \(1/4\) \citep*{LiEtAl2022}.

\begin{theorem}\label{card:thm:upper}
Let \(f:2^N\to\R_{\ge0}\) be monotone and submodular, and let
\(1\le k\le n\). For \(0<\varepsilon<1/4\) and \(0<\delta<1/2\),
a quantum algorithm returns a feasible set satisfying
\[
 f(S)\ge(1-1/e-\varepsilon)\OPT_k(f)
\]
with probability at least \(1-\delta\), using at most
\[
 \widetilde O\!\left(
       \min\{\varepsilon^{-3}\sqrt n,\varepsilon^{-4}n/k\}
       \right)
\]
queries on every execution path. The soft notation hides only
logarithmic factors in \(n,k,1/\varepsilon,1/\delta\).
When \(k=n\), returning \(N\) uses no queries.
\end{theorem}
The algorithm jointly estimates random-order marginal scores in a
pool, filters that pool into an exactly certified batch, and uses
amplitude amplification to find a valuable batch. Its analysis charges
search cost to the value obtained, rather than multiplying a
worst-case search cost by the number of updates. For fixed
accuracy and confidence, this gives a nearly quadratic improvement
over the near-linear classical query lower bound for \(k=o(n)\)
\citep*{PengRubinstein2025}, and matches the
\(\sqrt n\) dependence required for constant \(k\), up to
logarithmic factors.

The density-dependent term gives a stronger separation in another
regime. If \(k=cn\) for fixed rational
\(0<c<1-1/e-\varepsilon\), the explicit bound in
\eqref{card:eq:explicitcost} is \(O_{\varepsilon,c}(\log n)\)
at constant confidence. The classical bound is
\(\Omega_{\varepsilon,c}(n/\log n)\) \citep*{LiEtAl2022}. To compare the quantum
success-probability guarantee with the classical expected-ratio
bound, repeat a constant number of runs, query their output values,
and select the best feasible output. Its expected ratio can be
made strictly larger than \(c\), without changing either
asymptotic count.

Approaching the classical approximation thresholds from below is
therefore much cheaper with quantum queries. Crossing them remains
exponentially expensive.
\begin{theorem}\label{thm:thresholds}
In the same digital value-oracle model:
\begin{enumerate}[label=(\roman*)]
 \item For \(0<\varepsilon\le1/2\) and even \(n\) with
 \(\varepsilon^2n\) sufficiently large, any quantum algorithm
 guaranteeing ratio \(1/2+\varepsilon\) for every nonnegative
 submodular function with success probability at least \(2/3\)
 requires
 \(\Omega(\exp(\varepsilon^2n/8))\) queries.
 \item For \(0<\varepsilon<1/e\) and \(k/n\le\varepsilon\), any
 quantum algorithm guaranteeing a feasible
 \((1-1/e+\varepsilon)\)-approximation for every normalized
 nonnegative monotone submodular function with success probability
 at least \(2/3\) requires
 \(Q= \Omega(\exp(\varepsilon^2k/16384))\) queries.
\end{enumerate}
\end{theorem}
We use the hidden-partition construction of
\citet*{FeigeMirrokniVondrak2011} for the first part and the
hard family of \citet*{FK} for the second. Their exact agreement
with public reference functions permits a quantum hybrid argument
\citep*{BBBV}. In the cardinality case, a public slack perturbation
and dyadic rounding preserve submodularity, monotonicity, and
exact agreement, giving finite digital instances.
For unconstrained approximation in expectation, we additionally
obtain \(\Omega(\varepsilon\exp(\varepsilon^2n/8))\) queries
when \(\varepsilon^2n\ge C\log(2/\varepsilon)\);
see Corollary~\ref{lower:cor:usm-exponential}.
These lower bounds mean that exponential query barriers persist;
they do not exclude quadratic speedups within the exponential
scale or determine the optimal constants in the exponents.

\begin{theorem}\label{thm:card-constant}
Fix \(a\in(0,1]\) and \(1\le k\le n/2\).
Achieving ratio \(a\) for every normalized nonnegative monotone
submodular function with success probability at least \(2/3\)
requires \(\Omega_a(\sqrt{n/k})\) quantum value queries.
This holds even for capped modular functions with exact rational
values and permits queries to arbitrary, possibly infeasible sets.
\end{theorem}
Saturation makes large query sets uninformative, while every
successful feasible output must contain a hidden element.
The same hybrid argument then gives the lower bound.
It proves optimality up to logarithmic factors for constant \(k\);
it does not match the full upper-bound dependence for growing \(k\).

\subsection{Related work}
\paragraph{Monotone matroid constraints.}
Matroids describe feasibility conditions beyond a uniform size
budget, such as limits within categories and acyclicity in a graph.
Write \(r\) for the matroid rank.
Classical greedy gives a \(1/2\)-approximation
\citep*{FisherNemhauserWolsey1978}; continuous greedy followed
by rounding gives the optimal \(1-1/e\) ratio
\citep*{CalinescuChekuriPalVondrak2011}.
The oracle cost of implementations approaching this ratio has
progressively decreased: from the original high-degree polynomial
implementation to \(\widetilde O_\varepsilon(nr)\)
\citep*{BadanidiyuruVondrak2014}, then
\(\widetilde O_\varepsilon(n\sqrt r+r^2)\)
\citep*{BuchbinderFeldmanSchwartz2017}, and
\(\widetilde O_\varepsilon(n\sqrt r)\) through faster rounding
\citep*{KobayashiTerao2024}.
Non-oblivious local search gives
\(\widetilde O_\varepsilon(nr)\) deterministically and
\(\widetilde O_\varepsilon(n+r\sqrt n)\) with randomization
\citep*{BuchbinderFeldman2024Matroid}.
The subsequent algorithm of
\citet*{BuchbinderFeldman2025Extended} attains
\(1-1/e-\varepsilon\) with
\(O_\varepsilon(n\log(2r))\) value queries and
\[
 O_\varepsilon\!\left(n\log(2r)
                  +r^{3/2}\log^{3/2}(2r)\right)
\]
independence queries, for total
\(\widetilde O_\varepsilon(n+r^{3/2})\).
These are general independence-oracle bounds; a rank oracle is
a different access model.

Explicitly represented matroids admit further improvements in
running time. \citet*{EneNguyen2019Matroid} give nearly linear
algorithms for partition and graphic matroids.
\citet*{HenzingerLiuVondrakZheng2023} treat graphic, laminar,
and transversal matroids, with
\(O(n\varepsilon^{-5}\log^2(n/\varepsilon))\) value queries
and nearly linear time in the representation size.
Representation-specific data structures are central to these
results; their running-time bounds should not be assigned to
a general independence oracle.

\paragraph{Non-monotone constrained maximization.}
Measured continuous greedy gives \(1/e\) over a matroid
\citep*{FeldmanNaorSchwartz2011Unified}, and Random Greedy gives
the same ratio under a cardinality constraint using \(O(nk)\)
value queries \citep*{BuchbinderFeldmanNaorSchwartz2014}.
The approximation guarantee was improved to \(0.372\)
\citep*{EneNguyen2016}, then \(0.385\)
\citep*{BuchbinderFeldman2019Nonsymmetric}, and to \(0.401\)
\citep*{BuchbinderFeldman2024DR}.
The published \(0.401\) result applies to both matroid and
cardinality constraints through continuous optimization and
rounding, with polynomial query complexity. The known
value-oracle hardness for these constraints is \(0.478\)
\citep*{OveisGharanVondrak2011,Qi2024}.

Approximation quality and query complexity lead to different points
on the algorithmic frontier. For non-monotone cardinality maximization,
fast Random Greedy uses
\(O(n\varepsilon^{-2}\log(1/\varepsilon))\) queries for
ratio \(1/e-\varepsilon\)
\citep*{BuchbinderFeldmanSchwartz2017}.
Combinatorial algorithms give \(0.385-\varepsilon\) with
\(O(nk/\varepsilon)\) queries and \(0.377-\varepsilon\) with
\(O_\varepsilon(n\log(2k))\) queries
\citep*{ChenNathPengKuhnle2024}; another
\(0.385-\varepsilon\) algorithm uses
\(O_\varepsilon(n+k^2)\) queries
\citep*{TukanMualemFeldman2024}.
For general matroids, \citet*{BuchbinderFeldman2025Extended}
give \(1/e-\varepsilon\) with the same
\(\widetilde O_\varepsilon(n+r^{3/2})\) total-query bound
quoted above. The polynomial-query \(0.401\) guarantee is a
different point on this approximation--complexity frontier.
Our cardinality theorem concerns monotone functions.

\paragraph{Quantum optimization and gradient estimation.}
Besides quantum submodular minimization \citep*{HRRS} and
continuous DR-submodular maximization \citep*{QuantumDR2026},
quantum graph algorithms with cut queries provide related examples
of strong separations for aggregate set-value queries \citep*{LSZ}.
Quantum gradient estimation begins with \citet*{Jordan2005}
and is developed by \citet*{GAW2019}; spectral methods give
dimension-independent query bounds for analytic functions under
complex-domain phase access \citep*{ZhangShao2026}.
Our estimator instead starts from the reversible oracle for a
discrete set function. Access to its multilinear extension is
not supplied as an additional oracle. We obtain the required
random-context marginals directly from Boolean phase states,
and use submodularity to control the aggregate estimation error.

\subsection{Technique overview}
Both algorithms share a basic measurement principle: a
constant-query phase experiment exposes the characteristic
functions of many contextual marginal distributions in local
qubit coherences. Measuring all qubits provides one bounded
sample for each coordinate. Outcomes within a copy can be
correlated; independence is needed across fresh copies.
The two algorithms use different reconstructions and different
ways of translating estimation error into objective gain.

\subsubsection{Unconstrained maximization}
The optimization framework maintains fractional endpoints
\(x\le y\). At an element, the adding direction has gain
\(\partial_iF(x)\), and the deleting direction has gain
\(-\partial_iF(y)\). Simultaneous double greedy splits movement
in proportion to their positive gains, keeping all coordinate
gaps equal \citep*{ChenFeldmanKarbasi2019}.
We use this existing framework with noisy estimates. A simple
repair restores the constraint that the two gains sum to a
nonnegative number. Our finite-step inequality charges raw
estimation error in the \(\ell_1\) norm and charges the within-step
change of gradients to a decreasing gradient-difference potential.

The quantum estimator is the source of the query improvement.
A product-distribution superposition is given phase
\(e^{\ii\lambda f(S)}\) using a compute--phase--uncompute
sequence. On qubit \(i\), the two branches differ by the phase
of \(f(T+i)-f(T)\). Tracing out the other registers therefore
reveals the characteristic function of this marginal, averaged
over the random context \(T\). Equatorial measurements can
sample its real part at any selected phase angle, for all
coordinates on the same copy.

These trigonometric expectations are not yet marginal means.
We decompose each signed marginal into compact layers at
geometrically increasing magnitudes. A layer's Fourier
coefficients have bounded total absolute value, so sampling
a frequency in proportion to that value produces a bounded
estimator with constant variance, up to the known product-state
signal strength. Only \(O(\log(n/\varepsilon))\) scales are
needed, because the omitted smallest layers contribute a
controlled total additive error.

A second obstacle is that rare false layer estimates could
accumulate over many irrelevant coordinates. We screen each
layer with independent copies of a detector whose response is
exactly zero for an identically zero contextual marginal.
At scales much larger than a coordinate's marginals, its hit
probability is quadratically small. This makes each coordinate's
\emph{expected} error proportional to its endpoint marginal cap.
We can sum those expected errors without requiring a simultaneous
good event over the ground set. This avoids the confidence
allocation that would otherwise introduce extra dependence on \(n\).

Finally, a random bracket forces a small fraction of elements in
or out. It preserves most of an optimal set in expectation,
while its average endpoint potential is paid for by the initial
objective terms retained in the double-greedy inequality.
Interior initialization keeps all qubit signal strengths bounded
below by a quantity depending only on \(\varepsilon\). Thus the
number of iterations and the samples per scale depend only on
accuracy, and the remaining ground-set dependence is logarithmic.

\subsubsection{Cardinality-constrained maximization}
Classical threshold and stochastic greedy already use marginal
thresholds and random candidate sampling
\citep*{BadanidiyuruVondrak2014,MirzasoleimanEtAl2015}.
Our algorithm uses a random pool of size
\(\Theta(\varepsilon k)\), but evaluates its elements through
shared random-context scores rather than individual value queries.
A score is the expected marginal in a uniformly random
predecessor context, clipped at a density threshold.
Diminishing returns makes the sum of scores a certificate
for the gain of \emph{every} subcollection of the pool.

The clipped kernel has a triangular Fourier representation with
nonnegative coefficients summing to one. A sampled frequency
and a common product-state angle give bounded Pauli samples
for all pool coordinates. Mixing the angle converts independent
Bernoulli contexts into the predecessor distribution of a
uniform random ordering. Finite frequency and angle grids
control implementation bias. The joint-score circuit consequently
uses \(O_\varepsilon(\log k)\) value queries, including the
confidence needed for all pool elements.
Exact value checks certify the proposed filtered batch;
a whole-pool fallback handles the case where the pool's gain
is already large.

Joint scoring alone does not bound the total optimization cost:
many pools may be unproductive. We apply search with unknown
success probability \citep*{BBHT,BHMT} to this reversible source.
For a correct target, its normalized gain \(V\) has mean at least
\(\mu=\Theta(\varepsilon^2k/n)\), and every positive gain is
at least \(u=\Theta(1/k)\). Marking gains above
\(\max\{u,\mu/2\}\) relates expected search cost to the
expected gain of the returned batch. The normalized gains over
an entire trial sum to a constant, so cost is charged to this
sum, without a separate factor for the number of greedy updates.
This produces the two terms
\(\varepsilon^{-3}\sqrt n\) and
\(\varepsilon^{-4}n/k\). Public cutoffs turn the expected-cost
analysis into a query cap on every path. Certified target
search and repetition give high confidence when the optimum
is unknown.
The full algorithm is polylogarithmic for fixed density \(k/n\);
for smaller budgets its complexity is the sublinear bound in
Theorem~\ref{card:thm:upper}.

\subsubsection{Lower bounds}
All our hard families agree exactly with a public reference
function at every fixed address except for a small fraction of
hidden instances. Under the reference oracle, the algorithm's
query state is independent of the hidden instance. An oracle
replacement changes that state by at most twice the square root
of its weight on disagreeing addresses. Averaging and telescoping
over \(Q\) queries gives \(O(Q\sqrt p)\) distinguishability when
the disagreement probability is at most \(p\)
\citep*{BBBV}.

For unconstrained maximization, a balanced hidden partition
makes informative addresses exponentially rare in
\(\varepsilon^2n\) \citep*{FeigeMirrokniVondrak2011}.
Above \(1/2+\varepsilon\), a successful set must be informative.
For a sparse cardinality budget, the construction of \citet*{FK}
has the corresponding property with exponent
\(\varepsilon^2k\). Its original values need not have finite
digital encodings. We first add a public concave cardinality
term that supplies slack in monotonicity and submodularity,
then round all values by the same dyadic rule. Equality with
the reference is preserved exactly. A capped modular family
also gives the \(\Omega(\sqrt{n/k})\) bound by hiding an
element that every successful feasible output must contain.
One common hybrid lemma covers all three constructions.

\subsection{Paper structure}
Section~\ref{sec:prelim} gives the problem definitions, common
oracle model, quantum notation, and auxiliary tools.
Sections~\ref{usm:sec:framework}--\ref{usm:sec:complete} develop
the unconstrained algorithm: the noisy finite-step framework,
the quantum gradient estimator, and their combination.
Sections~\ref{sec:card-score}--\ref{sec:card-complete} develop
the cardinality algorithm: joint scores, certified batches,
search-cost charging, and the treatment of unknown targets
and small budgets. Section~\ref{sec:lower} contains all lower
bounds, beginning with a shared hybrid argument.

\section{Preliminaries and oracle model}\label{sec:prelim}
\subsection{Submodularity and product distributions}
Throughout, \(N\) is a finite ground set of size \(n\ge1\), and
\(f:2^N\to\R_{\ge0}\) is submodular. Write \(S+i=S\cup\{i\}\),
\(S-i=S\setminus\{i\}\), and
\[
 f(i\mid S)=f(S+i)-f(S).
\]
Submodularity means that
\(f(i\mid S)\ge f(i\mid T)\) whenever \(S\subseteq T\) and \(i\notin T\).
Equivalently, submodularity is the inequality
\(f(A)+f(B)\ge f(A\cup B)+f(A\cap B)\) for all \(A,B\subseteq N\).
The function is \emph{monotone} if \(f(A)\le f(B)\) whenever
\(A\subseteq B\), and \emph{normalized} if \(f(\varnothing)=0\).
Unconstrained maximization allows every subset; the cardinality problem
requires \(|S|\le k\), where \(0\le k\le n\) is an integer, and assumes
monotonicity. The unconstrained algorithm does not assume either
monotonicity or normalization. We use the optima \(\OPT_{\mathrm u}(f)\)
and \(\OPT_k(f)\) defined in the introduction, abbreviating the former
as \(\OPT_{\mathrm u}\) when the objective is clear. An expected
\(a\)-approximation requires \(\E f(S)\ge a\OPT_{\mathrm u}(f)\) or
\(\E f(S)\ge a\OPT_k(f)\), respectively. A guarantee with failure
probability \(\delta\) requires the corresponding value inequality to
hold with probability at least \(1-\delta\).
Let \(a^+=(a)_+=\max\{a,0\}\), and let \(\one_S\) be the incidence vector of \(S\).
The symbol \(\one\) denotes the all-ones vector on the current
coordinate set, and \(\mathbf1_A\) denotes the indicator of an event
\(A\). For real \(a\le b\), write
\(\clip_{[a,b]}(z)=\min\{b,\max\{a,z\}\}\).
Vector inequalities, maxima \(\vee\), and minima \(\wedge\) are
coordinatewise.

For a finite set \(W\) and \(q\in[0,1]^W\), let \(\mu_q\) be the
product distribution on its subsets:
\[
 \mu_q(R)=\prod_{i\in R}q_i\prod_{i\in W\setminus R}(1-q_i).
\]
Thus \(q\) is the vector of inclusion probabilities, also called the
vector of marginal probabilities. This use of \emph{marginal} is
distinct from an element's marginal \emph{gain}.
We write \(R(q)\sim\mu_q\) and \(\mu_{q,-i}\) for the product law
with coordinate \(i\) omitted. The global multilinear extension is
\[
 \cF(p)=\E_{R\sim\mu_p}f(R),\qquad p\in[0,1]^N.
\]
For a bracket \(X\subseteq Y\), put \(U=Y\setminus X\). Its restricted
extension is
\[
 F(q)=\E_{R\sim\mu_q}f(X\cup R),\qquad q\in[0,1]^U.
\]
The bracket will always be clear when \(F\) is used. For \(i\in U\), define
\[
 G_i(T)=f(i\mid X\cup T),\quad T\subseteq U-i,\qquad
 L_i=f(i\mid X),\quad M_i=f(i\mid Y-i).
\]
The paired symbols \(L_i,M_i\) denote the marginal in the smaller
and larger contexts, respectively. Finally, define
\begin{equation}\label{pre:eq:bracket-potential}
 \Phi(X,Y)=\sum_{i\in U}(L_i-M_i).
\end{equation}
These endpoint marginals and this potential are analysis quantities;
our algorithms do not evaluate them individually.

\begin{lemma}[See \citet*{Bach2013}]\label{pre:lem:dr}
The extension \(\cF\), and every restriction \(F\), is multilinear.
For a fixed bracket,
\begin{equation}\label{pre:eq:gradient-mean}
 \partial_iF(q)=\E_{T\sim\mu_{q,-i}}G_i(T).
\end{equation}
The extensions are DR-submodular: if \(p\le q\), then
\begin{equation}\label{pre:eq:dr}
 \partial_iF(p)\ge\partial_iF(q).
\end{equation}
Moreover, \(M_i\le G_i(T)\le L_i\), and the function
\(t\mapsto F(t\one)\) is concave on \([0,1]\).
\end{lemma}
\begin{proof}
Conditioning on whether \(i\) is present writes \(F(q)\) as an affine
function of \(q_i\), with slope equal to the right side of
\eqref{pre:eq:gradient-mean}. To compare \(p\) and \(q\), use independent
uniform thresholds for the other coordinates. This couples their
random subsets as \(T_p\subseteq T_q\); submodularity then gives
\(G_i(T_p)\ge G_i(T_q)\). Taking expectations proves \eqref{pre:eq:dr}.
The endpoint bounds follow by applying the same diminishing-returns
property to \(X\subseteq X\cup T\subseteq Y-i\).
Since \(F\) is a polynomial, its mixed second derivatives are
nonpositive by \eqref{pre:eq:dr}, and its pure second derivatives vanish
by multilinearity. Hence
\[
 \frac{d^2}{dt^2}F(t\one)
 =\sum_{i,j\in U}\partial_{ij}F(t\one)\le0.
\]
This proves the concavity assertion. The global case uses
\(X=\varnothing,Y=N\).
\end{proof}

The following standard inequality lower-bounds a product-distribution
expectation by a mixture of values on nested sets. We use it to
analyze random restrictions and fractional initialization.
\begin{lemma}\label{pre:lem:chain}
Let \(S_0\subseteq\cdots\subseteq S_r\), and let
\(\lambda_0,\ldots,\lambda_r\ge0\) sum to one. If
\(p=\sum_j\lambda_j\one_{S_j}\), then
\[
 \cF(p)\ge\sum_j\lambda_j f(S_j).
\]
\end{lemma}
\begin{proof}
The convex-closure characterization of the Lov\'asz extension
\citep*{Lovasz1983,Bach2013} states that, for a submodular function, its
Lov\'asz extension at \(p\) is the minimum of \(\E f(S)\) over all
distributions with inclusion probabilities \(p\).
For the displayed nested sets, its value is
\(\sum_j\lambda_jf(S_j)\). The product distribution \(\mu_p\) is one
feasible distribution in that minimization, proving the inequality.
If \(f(\varnothing)\ne0\), apply the characterization to
\(f-f(\varnothing)\) and add back the constant. This is only a proof
device; the algorithm always uses the original nonnegative function.
\end{proof}

\subsection{Quantum states and measurements}
We recall the notation needed for phase queries, simultaneous local
measurements, and the quantum hybrid argument; see
\citet*{NielsenChuang2010} for background.
A qubit has computational basis
\(|0\rangle=(1,0)^{\mathsf T}\), \(|1\rangle=(0,1)^{\mathsf T}\).
A register indexed by \(W\) has one qubit per element; \(|R\rangle\)
denotes the basis vector whose ones are exactly \(R\subseteq W\).
A pure state \(|\psi\rangle\) is a unit vector. Measuring a register
in its computational basis returns a basis label with probability
equal to the squared modulus of its amplitude. Tensor products
describe composite registers; \(A^{(i)}\) means that \(A\) acts on
qubit \(i\) and the identity acts on all other registers.
Here \(\ii^2=-1\), \(I\) is the identity operator, and
\(A^\dagger\) denotes the conjugate transpose of \(A\).
A unitary operator satisfies \(U^\dagger U=UU^\dagger=I\) and
describes a reversible state evolution.
A ket \(|\psi\rangle\) is a column vector; its bra
\(\langle\psi|=|\psi\rangle^\dagger\) is its conjugate transpose.
Thus \(\langle\phi|\psi\rangle\) is an inner product.

A density matrix \(\varrho\) is a positive semidefinite matrix of
trace one; the pure state \(|\psi\rangle\) corresponds to
\(|\psi\rangle\langle\psi|\). A reduced density matrix
\(\varrho_i=\Tr_{-i}\varrho\) describes qubit \(i\) after the other
registers are ignored. Its entry
\((\varrho_i)_{10}=\langle1|\varrho_i|0\rangle\) is the lower-left
entry in the ordered basis \((|0\rangle,|1\rangle)\).

An observable is a Hermitian matrix. Its measurement returns an
eigenvalue, with probabilities obtained by projection onto its
eigenspaces. Its expectation is
\[
 \langle A\rangle_\psi=\langle\psi|A|\psi\rangle,\qquad
 \langle A\rangle_{\varrho}=\Tr(\varrho A).
\]
The Pauli observables used below are
\[
 \sigma_x=\begin{pmatrix}0&1\\1&0\end{pmatrix},\qquad
 \sigma_y=\begin{pmatrix}0&-\ii\\\ii&0\end{pmatrix}.
\]
Each has eigenvalues \(+1,-1\), and direct matrix multiplication gives
\begin{equation}\label{pre:eq:pauli-coherence}
 \langle\sigma_x\rangle_{\varrho_i}
 =2\operatorname{Re}(\varrho_i)_{10},\qquad
 \langle\sigma_y\rangle_{\varrho_i}
 =2\operatorname{Im}(\varrho_i)_{10}.
\end{equation}
For an angle \(\vartheta\), define the equatorial observable
\begin{equation}\label{pre:eq:equatorial-definition}
 P_\vartheta=\cos\vartheta\,\sigma_x-\sin\vartheta\,\sigma_y.
\end{equation}
The Bloch coordinates of a qubit are the expectations of
\(\sigma_x,\sigma_y,\sigma_z\), where
\(\sigma_z=\operatorname{diag}(1,-1)\).
An equatorial observable has no \(\sigma_z\) component.
The observable \(P_\vartheta\) measures along the direction
\((\cos\vartheta,-\sin\vartheta,0)\) in the equatorial plane of the
Bloch sphere. Equivalently, it measures in the orthonormal basis
\[
 |\mathord{\pm}_\vartheta\rangle
 =\frac{|0\rangle\pm e^{-\ii\vartheta}|1\rangle}{\sqrt2},
\]
returning the corresponding sign. This basis convention follows the
minus sign in \eqref{pre:eq:equatorial-definition}. Operators on distinct
qubits commute, since
\((A\otimes I)(I\otimes B)=A\otimes B=(I\otimes B)(A\otimes I)\).
Thus these local measurements can be performed on every coordinate
of one state. Their outcomes may be correlated.

The operator norm is \(\|A\|=\sup_{\|v\|=1}\|Av\|\).
The group \(SU(2)\) consists of two-dimensional unitary matrices of
determinant one. The trace norm is \(\|A\|_1=\Tr\sqrt{A^\dagger A}\).
For density matrices, we use trace distance
\(D(\varrho,\varsigma)=\tfrac12\|\varrho-\varsigma\|_1\).
It bounds the difference in probabilities of any measurement event.
For pure states it is at most
\(\||\psi\rangle-|\phi\rangle\|\). These standard facts are also
valid after ancillary registers are ignored \citep*{NielsenChuang2010}.

\subsection{Oracle model and implementation conventions}
The objective values have a public finite exact digital encoding as
rational numbers, including fixed-point binary values. The maximum
encoding width is public. A \emph{value query}
is one application of the reversible XOR oracle
\[
 O_f:|S\rangle|z\rangle\longmapsto
 |S\rangle|z\oplus\enc(f(S))\rangle.
\]
The address register may be in an arbitrary superposition of subsets,
including subsets larger than the output budget \(k\). The algorithm
must output a feasible set, but its queries need not be feasible.
The inverse oracle equals the oracle itself. A query bound is a
worst-case cap on the number of such calls, unless explicitly stated
as an expected cost for an intermediate search routine.

For any real \(\lambda\), a query computes the encoded value, a
value-controlled phase applies \(e^{\ii\lambda f(S)}\), and a second
query clears the value register. Thus the phase transformation
\[
 |S\rangle\longmapsto e^{\ii\lambda f(S)}|S\rangle
\]
costs two value queries, independently of \(\lambda\). Known shifts,
residuals, and caps of \(f\) are computed reversibly from the same
encoded value. We do not assume an oracle for its multilinear extension.

Our common query model permits unrestricted inter-query unitaries:
state preparation, coefficient sampling, arithmetic, memory, and
classical processing are free in the query count. In particular,
processing \(n\) scores or writing the answer can require \(\Omega(n)\)
elementary operations. A polylogarithmic query bound does not assert a
polylogarithmic running time. The unconstrained upper bound uses
exact inter-query operations in this model. The cardinality upper
bound additionally has a finite-circuit implementation, specified next.

\paragraph{Finite-gate implementation of the cardinality algorithm.}
We use a fixed finite gate set consisting of inverse-closed generators of a dense subgroup of $SU(2)$, their exact controlled versions, and exact NOT, CNOT, Toffoli, Hadamard and phase-flip gates. Classical control, arithmetic, certificate predicates and query counters are exact; only public quantum rotations and state preparations require approximation. Seed decoders reversibly map every word to a legal value, retaining the original word, so all query addresses and pools remain valid.

No controlled-oracle primitive is assumed. To condition $O_f$ on a bit $c$, query into a zero workspace $w$, reversibly XOR $w$ into the intended answer register only when $c=1$, and uncompute $w$ with $O_f^{-1}$. This gives the exact controlled XOR oracle with two raw value queries. The constant overhead is included in all query bounds. In a deferred-measurement description, unused query slots can instead be padded with calls whose answer register is uniform over bit strings, which XOR leaves invariant.

\paragraph{Normalization for the cardinality algorithm.}
Query \(f(\varnothing)\) once and define
\(g(T)=f(T)-f(\varnothing)\).
For monotone \(f\), the function \(g\) is normalized, nonnegative,
monotone, and submodular. A set with \(g(S)\ge a\OPT_k(g)\), where
\(0\le a\le1\), satisfies
\[
 f(S)\ge a\bigl(\OPT_k(f)-f(\varnothing)\bigr)+f(\varnothing)
 \ge a\OPT_k(f).
\]
This reduction is used only in the cardinality part; the unconstrained
algorithm uses the original nonnegative objective.

\subsection{Concentration, search, and Fourier tools}

\begin{lemma}[Hoeffding Inequalities~\citep*{Hoeffding}, Theorems 1 and 4]\label{card:tool:hoeffding}
If $Y_1,\ldots,Y_D$ are independent variables in an interval of length $L$, then, for $t>0$,
\[\Prb\!\left[\left|D^{-1}\sum_{i=1}^D(Y_i-\E Y_i)\right|\ge t\right]
\le2\exp(-2Dt^2/L^2).\]
If $X$ is the sum of a uniform sample of $s$ entries drawn without replacement from a finite population in $[0,1]$, then
\[\Prb[X-\E X\ge t]\le\exp(-2t^2/s).\]
\end{lemma}

\begin{lemma}[Search with Unknown Success Probability~\citep*{BBHT,BHMT}]\label{card:tool:aa}
Let a reversible circuit $A$ prepare a state whose measurement is marked with probability $p>0$, and suppose the marking predicate is computable. Without knowing $p$, quantum search finds a marked output with probability one and uses $O(p^{-1/2})$ expected applications of $A,A^{-1}$ and the marking reflection. Its output distribution is the original distribution conditioned on being marked.
\end{lemma}
We impose a finite public query budget when using this procedure below.

\begin{lemma}[Finite Gate Approximation~\citep*{DN}]\label{card:tool:gates}
For the inverse-closed finite generators of a dense subgroup of $SU(2)$ chosen above, every specified $V\in SU(2)$ and every $\xi>0$ admit a finite word $\widetilde V$ satisfying $\|V-\widetilde V\|\le\xi$. This is an operator-norm approximation, without discarding a scalar phase. The public quantum rotations needed here therefore admit arbitrarily accurate finite synthesis while exact classical reversible subcircuits are left unchanged.
\end{lemma}
Only the existence of these approximations is needed. Their lengths are not part of the value-query complexity.

To preserve this guarantee under coherent control, replace a word by the product of its exact controlled generators. The resulting controlled operation $C(\widetilde V)=\operatorname{diag}(I,\widetilde V)$ fixes the control bit and satisfies
\[
 \|C(V)-C(\widetilde V)\|=\|V-\widetilde V\|.
\]
Any scalar phase required on a controlled branch is retained explicitly. Define
\[
 R_z(\theta)=\operatorname{diag}(e^{-i\theta/2},e^{i\theta/2}).
\]
Applying $C(R_z(-2\phi))$ to an auxiliary qubit initially in $|0\rangle$ multiplies the control-one branch by $e^{i\phi}$. Synthesize this determinant-one target rotation by the same rule and retain its auxiliary register. Thus a factorization $U=e^{i\phi}V$ never loses its branch-relative phase when controlled; the exact control bits are preserved even by the approximating circuit.

The following elementary Fourier convergence fact will justify reconstruction of the periodic, piecewise linear kernels used by the simultaneous score estimator.
\begin{lemma}[Fourier Convergence {\citep*[p.~77]{Tolstov}}]\label{card:tool:fourier-convergence}
Let $w$ be a continuous $2\Delta$-periodic function that is continuously differentiable on each of finitely many pieces of a period, with finite one-sided derivatives at their endpoints. Put
\[
 \widehat w_m=\frac1{2\Delta}\int_{-\Delta}^{\Delta}w(x)e^{-i\pi m x/\Delta}\,dx.
\]
Then, for every real $x$,
$w(x)=\lim_{J\to\infty}\sum_{m=-J}^{J}\widehat w_m e^{i\pi m x/\Delta}$.
\end{lemma}

\section{A noisy finite-step double-greedy framework}\label{usm:sec:framework}
We first describe the optimization process independently of how its
gradients are estimated. This separates the optimization argument from
the quantum measurement argument in Section~\ref{usm:sec:gradient}.
Fix a bracket and its restricted extension \(F\).

\subsection{The process and its intended effect}
The simultaneous double-greedy process of \citep*{ChenFeldmanKarbasi2019} maintains
two vectors \(x\le y\). Raising \(x_i\) has instantaneous gain
\(a_i=\partial_iF(x)\); lowering \(y_i\) has instantaneous gain
\(b_i=-\partial_iF(y)\). By Lemma~\ref{pre:lem:dr},
\(a_i+b_i\ge0\). If both gains are positive, we split the movement in
proportion to them. If only one direction is beneficial, the whole
movement uses that direction.

Estimates need not satisfy the same constraint. We therefore repair
a noisy pair before applying the response rule:
\begin{equation}\label{usm:eq:repair}
 \Repair(\widehat a,\widehat b)=
 \begin{cases}
  (\widehat a,\widehat b),&\widehat a+\widehat b\ge0,\\
  (-\widehat b,\widehat b),&\widehat a+\widehat b<0.
 \end{cases}
\end{equation}
This repair changes only the first coordinate when necessary.
For \(u+v\ge0\), define
\begin{equation}\label{usm:eq:response}
 \Resp(u,v)=
 \begin{cases}
  u/(u+v),&u>0,\ v>0,\\
  1,&u>0,\ v\le0,\\
  0,&u\le0.
 \end{cases}
\end{equation}
The convention at \((0,0)\) is \(0\).

Let \(\mathcal G(q)\) be any procedure returning a gradient estimate
at \(q\). Algorithm~\ref{usm:alg:dg} uses a fixed maximum step \(\zeta\).
Given the estimates, its rates are deterministic fractional velocities;
it does not toss a coin with probability \(r_i\) in each iteration.
In contrast to the gradient-based line search in \citep*{ChenFeldmanKarbasi2019},
we allow the gradients to change during a finite step and account for
that change explicitly in the proof.

\begin{pseudocode}{Noisy fractional double greedy
\(\QDG(F,\mathcal G,\rho,\zeta)\)}\label{usm:alg:dg}
\textbf{Input:} \(0<\rho<1/2\), \(\zeta>0\), and a gradient procedure \(\mathcal G\).
\begin{enumerate}
 \item \(x^0\gets\rho\one,\quad y^0\gets(1-\rho)\one,\quad
        d_0\gets1-2\rho,\quad t\gets0\).
 \item \textbf{while} \( d_t>0\):
 \begin{enumerate}[label=(\alph*),leftmargin=2em]
  \item \(\widehat a\gets\mathcal G(x^t),\quad
        \widehat b\gets-\mathcal G(y^t)\).
        \Comment{Adding and deleting estimates.}
  \item \((u_i,v_i)\gets\Repair(\widehat a_i,\widehat b_i),\quad
        r_i^t\gets\Resp(u_i,v_i)\quad(i\in U)\).
  \item \(s_t\gets\min\{\zeta, d_t\}\).
  \item \(x^{t+1}\gets x^t+s_tr^t,\quad
        y^{t+1}\gets y^t-s_t(\one-r^t)\).
  \item \( d_{t+1}\gets d_t-s_t,\quad t\gets t+1\).
        \Comment{All coordinate gaps shrink equally.}
 \end{enumerate}
 \item \textbf{return} \(w\gets x^t=y^t\).
\end{enumerate}
\end{pseudocode}

The bracket remains fixed. Only its fractional probabilities move.
The common-gap property makes termination predictable, and the
interior initialization keeps quantum measurements away from zero
signal strength. The returned vector is rounded only by the complete
algorithm in Section~\ref{usm:sec:complete}.

\subsection{Analysis of the process}
Define its gradient-difference potential by
\begin{equation}\label{usm:eq:pair-potential}
 \Psi(x,y)=\sum_{i\in U}
       \bigl(\partial_iF(x)-\partial_iF(y)\bigr).
\end{equation}
Each update decreases this potential. It therefore provides a budget
for the finite-step change of gradients without a line search.
The invariants in the next lemma follow the simultaneous process
of \citep*{ChenFeldmanKarbasi2019}; we include the proof and its explicit termination
bound.

\begin{lemma}\label{usm:lem:dg-invariants}
For every sequence of gradient estimates, Algorithm~\ref{usm:alg:dg}
preserves \(x^t\le y^t\),
\(y^t-x^t= d_t\one\), and
\(x^t,y^t\in[\rho,1-\rho]^U\).
If \(U\ne\varnothing\), it terminates after
\[
 K=\left\lceil(1-2\rho)/\zeta\right\rceil
\]
iterations. Along its updates, \(\Psi\) is nonnegative and nonincreasing.
\end{lemma}
\begin{proof}
The repair satisfies \(u_i+v_i\ge0\), and \eqref{usm:eq:response} gives
\(r_i\in[0,1]\). For a step \(s\le d\),
\[
 y'-x'=(y-x)-s\one=( d-s)\one.
\]
Thus \(x\le x'\le y'\le y\). This proves the order and interior
invariants inductively. Every nonfinal step has length \(\zeta\);
the last uses the remaining gap, proving the iteration count.
Lemma~\ref{pre:lem:dr} gives \(\Psi\ge0\).
As \(x\) increases, \(\partial_iF(x)\) decreases, and as \(y\)
decreases, \(\partial_iF(y)\) increases. Therefore \(\Psi\) decreases.
\end{proof}

\begin{lemma}\label{usm:lem:repair-local}
Let \(a+b\ge0\), let
\(e=|\widehat a-a|+|\widehat b-b|\), and set
\((u,v)=\Repair(\widehat a,\widehat b)\), \(r=\Resp(u,v)\).
Then \(|u-a|+|v-b|\le2e\). For either \(o\in\{0,1\}\), put
\[
 \mathcal L_o(a,b;r)=o(1-r)a+(1-o)rb.
\]
Then
\begin{equation}\label{usm:eq:local-noisy}
 ra+(1-r)b\ge2\mathcal L_o(a,b;r)-4e.
\end{equation}
\end{lemma}
\begin{proof}
If the estimate satisfies the constraint, its distance to \((a,b)\)
is \(e\). Otherwise its deficit
\(d=-(\widehat a+\widehat b)>0\) is at most \(e\), because \(a+b\ge0\).
The repair changes the first coordinate by \(d\), so the triangle
inequality bounds the new distance by \(e+d\le2e\).

For the repaired pair, first suppose \(u,v>0\). Both choices of \(o\)
give \(\mathcal L_o(u,v;r)=uv/(u+v)\), while
\[
 ru+(1-r)v-2\mathcal L_o(u,v;r)
 =\frac{(u-v)^2}{u+v}\ge0.
\]
If \(u>0,v\le0\), then \(r=1\), the left gain is \(u>0\), and
\(\mathcal L_1=0,\mathcal L_0=v\le0\). If \(u\le0\), the constraint
forces \(v\ge0\), \(r=0\), and the left gain is \(v\ge0\), whereas
\(\mathcal L_1=u\le0,\mathcal L_0=0\). Thus the same inequality holds
in every case. For fixed \(r,o\), the expression
\(rp+(1-r)q-2\mathcal L_o(p,q;r)\) is linear in \(p,q\);
each coefficient has magnitude at most \(2\).
Replacing \((u,v)\) by \((a,b)\) changes it by at most
\(2(|u-a|+|v-b|)\le4e\), proving \eqref{usm:eq:local-noisy}.
\end{proof}

\begin{lemma}\label{usm:lem:finite-step}
Suppose \(x,y\in[0,1]^U\), \(x\le y\),
\(y-x= d\one\), and \(0\le s\le d\).
Form rates from noisy estimates of
\(a_i=\partial_iF(x)\), \(b_i=-\partial_iF(y)\) using
\eqref{usm:eq:repair}--\eqref{usm:eq:response}, and set
\(x'=x+sr\), \(y'=y-s(\one-r)\).
For \(o\in\{0,1\}^U\), define
\(z=(o\vee x)\wedge y\), \(z'=(o\vee x')\wedge y'\).
If \(E=\sum_i(|\widehat a_i-a_i|+|\widehat b_i-b_i|)\), then
\begin{equation}\label{usm:eq:finite-step}
 \begin{split}
 [F(x')+F(y')]-[F(x)+F(y)]
 \ge{}&2[F(z)-F(z')]\\
 &-s[\Psi(x,y)-\Psi(x',y')]-4sE.
 \end{split}
\end{equation}
\end{lemma}
\begin{proof}
For \(0\le t\le s\), let
\[
 x(t)=x+tr,\quad y(t)=y-t(\one-r),\quad
 z(t)=(o\vee x(t))\wedge y(t).
\]
Define \(a_i(t)=\partial_iF(x(t))\),
\(b_i(t)=-\partial_iF(y(t))\), and
\[
 g(t)=\sum_i[r_ia_i(t)+(1-r_i)b_i(t)],\qquad
 \psi(t)=\sum_i[a_i(t)+b_i(t)].
\]
By the chain rule, \(g(t)\) is the derivative of
\(F(x(t))+F(y(t))\). Both \(a_i(t)\) and \(b_i(t)\) are nonincreasing
by Lemma~\ref{pre:lem:dr}. Their nonnegative decreases imply
\[
 \begin{split}
 g(0)-g(t)
 &=\sum_i\left[r_i(a_i(0)-a_i(t))
       +(1-r_i)(b_i(0)-b_i(t))\right]\\
 &\le\psi(0)-\psi(t)\le\psi(0)-\psi(s).
 \end{split}
\]
Let \(\ell(t)=-\frac{d}{dt}F(z(t))\) be the signed comparator loss.
Since \(o_i\) is binary, \(z_i(t)=y_i(t)\) if \(o_i=1\), and
\(z_i(t)=x_i(t)\) otherwise. Also the \emph{initial} endpoints obey
\(x\le z(t)\le y\). Lemma~\ref{pre:lem:dr} therefore gives
\[
 \begin{split}
 \ell(t)
 &=\sum_i\left[o_i(1-r_i)\partial_iF(z(t))
            -(1-o_i)r_i\partial_iF(z(t))\right]\\
 &\le\sum_i\mathcal L_{o_i}(a_i(0),b_i(0);r_i).
 \end{split}
\]
Summing \eqref{usm:eq:local-noisy} at the initial endpoints yields
\(g(0)\ge2\ell(t)-4E\), for every \(t\).
Combining the two bounds and integrating on \([0,s]\) gives
\[
 \int_0^sg(t)\,dt
 \ge2\int_0^s\ell(t)\,dt
       -s[\psi(0)-\psi(s)]-4sE.
\]
The three integrals and endpoint potentials are exactly the terms in
\eqref{usm:eq:finite-step}. No monotonicity of \(F(z(t))\) is assumed.
\end{proof}

\begin{lemma}\label{usm:lem:telescope}
For any realized estimates in Algorithm~\ref{usm:alg:dg}, fix
\(o\in\{0,1\}^U\) and let \(z^0=(o\vee x^0)\wedge y^0\).
Let \(E_t\) be the total raw error of its two gradient estimates
at iteration \(t\). Then its output satisfies
\begin{equation}\label{usm:eq:telescope}
 \begin{split}
 F(w)\ge{}&
 \tfrac14F(x^0)+\tfrac14F(y^0)+\tfrac12F(z^0)\\
 &-\tfrac\zeta4\Psi(x^0,y^0)-\sum_{t=0}^{K-1}s_tE_t.
 \end{split}
\end{equation}
\end{lemma}
\begin{proof}
Apply Lemma~\ref{usm:lem:finite-step} in every iteration, with the
same \(o\), and sum. At termination \(x^K=y^K=z^K=w\), so
\[
 2F(w)-F(x^0)-F(y^0)
 \ge2F(z^0)-2F(w)
   -\sum_ts_t(\Psi_t-\Psi_{t+1})-4\sum_ts_tE_t.
\]
Here \(\Psi_t=\Psi(x^t,y^t)\). By Lemma~\ref{usm:lem:dg-invariants},
these potential differences are nonnegative and
\(\Psi_K=0\). Since \(s_t\le\zeta\),
\(\sum_ts_t(\Psi_t-\Psi_{t+1})\le\zeta\Psi_0\).
Rearranging and dividing by four proves \eqref{usm:eq:telescope}.
\end{proof}

\section{A simultaneous quantum gradient estimator}\label{usm:sec:gradient}

Fix a bracket \(X\subseteq Y\) and a point \(q\in(0,1)^U\).
Our goal is to estimate all means in \eqref{pre:eq:gradient-mean}.
The algorithm below takes a scale \(B>0\) and an accuracy
\(0<\beta<1/4\). Its correctness analysis assumes \(B\ge\OPT_{\mathrm u}\);
the procedure itself has a bounded query schedule for every positive
\(B\). Empty \(U\) or zero \(B\) is handled before defining any
coordinate minimum or scale.

\subsection{Algorithm description}

\paragraph{Accessing trigonometric marginal statistics.}
For nonempty \(U\), define
\[
 \kappa_i=4q_i(1-q_i),\qquad \kappa=\min_{i\in U}\kappa_i.
\]
Choose the real product-state rotations
\[
 V_i=\begin{pmatrix}\sqrt{1-q_i}&-\sqrt{q_i}\\
                    \sqrt{q_i}&\sqrt{1-q_i}\end{pmatrix}.
\]
Their product prepares
\(\sum_{R\subseteq U}\sqrt{\mu_q(R)}|R\rangle\).
The compute--phase--uncompute operation of
Section~\ref{sec:prelim} then prepares
\begin{equation}\label{usm:eq:phase-state}
 |\psi_{\lambda,q}\rangle
 =\sum_{R\subseteq U}\sqrt{\mu_q(R)}
       e^{\ii\lambda f(X\cup R)}|R\rangle.
\end{equation}
For clarity, the value-register evolution is
\[
 |R\rangle|0\rangle
 \longmapsto |R\rangle|\enc(f(X\cup R))\rangle
 \longmapsto e^{\ii\lambda f(X\cup R)}
             |R\rangle|\enc(f(X\cup R))\rangle
 \longmapsto e^{\ii\lambda f(X\cup R)}|R\rangle|0\rangle.
\]
This costs two queries per independently prepared copy.

Let \(\varphi_i(\lambda)=\E_{T\sim\mu_{q,-i}}e^{\ii\lambda G_i(T)}\).
Lemma~\ref{usm:lem:phase-identity} below shows that an equatorial
measurement of qubit \(i\) yields a sign sample with mean
\(\sqrt{\kappa_i}\operatorname{Re}
(e^{\ii\vartheta}\varphi_i(\lambda))\).
Measurements on all coordinates can be made on the same copy.

\paragraph{Converting trigonometric statistics into means.}
These measurements provide sine and cosine expectations, whereas the
gradient requires \(\E G_i\). Directly approximating the mean at one
small frequency would require controlling both approximation bias and
sampling noise, including on coordinates with very small marginals.
We instead reconstruct a marginal from layers at different magnitudes.
This lets each coordinate's error scale with the size of its actual
contextual marginals rather than with the common upper bound \(B\).

Define the compact triangular function
\[
 h(u)=
 \begin{cases}
  2u-1,&1/2\le u\le1,\\
  2-u,&1\le u\le2,\\
  0,&\text{otherwise}.
 \end{cases}
\]
Set
\[
 \gamma=\beta/100,\qquad
 \tau_0=\frac{\beta B}{16n},\qquad
 \tau_s=2^s\tau_0,\qquad
 J=\min\{j:\tau_j\ge B\}.
\]
For \(0\le s\le J\), let
\[
 \ell_{s,+}(z)=\tau_sh(z/\tau_s),\qquad
 \ell_{s,-}(z)=\tau_sh(-z/\tau_s),\qquad
 u_{i,s,\nu}=\E\ell_{s,\nu}(G_i)\quad(\nu\in\{+,-\}).
\]
The positive and negative layers only respond when a marginal has the
appropriate sign and magnitude comparable to \(\tau_s\).
Lemma~\ref{usm:lem:layers} proves their reconstruction property.

To use the available phase measurements, we express the normalized
layers as periodic functions. Define them explicitly by
\[
 \mathsf h_{s,+}(z)=\sum_{r\in\mathbb Z}
                  h\!\left(\frac{z+8Br}{\tau_s}\right),
 \qquad
 \mathsf h_{s,-}(z)=\mathsf h_{s,+}(-z).
\]
They have period \(8B\). On \([-B,B]\) they equal
\(h(z/\tau_s)\) and \(h(-z/\tau_s)\), respectively.
Let \(\omega=\pi/(4B)\). The Fourier coefficients of these
\emph{periodized} functions are
\begin{equation}\label{usm:eq:fourier-definition}
 c_{s,\nu,k}=\frac1{8B}\int_{-4B}^{4B}
       \mathsf h_{s,\nu}(z)e^{-\ii k\omega z}\,dz.
\end{equation}
Choose \(D_s=\lceil16B/(\gamma\tau_s)\rceil\), and define
\[
 \mathsf P_{s,\nu}(z)
 =c_{s,\nu,0}
  +2\sum_{k=1}^{D_s}\operatorname{Re}
        (c_{s,\nu,k}e^{\ii k\omega z}),\qquad
 A_{s,\nu}=2\sum_{k=1}^{D_s}|c_{s,\nu,k}|.
\]
Lemma~\ref{usm:lem:fourier} proves that this choice has uniform error at
most \(\gamma/4\) and \(A_{s,\nu}=O(1)\).

Algorithm~\ref{usm:alg:layer} samples Fourier frequencies in proportion
to the coefficient magnitudes. If \(c_{s,\nu,k}\ne0\), its angle is
\(\arg c_{s,\nu,k}\), so that the equatorial measurement incorporates
the complex phase of that coefficient. It returns one estimate for
every coordinate. The number of copies is chosen using the known
\(\kappa\), not any unknown marginal.

\begin{pseudocode}{Fourier estimate of one normalized layer}\label{usm:alg:layer}
\textbf{Input:} \(X,Y,q,B,s,\nu\), and the parameters above.
\begin{enumerate}
 \item \(A\gets A_{s,\nu}\), \(c_k\gets c_{s,\nu,k}\),
       \(\delta_*\gets\gamma^2/100\).
 \item \textbf{if} \(A=0\), \textbf{return}
       \(v_i\gets\clip_{[0,1]}(c_0)\) for all \(i\in U\).
 \item \(m\gets\left\lceil
       \dfrac{32A^2}{\kappa\gamma^2}\log\dfrac2{\delta_*}
       \right\rceil\).
 \item \textbf{for} \(r=1,\ldots,m\), independently:
 \begin{enumerate}[label=(\alph*),leftmargin=2em]
  \item Draw \(K_r\) with
        \(\Prb(K_r=k)=2|c_k|/A\), \(1\le k\le D_s\).
  \item Prepare \(|\psi_{K_r\omega,q}\rangle\).
  \item Measure \(P_{\arg c_{K_r}}^{(i)}\) for every \(i\in U\);
        denote the outcomes by \(Z_i^{(r)}\in\{-1,1\}\).
 \end{enumerate}
 \item \textbf{return}
       \(v_i\gets\clip_{[0,1]}
       \left(c_0+\dfrac{A}{m\sqrt{\kappa_i}}
                          \sum_{r=1}^m Z_i^{(r)}\right)\).
\end{enumerate}
\end{pseudocode}

\paragraph{Suppressing contributions from inactive coordinates.}
An estimate of a zero layer can occasionally exceed a threshold.
If we controlled this only by making every coordinate correct
simultaneously, the confidence allocation would depend on the
ground-set size. We instead screen the estimate with an independent
test that has exactly zero response when \(G_i\equiv0\).

At scale \(\tau_s\), prepare the phase state at frequency \(1/\tau_s\),
apply \(\bigotimes_{i\in U}V_i^\dagger\), and measure in the computational
basis. Algorithm~\ref{usm:alg:mask} records whether each coordinate ever
returns one. The constant \(C_m=8\) is fixed throughout; it is unrelated
to the constants in the Fourier bounds. Different coordinates use
different initial segments of the same sequence of independent copies.

\begin{pseudocode}{Zero-response screens at scale \(s\)}\label{usm:alg:mask}
\textbf{Input:} \(X,Y,q,\tau_s,\gamma\).
\begin{enumerate}
 \item \(m_i\gets\lceil C_m/(\kappa_i\gamma)\rceil\) for \(i\in U\);
       \(m_{\max}\gets\max_i m_i\).
 \item \textbf{for} \(r=1,\ldots,m_{\max}\), independently:
       prepare \(|\psi_{1/\tau_s,q}\rangle\), apply
       \(\bigotimes_iV_i^\dagger\), and measure all coordinates,
       obtaining \(D_i^{(r)}\in\{0,1\}\).
 \item \textbf{return}
       \(H_{i,s}\gets\bigvee_{r=1}^{m_i}D_i^{(r)}\) for every \(i\in U\).
\end{enumerate}
\end{pseudocode}

The detector is not a test of whether the \emph{mean} marginal
vanishes: cancellation between positive and negative marginals need
not make it zero. Its exact null response applies to an identically
zero contextual marginal. At larger scales its hit probability is also
small for uniformly small contextual marginals, as quantified in
Lemma~\ref{usm:lem:mask}. This is what limits rare false estimates.

Algorithm~\ref{usm:alg:gradient} combines the estimates and screens.
The screen is shared between the two signs at a scale, but its copies
and randomness are independent of both Fourier estimates.
The estimates are first thresholded, then screened, then summed.
All loops have prescribed lengths.

\begin{pseudocode}{Simultaneous gradient estimate
\(\GEst(X,Y,q,B,\beta)\)}\label{usm:alg:gradient}
\textbf{Input:} \(q\in(0,1)^U\), \(B\ge0\), \(0<\beta<1/4\).
\begin{enumerate}
 \item \textbf{if} \(U=\varnothing\) or \(B=0\), \textbf{return} the zero vector.
 \item Form \(\kappa_i,\kappa,\gamma,\tau_0,\ldots,\tau_J\) as above;
       \(\widehat\mu_i\gets0\) for \(i\in U\).
 \item \textbf{for} \(s=0,\ldots,J\):
 \begin{enumerate}[label=(\alph*),leftmargin=2em]
  \item Obtain \(H_{i,s}\) from Algorithm~\ref{usm:alg:mask}.
  \item For each \(\nu\in\{+,-\}\), obtain \(v_{i,s,\nu}\)
        from Algorithm~\ref{usm:alg:layer}, using fresh independent copies.
  \item \(\widehat u_{i,s,\nu}\gets
        H_{i,s}\tau_sv_{i,s,\nu}\mathbf1_{\{v_{i,s,\nu}>\gamma\}}\).
        \Comment{Threshold, then screen.}
  \item \(\widehat\mu_i\gets\widehat\mu_i+
        \widehat u_{i,s,+}-\widehat u_{i,s,-}\quad(i\in U)\).
 \end{enumerate}
 \item \textbf{return} \(\widehat\mu\).
\end{enumerate}
\end{pseudocode}

\subsection{Analysis of the estimator}

We now prove the measurement identities, the layer approximation,
and the error guarantee. Every assertion in this subsection is made
under \(B\ge\OPT_{\mathrm u}\). Nonnegativity gives \(|G_i|\le B\).
The more local bounds from Lemma~\ref{pre:lem:dr} will determine the
aggregate error.

\begin{lemma}\label{usm:lem:phase-identity}
In the phase state \eqref{usm:eq:phase-state}, the reduced matrix of qubit
\(i\) satisfies
\[
 (\varrho_i)_{10}
 =\sqrt{q_i(1-q_i)}\,\varphi_i(\lambda).
\]
Consequently,
\begin{equation}\label{usm:eq:equatorial-mean}
 \E[P_\vartheta^{(i)}]
 =\sqrt{\kappa_i}
       \operatorname{Re}(e^{\ii\vartheta}\varphi_i(\lambda)).
\end{equation}
\end{lemma}
\begin{proof}
Separate qubit \(i\) from the other coordinates in
\eqref{usm:eq:phase-state}:
\[
 \begin{split}
 |\psi_{\lambda,q}\rangle
 =\sum_{T\subseteq U-i}&\sqrt{\mu_{q,-i}(T)}
 e^{\ii\lambda f(X\cup T)}
 \left(\sqrt{1-q_i}|0\rangle_i+
       \sqrt{q_i}e^{\ii\lambda G_i(T)}|1\rangle_i\right)
 |T\rangle_{-i}.
 \end{split}
\]
The \(|T\rangle_{-i}\) are orthonormal. Taking the partial trace
therefore removes cross terms with different \(T\), and the
lower-left entry is
\[
 \sqrt{q_i(1-q_i)}
 \sum_T\mu_{q,-i}(T)e^{\ii\lambda G_i(T)}.
\]
This is the asserted coherence. Substituting it into
\eqref{pre:eq:pauli-coherence} and \eqref{pre:eq:equatorial-definition}
proves \eqref{usm:eq:equatorial-mean}.
\end{proof}

\begin{lemma}\label{usm:lem:endpoint-cap}
Let \(\alpha_i=L_i^++(-M_i)^+\). Then \(|G_i|\le\alpha_i\), and
\begin{equation}\label{usm:eq:endpoint-cap}
 \sum_{i\in U}\alpha_i\le\Phi(X,Y)+f(X)+f(Y).
\end{equation}
\end{lemma}
\begin{proof}
Lemma~\ref{pre:lem:dr} bounds \(G_i\) between \(M_i\) and \(L_i\).
Its positive part is at most \(L_i^+\), and its negative part in
absolute value is at most \((-M_i)^+\), proving the pointwise bound.
Since \(L_i\ge M_i\), a scalar case split gives
\[
 L_i^++(-M_i)^+
 =(L_i-M_i)+(-L_i)\mathbf1_{\{L_i<0\}}
                       +M_i\mathbf1_{\{M_i>0\}}.
\]
Let \(D=\{i\in U:L_i<0\}\). Adding \(D\) to \(X\) in any order,
each marginal is at most its marginal at \(X\), so
\[
 f(X\cup D)-f(X)\le\sum_{i\in D}L_i.
\]
Because \(f(X\cup D)\ge0\), this implies
\(\sum_{i\in D}(-L_i)\le f(X)\).
Similarly, let \(C=\{i\in U:M_i>0\}\), and add its elements
to \(Y\setminus C\). Each context is contained in \(Y-i\), so each
marginal is at least \(M_i\). Hence
\[
 \sum_{i\in C}M_i\le f(Y)-f(Y\setminus C)\le f(Y).
\]
Summing the scalar identity and using \eqref{pre:eq:bracket-potential}
proves \eqref{usm:eq:endpoint-cap}.
\end{proof}

\begin{lemma}\label{usm:lem:layers}
For every \(|z|\le B\),
\begin{equation}\label{usm:eq:layer-residual}
 \left|z-\sum_{s=0}^J
       \bigl(\ell_{s,+}(z)-\ell_{s,-}(z)\bigr)\right|\le\tau_0.
\end{equation}
Also \(J+1=O(\log(2n/\beta))\) and \(B\le\tau_J<2B\).
\end{lemma}
\begin{proof}
Extend the dyadic scale sequence to all integer indices.
For \(z>0\), choose \(j\in\mathbb Z\) with
\(z=a\tau_j\), \(1\le a<2\). Only the layers at \(j,j+1\)
can have nonzero values, and
\[
 \tau_jh(a)+2\tau_jh(a/2)
 =\tau_j(2-a)+2\tau_j(a-1)=z.
\]
The negative layers give the corresponding identity for \(z<0\);
at zero all layers vanish.
By minimality of \(J\), and since \(\tau_0<B\),
\(B\le\tau_J<2B\). For \(s>J\), \(|z|/\tau_s\le1/2\),
so all omitted large layers are zero. Omitted small layers can
contribute only if \(|z|<\tau_0\). Their contributions to the
appropriate sign are nonnegative and sum to at most \(|z|\),
which proves \eqref{usm:eq:layer-residual}.
Finally, \(B/\tau_0=16n/\beta\), giving
\(J=\lceil\log_2(16n/\beta)\rceil\).
\end{proof}

\begin{lemma}\label{usm:lem:fourier}
The coefficients defined in \eqref{usm:eq:fourier-definition} satisfy,
for \(k\ne0\),
\[
 |c_{s,\nu,k}|
 \le\min\left\{\frac{3\tau_s}{32B},
               \frac{12B}{\pi^2\tau_s k^2}\right\},
 \qquad
 \sum_{k\in\mathbb Z}|c_{s,\nu,k}|=O(1).
\]
The truncation used by Algorithm~\ref{usm:alg:layer} satisfies
\[
 \sup_{z\in\R}
 |\mathsf h_{s,\nu}(z)-\mathsf P_{s,\nu}(z)|\le\gamma/4.
\]
\end{lemma}
\begin{proof}
Consider the positive layer and write \(\tau=\tau_s\).
By Lemma~\ref{usm:lem:layers}, its support in the fundamental interval is
contained in \([\tau/2,2\tau]\subset(-4B,4B)\).
The triangular function has area \(3\tau/4\), so
the integral definition gives \(|c_k|\le3\tau/(32B)\).

Its derivative is zero off the support, \(2/\tau\) on
\((\tau/2,\tau)\), and \(-1/\tau\) on \((\tau,2\tau)\).
The jumps have magnitudes \(2/\tau,3/\tau,1/\tau\), so the
total variation of the derivative is \(6/\tau\).
Both the function and derivative vanish at the endpoints of the
fundamental interval. Integrating by parts twice, with the second
derivative understood as a finite signed measure, gives
\[
 |c_k|\le\frac{6/\tau}{8B(k\omega)^2}
       =\frac{12B}{\pi^2\tau k^2}.
\]
These integrations can equivalently be done separately on the
three smooth subintervals, retaining the derivative jumps.

If \(\tau\le B\), split the nonzero frequencies at
\(\lceil B/\tau\rceil\). Below this point the area bound sums to
\(O(1)\); above it the \(k^{-2}\) bound also sums to \(O(1)\).
If \(B<\tau<2B\), the latter bound directly sums to \(O(1)\).
The constant coefficient is \(3\tau/(32B)=O(1)\).
This proves the coefficient-norm assertion.
The real periodic function has \(c_{-k}=\overline{c_k}\).
Its absolutely summable Fourier series represents the function,
and the uniform truncation error is at most
\[
 2\sum_{k>D_s}|c_k|
 \le\frac{24B}{\pi^2\tau D_s}
 \le\frac{24}{16\pi^2}\gamma<\gamma/4.
\]
Reflection preserves all the bounds for the negative layer.
\end{proof}

\begin{lemma}\label{usm:lem:layer-estimate}
For every fixed \(i,s,\nu\), Algorithm~\ref{usm:alg:layer} returns
\(v_{i,s,\nu}\in[0,1]\) such that
\begin{equation}\label{usm:eq:layer-confidence}
 \Prb\!\left[
 \left|v_{i,s,\nu}-u_{i,s,\nu}/\tau_s\right|>\gamma/2
 \right]\le\delta_*=\gamma^2/100.
\end{equation}
It uses \(O(\kappa^{-1}\gamma^{-2}\log(2/\gamma))\) copies.
For the thresholded estimate
\(\bar u_{i,s,\nu}=\tau_sv_{i,s,\nu}
\mathbf1_{\{v_{i,s,\nu}>\gamma\}}\), its good event in
\eqref{usm:eq:layer-confidence} implies
\(|\bar u_{i,s,\nu}-u_{i,s,\nu}|\le2\gamma\tau_s\).
If \(u_{i,s,\nu}=0\), then \(\bar u_{i,s,\nu}=0\) on that event.
\end{lemma}
\begin{proof}
Suppress \(s,\nu\) in the notation. If \(A=0\), the truncated
polynomial is the known constant \(c_0\), and
Lemma~\ref{usm:lem:fourier} gives the required deterministic accuracy.
Otherwise, for one repetition let
\(W_i=c_0+AZ_i/\sqrt{\kappa_i}\).
Condition on the sampled frequency \(k\) and use
\eqref{usm:eq:equatorial-mean}:
\[
 \begin{split}
 \E W_i
 &=c_0+\sum_{k=1}^{D_s}
       \frac{2|c_k|}{A}\frac{A}{\sqrt{\kappa_i}}
       \sqrt{\kappa_i}
       \operatorname{Re}(e^{\ii\arg c_k}\varphi_i(k\omega))\\
 &=c_0+2\sum_{k=1}^{D_s}
             \operatorname{Re}(c_k\varphi_i(k\omega))
 =\E\mathsf P_{s,\nu}(G_i).
 \end{split}
\]
The sample range has length \(2A/\sqrt{\kappa_i}\).
The independently prepared copies and independently sampled
frequencies make these samples independent across repetitions.
Hoeffding's inequality \citep*{Hoeffding} therefore gives
\[
 \Prb\left[|\overline W_i-\E W_i|>\gamma/4\right]
 \le2\exp\!\left(-\frac{m\kappa_i\gamma^2}{32A^2}\right)
 \le\delta_*.
\]
Because \(|G_i|\le B\), the true normalized layer equals
\(\mathsf h_{s,\nu}(G_i)\); Lemma~\ref{usm:lem:fourier} adds at most
\(\gamma/4\) bias. The two errors total at most \(\gamma/2\).
Clipping to \([0,1]\) cannot increase distance to the true normalized
mean in this interval. The copy bound follows from \(A=O(1)\).

On the good event, a retained estimate has error at most
\(\gamma\tau_s/2\). If it is discarded, \(v\le\gamma\) and
the true normalized mean is at most \(3\gamma/2\), giving error
at most \(3\gamma\tau_s/2\). If the true layer is zero,
the good event gives \(v\le\gamma/2\), so it is discarded.
Finally, on every event \(u,\bar u\in[0,\tau_s]\);
hence their error is at most \(\tau_s\).
\end{proof}

\begin{lemma}\label{usm:lem:mask}
For a single inverse-rotation measurement at frequency \(\lambda\),
\begin{equation}\label{usm:eq:null-response}
 \Prb(D_i=1\mid\lambda)
       =\kappa_i\E\sin^2(\lambda G_i/2).
\end{equation}
Let \(p_{i,s}=\Prb(\tau_s/2\le|G_i|\le2\tau_s)\) and
\(c_*=\sin^2(1/4)>0\).
Algorithm~\ref{usm:alg:mask} satisfies
\begin{equation}\label{usm:eq:mask-bounds}
 \begin{split}
 \Prb(H_{i,s}=0)&\le
             \exp(-c_*\kappa_i m_i p_{i,s}),\\
 \Prb(H_{i,s}=1)&\le
             \min\left\{1,\frac{C\alpha_i^2}{\gamma\tau_s^2}\right\}.
 \end{split}
\end{equation}
It uses \(O(\kappa^{-1}\gamma^{-1})\) copies.
If \(G_i\equiv0\), then \(H_{i,s}=0\) with certainty.
\end{lemma}
\begin{proof}
In the expansion used in Lemma~\ref{usm:lem:phase-identity}, apply
\(V_i^\dagger\) to the local conditional state.
Its \(|1\rangle\) amplitude becomes
\[
 e^{\ii\lambda f(X\cup T)}
 \sqrt{q_i(1-q_i)}
       (e^{\ii\lambda G_i(T)}-1).
\]
The squared modulus is
\(4q_i(1-q_i)\sin^2(\lambda G_i(T)/2)\).
Averaging over \(T\) proves \eqref{usm:eq:null-response}.
The inverse rotations on the other qubits do not affect this local
law, because a unitary on the other registers preserves the reduced
matrix of qubit \(i\). Exact zero response follows immediately.

At frequency \(1/\tau_s\), the absolute argument of the sine is
between \(1/4\) and \(1\) on the layer support. Its squared value
is therefore at least \(c_*\). Thus a single hit probability
\(\xi_{i,s}\) is at least \(c_*\kappa_i p_{i,s}\).
Since the first \(m_i\) copies are independent,
\[
 \Prb(H_{i,s}=0)=(1-\xi_{i,s})^{m_i}
 \le e^{-m_i\xi_{i,s}},
\]
proving the first bound in \eqref{usm:eq:mask-bounds}.
For the other direction, use \(\sin^2u\le u^2\) and
Lemma~\ref{usm:lem:endpoint-cap} to obtain
\[
 \xi_{i,s}\le\frac{\kappa_i\E G_i^2}{4\tau_s^2}
           \le\frac{\kappa_i\alpha_i^2}{4\tau_s^2}.
\]
A union bound over the \(m_i\) copies gives a hit probability
at most \(\min\{1,m_i\xi_{i,s}\}\).
The ceiling in \(m_i\) implies
\(\kappa_i m_i\le C_m/\gamma+\kappa_i
\le(C_m+1)/\gamma\), proving the second bound.
The maximum copy count is \(O(\kappa^{-1}\gamma^{-1})\).
\end{proof}

\begin{theorem}\label{usm:thm:gradient}
For every fixed bracket, \(B\ge\OPT_{\mathrm u}\), \(q\in(0,1)^U\), and
\(0<\beta<1/4\), Algorithm~\ref{usm:alg:gradient} returns
\(\widehat\mu\) satisfying
\begin{equation}\label{usm:eq:gradient-error}
 \E\sum_{i\in U}|\widehat\mu_i-\partial_iF(q)|
 \le C_{\mathrm{dec}}\beta
           \bigl(\Phi(X,Y)+f(X)+f(Y)+B\bigr),
\end{equation}
for an absolute constant \(C_{\mathrm{dec}}\).
Its deterministic query bound is
\begin{equation}\label{usm:eq:gradient-cost}
 O\!\left(\frac{\log(2n/\beta)}{\kappa\beta^2}
                         \log\frac2\beta\right).
\end{equation}
If \(U=\varnothing\) or \(B=0\), it needs no queries.
\end{theorem}
\begin{proof}
The empty case is immediate. If \(B=0\ge\OPT_{\mathrm u}\), nonnegativity
implies \(f\equiv0\), so the zero vector is correct.
Otherwise fix a coordinate \(i\), and write \(\alpha=\alpha_i\).
We bound its expected error without imposing a simultaneous good
event on all coordinates or layers.

Fix a scale \(s\) and a sign \(\nu\). Write \(u=u_{i,s,\nu}\),
\(\bar u=\bar u_{i,s,\nu}\), and \(H=H_{i,s}\).
Lemma~\ref{usm:lem:layer-estimate}, including its bound on all outcomes,
gives
\[
 \E|\bar u-u|\le2\gamma\tau_s+\delta_*\tau_s.
\]
Since \(H\in\{0,1\}\),
\[
 |H\bar u-u|
 \le H|\bar u-u|+(1-H)u.
\]

First consider the scales \(\tau_s\le2\alpha\).
The layer is bounded by \(\tau_s\) and supported on the event
defining \(p_{i,s}\), so \(u\le\tau_s p_{i,s}\).
By Lemma~\ref{usm:lem:mask} and \(\kappa_i m_i\ge C_m/\gamma\),
\[
 u\Prb(H=0)
 \le\tau_s p_{i,s}
             e^{-c_*C_mp_{i,s}/\gamma}
 \le C\gamma\tau_s,
\]
where the last step maximizes \(x e^{-c_*C_mx/\gamma}\) over
\(x\ge0\). Dropping \(H\) from the first nonnegative error term
therefore proves
\(\E|H\bar u-u|\le C\gamma\tau_s\).
The dyadic scales in this range sum to at most \(4\alpha\),
so both signs together contribute \(O(\gamma\alpha)\).

Next consider \(\tau_s>2\alpha\). By
Lemma~\ref{usm:lem:endpoint-cap}, \(|G_i|\le\alpha<\tau_s/2\),
so both true layers are zero. Lemma~\ref{usm:lem:layer-estimate} implies
that \(\bar u>0\) requires a failed Fourier estimate.
The screen and Fourier estimate use independent copies and
randomness. Since \(\bar u\le\tau_s\), Lemma~\ref{usm:lem:mask} gives
\[
 \E|H\bar u|
 \le\tau_s\delta_*\Prb(H=1)
 \le\frac{C\delta_*\alpha^2}{\gamma\tau_s}.
\]
If this range contains no scale, its contribution is zero.
Otherwise, for \(\alpha>0\), let \(\tau_*\) be its first scale. Then
\(\sum_{\tau_s>2\alpha}\tau_s^{-1}
\le2/\tau_*<1/\alpha\).
This also covers the case \(\tau_*=\tau_0\).
The total contribution is
\(O((\delta_*/\gamma)\alpha)=O(\gamma\alpha)\).
If \(\alpha=0\), then \(G_i\equiv0\); every screen is zero
with certainty, and this coordinate is estimated exactly.

By \eqref{pre:eq:gradient-mean} and Lemma~\ref{usm:lem:layers}, the error
from replacing \(G_i\) by its retained layers is at most \(\tau_0\).
The triangle inequality and the preceding bounds yield
\[
 \E|\widehat\mu_i-\partial_iF(q)|
 \le\tau_0+C\gamma\alpha_i.
\]
Summing over \(i\), using \(|U|\le n\),
\(n\tau_0=\beta B/16\), \(\gamma=\beta/100\), and
Lemma~\ref{usm:lem:endpoint-cap}, proves \eqref{usm:eq:gradient-error}.
Notice that the analysis used \(|G_i|\le B\) for reconstruction
and \(|G_i|\le\alpha_i\) for the local error; it does not require
\(\alpha_i\le B\).

Each scale uses two Fourier estimates and one screen. By
Lemmas~\ref{usm:lem:layer-estimate} and \ref{usm:lem:mask}, their combined
copy count is
\(O(\kappa^{-1}\gamma^{-2}\log(2/\gamma))\).
Lemma~\ref{usm:lem:layers} gives \(O(\log(2n/\beta))\) scales,
and each copy costs two digital queries. This proves
\eqref{usm:eq:gradient-cost}. The schedule remains bounded even when
an input \(B>0\) underestimates \(\OPT_{\mathrm u}\): correctness may then fail,
but \(J=\lceil\log_2(16n/\beta)\rceil\), the coefficient-norm bound,
and the number of copies are unchanged.
\end{proof}

\begin{corollary}\label{usm:cor:adaptive-gradient}
Suppose a bracket and a valid scale \(B\) are fixed, and \(q\) is
chosen as a function of an arbitrary previous history.
If Algorithm~\ref{usm:alg:gradient} uses fresh copies and randomness,
\eqref{usm:eq:gradient-error} holds conditionally on that history.
\end{corollary}
\begin{proof}
Conditioning fixes the point \(q\). The remaining randomness is exactly
that of a fresh run covered by Theorem~\ref{usm:thm:gradient}.
\end{proof}

The largest Fourier frequency used is
\(O(n/(\beta^2B))\). Its size does not change the query count in
the model of Section~\ref{sec:prelim}. This dependence does, however,
explain why exact value access and phase arithmetic were specified.
\section{The complete unconstrained algorithm}\label{usm:sec:complete}
The estimator's error depends on the endpoint potential of a bracket,
not only on the optimum. We therefore choose the bracket randomly.
The same randomization that preserves most of an optimal set also
relates the average potential to the values of the bracket endpoints.
Those values already appear with positive coefficients in
Lemma~\ref{usm:lem:telescope}; retaining them is what makes the two
components fit together.

\subsection{Algorithm description}
Fix a universal constant \(C_1\ge\max\{1,2C_{\mathrm{dec}}\}\), where
\(C_{\mathrm{dec}}\) is the constant in Theorem~\ref{usm:thm:gradient}.
Given the target error \(\varepsilon\), use
\[
 \eta=\varepsilon/3,\qquad
 \theta=\rho=\eta,\qquad
 \zeta=\theta/4,\qquad
 \beta=\theta/(64C_1).
\]
The algorithm first queries several independent uniformly random
sets. Their largest observed value gives a scale \(B\). It next
assigns each element independently to one of three categories:
forced in, forced out, or unresolved. This produces a bracket
\(X\subseteq Y\), within which it runs Algorithm~\ref{usm:alg:dg}.
Finally it samples a set from the returned fractional vector.

\begin{pseudocode}{Quantum unconstrained maximization
\(\QUSM(f,\varepsilon)\)}\label{usm:alg:usm}
\textbf{Input:} \(0<\varepsilon<1/10\); parameters as above.
\begin{enumerate}
 \item \(L\gets\left\lceil\log(1/\eta)/\log(7/6)\right\rceil\).
       Draw \(R_1,\ldots,R_L\sim\mu_{\one/2}\) independently and
       query \(v_j\gets f(R_j)\).
 \item \(V\gets\max_jv_j,\quad B\gets8V\).
       If \(V=0\), return \(R_1\).
       \Comment{A scale estimate without finding the optimum.}
 \item For each \(i\in N\), independently draw
       \(\ell_i\in\{\mathrm{in},\mathrm{out},\mathrm{free}\}\) with
       probabilities \((\theta,\theta,1-2\theta)\).
       Set
       \[
        X\gets\{i:\ell_i=\mathrm{in}\},\qquad
        Y\gets N\setminus\{i:\ell_i=\mathrm{out}\},\qquad
        U\gets Y\setminus X.
       \]
       If \(U=\varnothing\), return \(X\).
 \item \(\mathcal G(q)\gets\GEst(X,Y,q,B,\beta)\).
       \Comment{Use fresh estimator randomness on every call.}
 \item \(w\gets\QDG(F,\mathcal G,\rho,\zeta)\).
 \item Draw \(R\sim\mu_w\); return \(S\gets X\cup R\).
       \Comment{Independent product rounding.}
\end{enumerate}
\end{pseudocode}

Here \(F\) is the restricted extension for the selected bracket,
as defined in Section~\ref{sec:prelim}. It specifies the objective
used in the analysis; the algorithm
does not evaluate its expectation exactly. Its only access to \(F\)
is through the quantum estimator, implemented with queries to \(f\).
In Algorithm~\ref{usm:alg:dg}, the second call is negated because lowering
\(y_i\) has gain \(-\partial_iF(y)\), whereas the estimator always
returns an estimate of \(+\partial_iF\).

\subsection{Analysis of the initialization}
\begin{lemma}\label{usm:lem:scale}
If \(\OPT_{\mathrm u}>0\), the scale in Algorithm~\ref{usm:alg:usm} satisfies
\(\OPT_{\mathrm u}\le B\le8\OPT_{\mathrm u}\) with probability at least \(1-\eta\).
The scale step makes \(O(\log(1/\eta))\) queries.
\end{lemma}
\begin{proof}
Let \(O\) be an optimal set and let \(R\sim\mu_{\one/2}\).
Conditionally on \(R\setminus O=T\), the inclusion-probability vector
of \(T\cup(R\cap O)\) is the midpoint of the nested sets \(T,T\cup O\).
Lemma~\ref{pre:lem:chain} gives
\[
 \E[f(R)\mid R\setminus O=T]
 \ge\tfrac12 f(T)+\tfrac12 f(T\cup O)
 \ge\tfrac12 f(T\cup O).
\]
The distribution of \(O\cup(R\setminus O)\) is in turn the product
distribution at the midpoint of \(O,N\). Another application of
Lemma~\ref{pre:lem:chain} gives
\(\E f(R)\ge\frac14(f(O)+f(N))\ge\OPT_{\mathrm u}/4\).
This is the standard random-set guarantee of \citep*{FeigeMirrokniVondrak2011}.
Since \(0\le f(R)\le\OPT_{\mathrm u}\), putting
\(p_*=\Prb[f(R)\ge\OPT_{\mathrm u}/8]\) yields
\[
 \OPT_{\mathrm u}/4\le\E f(R)
 \le p_*\OPT_{\mathrm u}+(1-p_*)\OPT_{\mathrm u}/8,
\]
and hence \(p_*\ge1/7\).
Independence bounds the probability that all \(L\) samples have value
below \(\OPT_{\mathrm u}/8\) by \((6/7)^L\le\eta\).
Always \(V\le\OPT_{\mathrm u}\); whenever one sample reaches \(\OPT_{\mathrm u}/8\),
the scale \(B=8V\) lies in the asserted interval.
\end{proof}

For the remainder of this subsection, the bracket is drawn with
parameter \(\theta\), independently of the scale step. Fix an optimal
set \(O\) and define its clipped version
\[
 O_{\mathrm{br}}=(O\cup X)\cap Y,\qquad
 D=\E[f(X)+f(Y)].
\]
Expectations here are over the bracket.
\begin{lemma}\label{usm:lem:bracket}
The random bracket satisfies
\begin{equation}\label{usm:eq:bracket-credit}
 \E f(O_{\mathrm{br}})\ge(1-2\theta)\OPT_{\mathrm u},\qquad
 \E\Phi(X,Y)\le\frac{1-2\theta}{\theta}\,D.
\end{equation}
\end{lemma}
\begin{proof}
An element of \(O\) belongs to \(O_{\mathrm{br}}\) with probability
\(1-\theta\); an element outside \(O\) belongs with probability
\(\theta\). The choices are independent across elements. Its global
inclusion-probability vector is therefore
\[
 (1-2\theta)\one_O+\theta\one_N+\theta\one_\varnothing.
\]
Applying Lemma~\ref{pre:lem:chain} to
\(\varnothing\subseteq O\subseteq N\), and using nonnegativity,
proves the first inequality of \eqref{usm:eq:bracket-credit}.

To analyze the potential, introduce the scalar functions
\[
 \mathcal H(t)=\cF(t\one)+\cF((1-t)\one),\qquad
 \mathcal P(t)=\mathcal H'(t)
 =\sum_{i\in N}
   [\partial_i\cF(t\one)-\partial_i\cF((1-t)\one)].
\]
For \(0\le t\le1/2\), Lemma~\ref{pre:lem:dr} makes each summand
nonnegative. As \(t\) increases, its first gradient decreases and
its second gradient increases. Thus \(\mathcal P\) is nonnegative
and nonincreasing on this interval.
Since \(X\sim\mu_{\theta\one}\) and
\(Y\sim\mu_{(1-\theta)\one}\), we have \(D=\mathcal H(\theta)\).
The two sets need not be independent for this identity.
Consequently,
\[
 D=\mathcal H(0)+\int_0^\theta\mathcal P(t)\,dt
   \ge\theta\mathcal P(\theta),
\]
because \(f\) is nonnegative.

Conditionally on \(i\in U\), the labels of all other elements
retain their independent laws. Thus \(X\) contains every other
element with probability \(\theta\), and \(Y-i\) contains it with
probability \(1-\theta\). By \eqref{pre:eq:gradient-mean},
\[
 \E[L_i\mid i\in U]=\partial_i\cF(\theta\one),\qquad
 \E[M_i\mid i\in U]=\partial_i\cF((1-\theta)\one).
\]
The \(i\)-th coordinate in these gradient arguments is immaterial
by multilinearity. Since \(\Prb[i\in U]=1-2\theta\), summing
\eqref{pre:eq:bracket-potential} gives
\(\E\Phi(X,Y)=(1-2\theta)\mathcal P(\theta)\).
Combining this identity with the preceding bound on \(D\) proves
the second inequality.
\end{proof}

\begin{lemma}\label{usm:lem:inset}
For a fixed bracket, set \(x^0=\rho\one\),
\(y^0=(1-\rho)\one\), \(o=\one_{O_{\mathrm{br}}\setminus X}\),
and \(z^0=(o\vee x^0)\wedge y^0\). Then
\[
 F(z^0)\ge(1-2\rho)f(O_{\mathrm{br}}),\qquad
 F(x^0)+F(y^0)\ge f(X)+f(Y),\qquad
 \Psi(x^0,y^0)\le\Phi(X,Y).
\]
\end{lemma}
\begin{proof}
The global inclusion-probability vector corresponding to \(z^0\) is
\[
 \rho\one_X+(1-2\rho)\one_{O_{\mathrm{br}}}+\rho\one_Y.
\]
Indeed its coordinates are \(1\) on \(X\), \(1-\rho\) on
\(O_{\mathrm{br}}\setminus X\), \(\rho\) on
\(Y\setminus O_{\mathrm{br}}\), and \(0\) outside \(Y\).
Lemma~\ref{pre:lem:chain} applied to the three nested sets gives the
first assertion after dropping nonnegative endpoint terms.
By Lemma~\ref{pre:lem:dr}, the restricted diagonal is concave, so
\[
 F(\rho\one)\ge(1-\rho)f(X)+\rho f(Y),\qquad
 F((1-\rho)\one)\ge\rho f(X)+(1-\rho)f(Y).
\]
Adding proves the second assertion. Finally the gradient in the
smaller context is at most \(L_i\), and that in the larger context
is at least \(M_i\), by Lemma~\ref{pre:lem:dr}. Summing their differences,
as defined in \eqref{usm:eq:pair-potential}, proves the last assertion.
\end{proof}

\subsection{Approximation guarantee and query count}
\begin{proof}[Proof of Theorem~\ref{usm:thm:main-upper}]
If \(\OPT_{\mathrm u}=0\), every output satisfies the guarantee. Suppose
\(\OPT_{\mathrm u}>0\), and first fix a successful scale realization
\(\OPT_{\mathrm u}\le B\le8\OPT_{\mathrm u}\), and a bracket with \(U\ne\varnothing\).
Conditionally on any iteration history, the two gradient calls
have total expected raw error at most
\[
 \E[E_t\mid\text{history}]
 \le C_1\beta\bigl(\Phi(X,Y)+f(X)+f(Y)+B\bigr),
\]
by Corollary~\ref{usm:cor:adaptive-gradient}.
Negating both an estimate and its true value does not change its
absolute error. Because
\(\sum_ts_t=1-2\rho\le1\), the tower property gives the same
upper bound for \(\E\sum_ts_tE_t\).
There is no event on which every coordinate or every iteration
must be estimated correctly.

Apply Lemma~\ref{usm:lem:telescope} and then Lemma~\ref{usm:lem:inset}.
The coefficient \(1/4-C_1\beta\) is positive for our parameters.
After taking the expectation over the estimator, we obtain
\begin{equation}\label{usm:eq:fixed-bracket-bound}
 \begin{split}
 \E[F(w)\mid B,X,Y]\ge{}&
 \tfrac12(1-2\rho)f(O_{\mathrm{br}})
 +(1/4-C_1\beta)[f(X)+f(Y)]\\
 &-(\zeta/4+C_1\beta)\Phi(X,Y)-C_1\beta B.
 \end{split}
\end{equation}
When \(U=\varnothing\), the algorithm returns \(X=Y=O_{\mathrm{br}}\).
In this case \(\Phi=0\), and the right side of
\eqref{usm:eq:fixed-bracket-bound} is at most
\((1-\rho-2C_1\beta)f(X)-C_1\beta B\le f(X)\).
Thus the same inequality is valid for the actual output in that case.

Average \eqref{usm:eq:fixed-bracket-bound} over the independent random
bracket. Lemma~\ref{usm:lem:bracket} yields
\[
 \begin{split}
 \E[F(w)\mid B]\ge{}&
 \tfrac12(1-2\rho)(1-2\theta)\OPT_{\mathrm u}-C_1\beta B\\
 &+\left[\frac14-C_1\beta
 -\frac{1-2\theta}{\theta}(\zeta/4+C_1\beta)\right]D.
 \end{split}
\]
Here and below \(F(w)\) means \(f(X)\) for an empty bracket.
With \(\theta=\rho=\eta\), \(\zeta=\eta/4\), and
\(C_1\beta=\eta/64\), the coefficient of \(D\) is at least
\[
 \frac14-\frac{\eta}{64}-\frac5{64}>0.
 \]
The positive initial objective terms have therefore paid for all
potential-dependent errors. Since \(B\le8\OPT_{\mathrm u}\), we conclude that
\[
 \E[F(w)\mid B]
 \ge\left[\tfrac12(1-2\eta)^2-\frac{\eta}{8}\right]\OPT_{\mathrm u}
 \ge(1/2-17\eta/8)\OPT_{\mathrm u}.
\]
Product rounding has conditional expected value \(F(w)\).
On a failed scale realization, every output still has nonnegative
value. Lemma~\ref{usm:lem:scale} consequently implies
\[
 \E f(S)\ge(1-\eta)(1/2-17\eta/8)\OPT_{\mathrm u}
 \ge(1/2-3\eta)\OPT_{\mathrm u}=(1/2-\varepsilon)\OPT_{\mathrm u}.
\]

It remains to bound the queries for every realization.
Lemma~\ref{usm:lem:dg-invariants} gives \(K=O(1/\eta)\) iterations.
Every gradient argument has coordinates in \([\rho,1-\rho]\), so
its signal parameter satisfies
\(\kappa=\min_i4q_i(1-q_i)\ge4\rho(1-\rho)=\Omega(\eta)\).
By Theorem~\ref{usm:thm:gradient}, each of the two calls in an iteration
costs
\[
 O\!\left(\eta^{-3}\log\frac{2n}{\eta}
                        \log\frac2\eta\right)
\]
queries. The deterministic schedule bound in that theorem remains
valid if a positive scale underestimates \(\OPT_{\mathrm u}\).
If \(V=0\), the algorithm stops after the scale step.
The \(O(\log(1/\eta))\) scale queries and the optional empty-bracket
exit are dominated by the claimed bound.
Substitution of \(\eta=\varepsilon/3\) completes the proof.
\end{proof}

For a fixed approximation loss, Theorem~\ref{usm:thm:main-upper} uses
\(O(\log n)\) queries, whereas any fixed classical improvement over
\(1/4\) needs \(\Omega(n/\log n)\) queries
\citep*[Theorem~4.3]{LiEtAl2022}. This is a comparison with an existing
classical lower bound; no new classical lower bound is claimed here.

\section{A simultaneous quantum marginal-score estimator}
\label{sec:card-score}\label{card:sec:score}
\subsection{Algorithm and fixed-target notation}\label{card:sec:targetnotation}
In this part, $g$ is the normalized cardinality objective from
Section~\ref{sec:prelim}. The public target scale is separate from the scale used by the
unconstrained gradient estimator.

Choose a positive rational target $M$ and define
\begin{equation}\label{card:eq:residual}
G_M(T)=\min\{g(T),M\},\quad \Delta=M-G_M(S),\quad f_S(T)=G_M(S\cup T)-G_M(S).
\end{equation}
We use this notation only before accepting the target, so $\Delta>0$. Lemma~\ref{card:lemma:residual} establishes the residual properties
needed by the score circuit. The comparison set used in its analysis is not
known to the algorithm.
Choose rational $\nu\in[\eps/32,\eps/16]$, and set
\begin{equation}\label{card:eq:params}
\tau=(1-\nu)\Delta/k,\qquad b=\lfloor\nu k/32\rfloor.
\end{equation}
Sections~\ref{sec:card-score}--\ref{sec:card-search} assume $k\ge64/\nu$. Then $x=\nu k/32\ge2$, so $\lfloor x\rfloor\ge x-1\ge x/2$, giving $\nu k/64\le b\le\nu k/32$. Section~\ref{card:sec:complete} handles smaller $k$. A pool $R$ is a uniformly random $b$-subset of $N$, including elements already in $S$; those elements have zero residual marginal and are removed before adding a batch.

Throughout this section, $R\subseteq N$ is a fixed pool of size $b\ge1$. The residual function $f_S$ is normalized, monotone, and submodular, and satisfies $0\le f_S(A)\le\Delta$, where $\Delta>0$. Marginals use the convention from the preliminaries. Recall that $\tau=(1-\nu)\Delta/k$, so $0<\tau<\Delta$. If $\pi$ is a uniformly random ordering of $R$, let $R_{<e}^{\pi}$ denote the elements preceding $e$. Define
\[
 s_e(R)=\E_{\pi}\min\left\{\frac{f_S(e\mid R_{<e}^{\pi})}{\tau},1\right\},
 \qquad e\in R.
\]
These scores provide a simultaneous certificate for every subcollection of the pool. They also admit simultaneous estimation: one quantum sample produces a bounded estimator for every coordinate, using two value queries in total.

The score is an average over random predecessor contexts. It is designed so that
selecting several elements with large scores certifies the value of their
entire batch, rather than only their singleton values. To estimate these scores,
we reconstruct the bounded kernel
\[
 \chi(x)=(1-|x|/\tau)_+,\qquad -\Delta\le x\le\Delta,
\]
extended with period $2\Delta$. Its nonnegative Fourier coefficients form a
probability distribution. We use
$\operatorname{sinc}(0)=1$ and $\operatorname{sinc}(x)=\sin(x)/x$ for $x\ne0$. Choosing one frequency from this distribution and
one common angle produces all pool-coordinate samples together. The angle
mixture converts independent Bernoulli contexts into the predecessor law of a
uniformly random ordering. This is why the sample cost is independent of the
pool size; the only pool-size factor comes from the confidence union bound.

\begin{pseudocode}{Simultaneous random-order scores}\label{card:alg:score}
\textbf{Input:} Pool $R$, residual objective $f_S$, gap $\Delta$, threshold $\tau$,
and accuracy parameter $\nu$.
\begin{enumerate}
\item Define
\[
 c_m=\frac{\tau}{2\Delta}\operatorname{sinc}^2
       \left(\frac{\pi m\tau}{2\Delta}\right),\quad
 L=\left\lceil\frac{256\Delta}{\tau\nu}\right\rceil,\quad
 \eta_L=1-\sum_{|m|\le L}c_m .
\]
Use the finite distribution $\widetilde c_0=c_0+\eta_L$,
$\widetilde c_m=c_m$ for $0<|m|\le L$.
\item Choose integers
$K\ge16\pi^2\sqrt{|R|}/\nu$ and
$D=\lceil2048\nu^{-2}\log(16|R|)\rceil$.
Choose rational $a_\nu\in(0,2)$ with
$|a_\nu-\pi/2|\le\nu/64$.
\item For $j=1,\ldots,D$, independently choose
$m_j\sim\widetilde c$ and
$\theta_j\in\{(\ell+1/2)\pi/(2K):0\le\ell<K\}$ uniformly.
Prepare
\[
 \bigotimes_{e\in R}(\cos\theta_j|0\rangle+\sin\theta_j|1\rangle),
 \qquad
 |T\rangle\longmapsto
 e^{i\pi m_jf_S(T)/\Delta}|T\rangle .
\]
Record all Pauli-$X$ outcomes $Z_e^{(j)}\in\{-1,1\}$.
\item Return
$\widehat s_e=D^{-1}\sum_{j=1}^{D}(1-a_\nu Z_e^{(j)})$ for every $e\in R$.
\end{enumerate}
The public rotations and phases are synthesized so that one joint sample
has total variation error at most $\nu/256$ from the specified finite experiment.
All seeds, outcome bits and work registers are retained when the circuit will
be used coherently inside amplitude amplification.
\end{pseudocode}

The phase uses a compute--phase--uncompute sequence: query $f(S\cup T)$,
compute $f_S(T)$ using the cached background value, apply the phase,
and reverse the query. Thus one joint sample uses two value queries.
The finite frequency cutoff, angle grid and rational coefficient make the
source an invertible finite circuit. Their errors are accounted for in the
analysis below.

\subsection{Analysis of estimator}

\begin{lemma}\label{card:lemma:residual}
The functions $G_M$ and $f_S$ are monotone and submodular.
Moreover, $f_S(\varnothing)=0$ and $0\le f_S(T)\le\Delta$ for every $T$.
If $M\le\OPT_k(g)$, there exists $|O|\le k$ such that
\begin{equation}\label{card:eq:comparator}f_S(O)=\Delta,\qquad f_S(O\cup T)=\Delta\quad\text{for every }T.\end{equation}
\end{lemma}
\begin{proof}
To check truncation explicitly, for $e\notin A$ consider two cases. If $g(A)\ge M$, then $G_M(A)=G_M(A\cup\{e\})=M$, giving marginal zero. If $g(A)<M$, then
\[G_M(e\mid A)=\min\{g(A)+g(e\mid A),M\}-g(A)=\min\{g(e\mid A),M-g(A)\}.\]
Combining the cases gives
\[G_M(e\mid A)=\min\{g(e\mid A),(M-g(A))_+\}.\]
Both arguments are nonnegative and nonincreasing when $A$ grows: the first by submodularity, the second by monotonicity of $g$. Their minimum therefore has both properties, proving that $G_M$ is monotone and submodular. Also $f_S(\varnothing)=0$ and $f_S(e\mid A)=G_M(e\mid S\cup A)$. Enlarging $A$ enlarges this base, so $f_S$ is also monotone and submodular. Finally,
\[0\le f_S(A)=G_M(S\cup A)-G_M(S)\le M-G_M(S)=\Delta.\]

For a correct target, select a feasible $O$ with $g(O)\ge M$.
Monotonicity then gives $G_M(S\cup O\cup T)=M$ for every $T$,
which proves~\eqref{card:eq:comparator}.
\end{proof}

\begin{lemma}[Joint Density Certificate]\label{card:lemma:density}
For every $A\subseteq R$,
\[
 f_S(A)\ge \tau\sum_{e\in A}s_e(R).
\]
\end{lemma}
\begin{proof}
Fix an ordering $\pi$ of $R$, and restrict it to $A$. The predecessors $A_{<e}^{\pi}$ of an element $e\in A$ satisfy $A_{<e}^{\pi}\subseteq R_{<e}^{\pi}$. Diminishing returns gives
\[
 f_S(e\mid A_{<e}^{\pi})\ge f_S(e\mid R_{<e}^{\pi})
 \ge \min\{f_S(e\mid R_{<e}^{\pi}),\tau\}.
\]
To see the telescoping explicitly, write the restricted order as $a_1,\ldots,a_{|A|}$. Then
\[
 \begin{aligned}
 \sum_{j=1}^{|A|}f_S(a_j\mid\{a_1,\ldots,a_{j-1}\})
 &=\sum_{j=1}^{|A|}\bigl[f_S(\{a_1,\ldots,a_j\})-f_S(\{a_1,\ldots,a_{j-1}\})\bigr]\\
 &=f_S(A),
 \end{aligned}
\]
where the empty prefix has value zero. Thus
\[
 f_S(A)\ge\sum_{e\in A}\min\{f_S(e\mid R_{<e}^{\pi}),\tau\}.
\]
Taking expectation and using linearity yields
\[
 f_S(A)\ge\sum_{e\in A}\E_\pi
 \min\{f_S(e\mid R_{<e}^{\pi}),\tau\}
 =\tau\sum_{e\in A}s_e(R).
\]
\end{proof}

\begin{lemma}[A simultaneous quantum marginal-score estimator]\label{card:lemma:score}
There is a quantum circuit using $O(\nu^{-2}\log(2b))$ value queries that outputs rational estimates $\widehat s_e$ satisfying
\begin{equation}
 \Prb\left[\max_{e\in R}|\widehat s_e-s_e(R)|\le\nu/8\right]\ge 7/8.
 \label{card:sc:guarantee}
\end{equation}
The circuit can retain all random seeds, outcome registers, and work registers. Its inverse has the same query cost. The statement is uniform over the fixed pool $R$ and the current residual function $f_S$.
\end{lemma}
\begin{proof}
We prove the ideal expectation identity behind Algorithm~\ref{card:alg:score}, then bound each finite-implementation bias, and finally prove simultaneous concentration.

\paragraph{A Positive Fourier Distribution.}
On $[-\Delta,\Delta]$, recall the kernel
\[
 \chi(x)=\left(1-\frac{|x|}{\tau}\right)_+,
\]
and extend $\chi$ periodically with period $2\Delta$. Its Fourier coefficients are
\begin{align}
 c_m&=\frac{1}{2\Delta}\int_{-\tau}^{\tau}
       \left(1-\frac{|x|}{\tau}\right)e^{-i\pi m x/\Delta}\,dx\notag\\
    &=\frac{\tau}{2\Delta}\operatorname{sinc}^2
       \left(\frac{\pi m\tau}{2\Delta}\right),\qquad m\in\mathbb Z,
 \label{card:sc:fourier}
\end{align}
We now verify the coefficient calculation. The sine integrand is odd and integrates to zero. If $m\ne0$, put $\omega=\pi m/\Delta$ and integrate the remaining even part by parts:
\begin{align*}
 c_m
 &=\frac1\Delta\int_0^\tau(1-x/\tau)\cos(\omega x)\,dx\\
 &=\frac1\Delta\left[\frac{(1-x/\tau)\sin(\omega x)}{\omega}\right]_0^\tau
   +\frac1{\Delta\tau\omega}\int_0^\tau\sin(\omega x)\,dx\\
 &=\frac{1-\cos(\omega\tau)}{\Delta\tau\omega^2}
 =\frac{2\Delta}{\pi^2m^2\tau}\sin^2\left(\frac{\pi m\tau}{2\Delta}\right).
\end{align*}
The boundary term vanishes at both endpoints. For $m=0$, direct integration gives
$c_0=\Delta^{-1}[x-x^2/(2\tau)]_0^\tau=\tau/(2\Delta)$, agreeing with~\eqref{card:sc:fourier}.
Thus $c_m\ge0$, $c_{-m}=c_m$, and
\[
 \sum_{m\ne0}|c_m|
 \le\frac{4\Delta}{\pi^2\tau}\sum_{m=1}^{\infty}\frac1{m^2}<\infty.
\]
The convergence of the last series follows, for example, from
$\sum_{m=1}^{\infty}m^{-2}\le1+\int_1^\infty x^{-2}\,dx=2$.
The periodic function $\chi$ is continuous and piecewise linear, so Lemma~\ref{card:tool:fourier-convergence} applies. Absolute convergence allows us to combine the terms at $m$ and $-m$, giving
\[
 \chi(x)=\sum_{m\in\mathbb Z}c_me^{i\pi m x/\Delta}
     =\sum_{m\in\mathbb Z}c_m\cos(\pi m x/\Delta).
\]
At $x=0$, every exponential is one, so $\sum_m c_m=\chi(0)=1$. Consequently $(c_m)_{m\in\mathbb Z}$ is a probability distribution, and
\begin{equation}
 \E_{m\sim c}\cos(\pi m x/\Delta)=\chi(x).
 \label{card:sc:kernel-expectation}
\end{equation}

\paragraph{Two Queries Produce One Joint Sample.}
Independently choose $m$ according to $(c_m)$ and $\theta$ uniformly from $[0,\pi/2]$. Prepare
\[
 |\psi_\theta\rangle=\bigotimes_{e\in R}
        (\cos\theta\,|0\rangle+\sin\theta\,|1\rangle).
\]
The bit string encodes a subset $A\subseteq R$. Apply the diagonal phase
\[
 U_m|A\rangle=e^{i\pi m f_S(A)/\Delta}|A\rangle.
\]
To implement it, query $f(S\cup A)$ and reversibly compute
\[
 f_S(A)=\min\{f(S\cup A)-f(\varnothing),M\}-G_M(S)
\]
using the cached values $f(\varnothing),M,G_M(S)$. Apply the controlled phase and reverse the arithmetic and value query, returning all value work registers to zero. This uses one forward and one inverse value query. The ideal phase is used only to derive the identity; its finite approximation is specified below. Measure every pool qubit in the Pauli-$X$ basis and write $Z_e\in\{-1,1\}$ for its outcome. Since the measured operators act on different qubits, they commute and all outcomes can be recorded in the same experiment.

Fix $e\in R$, and display its qubit separately. For $T\subseteq R\setminus\{e\}$, write
\[
 a_T=(\sin\theta)^{|T|}(\cos\theta)^{b-1-|T|},
 \qquad \phi_A=\frac{\pi m f_S(A)}\Delta.
\]
Then the state after the phase operation is exactly
\[
 U_m|\psi_\theta\rangle
 =\sum_{T\subseteq R\setminus\{e\}}a_T
 \left(\cos\theta\,e^{i\phi_T}|0\rangle_e
       +\sin\theta\,e^{i\phi_{T\cup\{e\}}}|1\rangle_e\right)|T\rangle.
\]
For these two components, the phase difference is
\[
 \phi_{T\cup\{e\}}-\phi_T
 =\frac{\pi m}{\Delta}\bigl(f_S(T\cup\{e\})-f_S(T)\bigr)
 =\frac{\pi m f_S(e\mid T)}{\Delta}.
\]
Put $p=\sin^2\theta$. The squared amplitude on the other qubits is
\[
 a_T^2=p^{|T|}(1-p)^{b-1-|T|}
 =\Prb_{T'\sim\mu_{p\one}}[T'=T],
\]
where each element of $R\setminus\{e\}$ belongs to $T'$ independently with probability $p$; the endpoints $p=0,1$ follow by continuity. The identities $\sigma_x|0\rangle=|1\rangle$ and $\sigma_x|1\rangle=|0\rangle$ show that $\sigma_x^{(e)}$ exchanges the two components with the same $T$. Orthogonality, $\langle T'|T\rangle=0$ for $T'\ne T$, eliminates all other cross terms. The two surviving terms are complex conjugates, so the measurement expectation is
\begin{align}
 \E[Z_e\mid\theta,m]
 &=\langle\psi_\theta|U_m^\dagger \sigma_x^{(e)}U_m|\psi_\theta\rangle\notag\\
 &=\sum_T a_T^2\cos\theta\sin\theta
   \left(e^{i(\phi_{T\cup\{e\}}-\phi_T)}
         +e^{-i(\phi_{T\cup\{e\}}-\phi_T)}\right)\notag\\
 &=\sin(2\theta)\sum_Ta_T^2
       \cos\left(\frac{\pi m f_S(e\mid T)}{\Delta}\right)\notag\\
 &=\sin(2\theta)\E_{T\sim\mu_{p\one}}
       \cos\left(\frac{\pi m f_S(e\mid T)}{\Delta}\right).
 \label{card:sc:paired}
\end{align}
Here the first equality uses the eigenvalues $+1,-1$ of the Pauli-$X$ measurement, the third uses the phase difference computed above, and the last uses the displayed Bernoulli probability.

\paragraph{Averaging the Frequency and Angle.}
For a fixed $\theta$, first average~\eqref{card:sc:paired} over $m$. There are only $2^{b-1}$ possible sets $T$, and the cosine is bounded, so the finite $T$-sum can be interchanged with the frequency average. Using~\eqref{card:sc:kernel-expectation} at $x=f_S(e\mid T)$ gives
\begin{align*}
 \E[Z_e\mid\theta]
 &=\sin(2\theta)\sum_Tp^{|T|}(1-p)^{b-1-|T|}
       \sum_{m\in\mathbb Z}c_m\cos\left(\frac{\pi m f_S(e\mid T)}\Delta\right)\\
 &=\sin(2\theta)\sum_Tp^{|T|}(1-p)^{b-1-|T|}\chi(f_S(e\mid T))\\
 &=\sin(2\theta)\E_{T\sim\mu_{p\one}}\chi(f_S(e\mid T)).
\end{align*}
A uniform angle on an interval of length $\pi/2$ has density $2/\pi$. The law of total expectation therefore says
\[
 \E Z_e=\frac2\pi\int_0^{\pi/2}\sin(2\theta)
       \E_{T\sim\mu_{(\sin^2\theta)\one}}\chi(f_S(e\mid T))\,d\theta.
\]
Multiplying by $\pi/2$ cancels that density. Under the increasing substitution $p=\sin^2\theta$, the endpoints become $0$ and $1$, while
\[
 \frac{dp}{d\theta}=2\sin\theta\cos\theta=\sin(2\theta).
\]
Thus the Jacobian absorbs exactly the prefactor in the integrand:
\begin{align}
 \frac{\pi}{2}\E Z_e
 &=\int_0^{\pi/2}\sin(2\theta)
       \E_{T\sim\mu_{(\sin^2\theta)\one}}\chi(f_S(e\mid T))\,d\theta\notag\\
 &=\int_0^1\E_{T\sim\mu_{p\one}}\chi(f_S(e\mid T))\,dp.
 \label{card:sc:mixture}
\end{align}
In particular, $\theta$ itself is uniform, whereas $p=\sin^2\theta$ is not uniform under that choice alone. The uniform $dp$ integral arises because the measured expectation contributes the extra factor $\sin(2\theta)$.

\paragraph{Identifying the Predecessor Distribution.}
Expand the finite expectation in the last integral of~\eqref{card:sc:mixture} and interchange the finite sum and integral:
\[
 \frac\pi2\E Z_e
 =\sum_{T\subseteq R\setminus\{e\}}\chi(f_S(e\mid T))
       \int_0^1p^{|T|}(1-p)^{b-1-|T|}\,dp.
\]
For nonnegative integers $t,s$, let $I_{t,s}=\int_0^1p^t(1-p)^s\,dp$. If $t\ge1$, integration by parts with $u=p^t$ and $dv=(1-p)^sdp$ yields
\[
 I_{t,s}
 =\left[-\frac{p^t(1-p)^{s+1}}{s+1}\right]_0^1
   +\frac{t}{s+1}I_{t-1,s+1}
 =\frac{t}{s+1}I_{t-1,s+1}.
\]
The boundary term vanishes because $t\ge1$ and $s+1\ge1$.
Also $I_{0,s}=1/(s+1)$. Applying the recurrence $t$ times gives
\[
 I_{t,s}
 =\frac{t!}{(s+1)(s+2)\cdots(s+t)}\,I_{0,s+t}
 =\frac{t!s!}{(s+t+1)!},
\]
with the same final formula for $t=0$. Taking $s=b-1-t$ therefore gives
\[
 \int_0^1p^t(1-p)^{b-1-t}\,dp
 =\frac{t!(b-1-t)!}{b!}
 =\frac1{b\binom{b-1}{t}}.
\]
Now fix a particular $t$-element set $T$. A permutation has predecessor set $R_{<e}^{\pi}=T$ exactly when its first $t$ positions contain the elements of $T$, its next position contains $e$, and its remaining positions contain the other elements. There are $t!(b-1-t)!$ such permutations out of $b!$. Consequently
\[
 \frac\pi2\E Z_e
 =\sum_T\Prb_\pi[R_{<e}^{\pi}=T]\chi(f_S(e\mid T))
 =\E_\pi \chi(f_S(e\mid R_{<e}^{\pi})).
\]
Monotonicity and $f_S\le\Delta$ imply
$0\le f_S(e\mid T)=f_S(T\cup\{e\})-f_S(T)\le\Delta$.
For a number $x$ in this interval, the original branch of the periodic kernel applies, and
\[
 1-\chi(x)=
 \begin{cases}
 x/\tau,&0\le x\le\tau,\\
 1,&\tau\le x\le\Delta,
 \end{cases}
 =\min\{x/\tau,1\}.
\]
It follows, one equality at a time, that
\begin{align*}
 \E[1-(\pi/2)Z_e]
 &=1-\frac\pi2\E Z_e\\
 &=1-\E_\pi \chi(f_S(e\mid R_{<e}^{\pi}))\\
 &=\E_\pi[1-\chi(f_S(e\mid R_{<e}^{\pi}))]\\
 &=\E_\pi\min\{f_S(e\mid R_{<e}^{\pi})/\tau,1\}
 =s_e(R).
\end{align*}
We have proved the identity
\[
 \E\left[1-\frac{\pi}{2}Z_e\right]=s_e(R)
 \quad\text{for every }e\in R.
\]
All coordinates arise from the same two-query phase experiment. The estimator takes the two values $1-\pi/2$ and $1+\pi/2$; its constant range length, rather than membership in $[0,1]$, is what concentration requires.

\paragraph{Finite Frequencies.}
For an integer $L\ge1$, the coefficient bound above and the decreasing function $x^{-2}$ imply
\[
 \eta_L:=\sum_{|m|>L}c_m
 \le\frac{4\Delta}{\pi^2\tau}\sum_{m=L+1}^{\infty}\frac{1}{m^2}
 \le\frac{4\Delta}{\pi^2\tau}\int_L^\infty\frac{dx}{x^2}
 =\frac{4\Delta}{\pi^2\tau L}.
\]
Retain the frequencies $|m|\le L$ and move all tail mass to $m=0$: the new probability at zero is $c_0+\eta_L$, and the other retained probabilities remain $c_m$. These probabilities are nonnegative and sum to one. The retained distribution is computable by the finite formula $\eta_L=1-\sum_{|m|\le L}c_m$, so no infinite sum has to be evaluated by the circuit. The new kernel $\chi_L$ satisfies
\[
 \chi_L(x)-\chi(x)
 =\sum_{|m|>L}c_m\bigl[1-\cos(\pi m x/\Delta)\bigr],
 \qquad |\chi_L(x)-\chi(x)|\le2\eta_L.
\]
For fixed $\theta$, the formula preceding~\eqref{card:sc:mixture} shows that the change in $\E[Z_e\mid\theta]$ is at most
$2\eta_L\sin(2\theta)\le2\eta_L$. Averaging over $\theta$ and multiplying by $\pi/2$ changes the score expectation by at most $\pi\eta_L$. Choose
\[
 L=\left\lceil\frac{256\Delta}{\tau\nu}\right\rceil
 =\left\lceil\frac{256k}{(1-\nu)\nu}\right\rceil.
\]
Then $\pi\eta_L\le\nu/(64\pi)\le\nu/64$. Since $\nu\le1/64$, this choice is $O(k/\nu)$ and depends only on the public parameters $k,\nu$.

\paragraph{Finite Angles.}
Write $|v_\theta\rangle=\cos\theta|0\rangle+\sin\theta|1\rangle$.
Its derivative is $|v'_\theta\rangle=-\sin\theta|0\rangle+\cos\theta|1\rangle$, so
\[
 \|v_\theta\|=\|v'_\theta\|=1,
 \qquad \langle v_\theta|v'_\theta\rangle=0.
\]
Differentiating the tensor product gives $b$ terms, each with one factor $v'_\theta$ and all other factors $v_\theta$. Each term has norm one. Two different terms have zero inner product because at a coordinate differentiated in only one of them the inner product is $\langle v_\theta|v'_\theta\rangle$ or its conjugate. Thus
\[
 \|\partial_\theta\psi_\theta\|^2=b.
\]
For fixed $m$, put $A_m=U_m^\dagger \sigma_x^{(e)}U_m$. Unitarity gives $\|A_m\|=1$, and $A_m$ does not depend on $\theta$. Differentiating and applying Cauchy--Schwarz to both terms yields
\begin{align*}
 \left|\frac{d}{d\theta}\langle\psi_\theta|A_m|\psi_\theta\rangle\right|
 &=\left|\langle\partial_\theta\psi_\theta|A_m|\psi_\theta\rangle
       +\langle\psi_\theta|A_m|\partial_\theta\psi_\theta\rangle\right|\\
 &\le2\|\partial_\theta\psi_\theta\|\,\|A_m\|\,\|\psi_\theta\|
 =2\sqrt b.
\end{align*}
The finitely supported frequency average has the same derivative bound by the triangle inequality, since its weights sum to one.
Replace the uniform angle by a uniform choice among the midpoints
$\theta_j=(j+1/2)\pi/(2K)$, $0\le j<K$.
Couple a continuous uniform angle to the midpoint of its own subinterval. This midpoint is uniform on the $K$ choices, and the angle changes by at most $\pi/(4K)$. The mean value bound above gives an error at most
$2\sqrt b\cdot\pi/(4K)=\pi\sqrt b/(2K)$ in $\E Z_e$.
Multiplying by $\pi/2$, the score error is at most $\pi^2\sqrt b/(4K)$.
Choose an integer $K\ge16\pi^2\sqrt b/\nu$ of order $\sqrt b/\nu$; the error is then at most $\nu/64$.

\paragraph{Rational Scores and Finite Gates.}
Choose a rational $a_\nu$ satisfying $|a_\nu-\pi/2|\le\nu/64$ and $0<a_\nu<2$, and use the sample value $Y_e=1-a_\nu Z_e$.
Because $|\E Z_e|\le1$, this coefficient replacement changes the score expectation by at most
$|a_\nu-\pi/2|\,|\E Z_e|\le\nu/64$.
All computed sample averages and later score comparisons are now rational.

The remaining ideal experiment has finitely many seed values and finite registers. We next justify a uniformly accurate finite implementation and its effect on all coordinates. A state with prescribed computable amplitudes on a finite basis can be prepared by a finite binary tree of controlled single-qubit rotations: at each node rotate with squared amplitudes equal to the probabilities of its two descendant branches. Zero-probability branches can be omitted. This constructs both the finite frequency seed and the angle seed. The required probabilities and angles are computable from the public parameters. The ratio $f_S(A)/\Delta$ is a rational number computed from the oracle encoding, and lies in $[0,1]$; with $|m|\le L$, the desired phase magnitude is at most $\pi L$. Finite arithmetic therefore computes all rotations and diagonal phases to any prescribed uniform accuracy, for every allowed oracle value.

Here is an explicit error budget for these operations. Fix a finite ideal decomposition with $J$ non-oracle gates. Keep all classical reversible gates and the value oracle and its inverse exact. For each quantum rotation requiring synthesis, use Lemma~\ref{card:tool:gates} to obtain operator-norm error at most $\xi/J$; exact gate positions have zero error. The identity
\[
 V_J\cdots V_1-\widetilde V_J\cdots\widetilde V_1
 =\sum_{j=1}^J V_J\cdots V_{j+1}
 (V_j-\widetilde V_j)\widetilde V_{j-1}\cdots\widetilde V_1
\]
shows that the total operator error is at most $\xi$, since all surrounding factors are unitary. Oracle calls interspersed between these gates have norm one and do not change the bound. The phase error is uniformly bounded because $|e^{ix}-e^{iy}|\le|x-y|$.
For normalized final states with amplitudes $a_z,\widetilde a_z$ in the measured basis, their outcome distributions have total variation distance at most
\begin{align*}
 \frac12\sum_z\bigl||a_z|^2-|\widetilde a_z|^2\bigr|
 &\le\frac12\sum_z|a_z-\widetilde a_z|(|a_z|+|\widetilde a_z|)\\
 &\le\frac12\|a-\widetilde a\|_2
           \bigl(\|a\|_2+\|\widetilde a\|_2\bigr)
 \le\xi.
\end{align*}
Summing out seed or work registers cannot increase this distance, by the triangle inequality. Thus $\xi=\nu/256$ suffices for the actual and ideal joint outcome distributions to have total variation distance at most $\nu/256$.
No query is added by this precision increase. We only require a finite description and a finite inverse circuit, not a bound on their number of non-oracle gates.

For distributions $P,Q$ at total variation distance at most $\xi$, a $\{-1,1\}$-valued variable satisfies
$|\E_PZ_e-\E_QZ_e|\le\sum_z|P(z)-Q(z)|=2\xi$.
Since $a_\nu<2$, the resulting change in $\E Y_e$ is at most $4\xi=\nu/64$.
The frequency truncation, angle discretization, rational coefficient, and circuit approximation each contribute at most $\nu/64$. Applying the triangle inequality across these four successive experiments gives
\[
 |\E Y_e-s_e(R)|\le4(\nu/64)=\nu/16.
\]

\paragraph{Concentration of Independent Joint Samples.}
Prepare $D$ independent copies of the actual finite experiment, with independent seed registers, and put
\[
 \widehat s_e=\frac1D\sum_{j=1}^D Y_e^{(j)}.
\]
For each fixed $e$, these $D$ variables are independent and lie in
$[1-a_\nu,1+a_\nu]$, an interval of length $2a_\nu<4$.
Hoeffding's inequality (Lemma~\ref{card:tool:hoeffding}) at deviation $\nu/16$ gives
\begin{align*}
 \Prb\left[|\widehat s_e-\E Y_e|>\nu/16\right]
 &\le2\exp\left(-\frac{2D(\nu/16)^2}{(2a_\nu)^2}\right)\\
 &\le2\exp(-D\nu^2/2048).
\end{align*}
Take $D=\lceil2048\nu^{-2}\log(16b)\rceil$. A union bound makes the probability of any coordinate violating the deviation bound at most
\[
 \sum_{e\in R}2e^{-D\nu^2/2048}
 \le2b\,e^{-\log(16b)}=1/8.
\]
On the complementary event, for every $e$,
\[
 |\widehat s_e-s_e(R)|
 \le|\widehat s_e-\E Y_e|+|\E Y_e-s_e(R)|
 \le\nu/16+\nu/16=\nu/8.
\]
This proves~\eqref{card:sc:guarantee}. Coordinates within one joint sample may be correlated; independence is needed only between sample copies. There are exactly $2D=O(\nu^{-2}\log(2b))$ value queries.

Finally, keep each random choice in a seed register. An $X$-basis measurement can be replaced by a Hadamard followed by copying the resulting computational-basis bit to a fresh outcome register; keeping that register gives the same eventual measurement probabilities. Compute the rational sample averages reversibly from these recorded bits and retain the arithmetic work registers. This produces one fixed finite unitary source circuit whose final measurement has exactly the distribution analyzed above. Include the inverse of every chosen elementary gate among the allowed gates, and reverse the complete circuit when its inverse is needed. Reversal replaces each value query by its inverse and vice versa, so the inverse also uses $2D$ queries. Subsequent amplification therefore uses the actual finite circuit, with all estimation and inverse-query costs included.
\end{proof}

\section{Random batch construction and gain analysis}\label{sec:card-batch}\label{card:sec:batch}
\subsection{Random batch construction algorithm}
The score vector is used to filter a random pool. A score failure is allowed to
change which elements survive, but an exact value test certifies every returned
batch. A second test also permits returning the entire pool when its total
gain is already large. This alternative is needed in the expected-gain
analysis: it compensates for comparison credit lost because pool elements
overlap in value.

Algorithm~\ref{card:alg:batch} constructs the reversible source $A_S$ for the current residual objective.
\begin{pseudocode}{Certified Batch Source $A_S$}\label{card:alg:batch}
\textbf{Input:} Current set $S$; residual function $f_S$ and gap $\Delta>0$ from~\eqref{card:eq:residual}; parameters $k,\nu,b$ from~\eqref{card:eq:params}.

\textbf{Output:} A batch $B\subseteq N\setminus S$ and its normalized gain $V=f_S(B)/\Delta$.
\begin{enumerate}
\item Sample a uniform $b$-set $R$. Use Lemma~\ref{card:lemma:score} to obtain estimates $\widehat s_e$, and put $B_0=\{e\in R:\widehat s_e\ge1-\nu/2\}$.
\item Evaluate the true residual gain. Set $B_f=B_0$ if $f_S(B_0)\ge(1-\nu)^2\Delta|B_0|/k$, and set $B_f=\varnothing$ otherwise.
\item If $f_S(R)\ge(1-\nu)^2\Delta b/k$, return $B=R\setminus S$; otherwise return $B=B_f\setminus S$. Retain the exact number $V=f_S(B)/\Delta$ and all records needed for reversal.
\end{enumerate}
\end{pseudocode}

The background value $G_M(S)$ is cached. A new residual value is obtained by querying $f(S\cup A)$, subtracting the cached $f(\varnothing)$, applying the cap $M$, and subtracting $G_M(S)$. Reversing this computation uses the inverse query. Thus each residual evaluation and its reversal use a constant number of raw value queries. The score circuit and the two density tests give a public cost bound
\[d_0=O(\nu^{-2}\log(2b))\]
for the complete source, including the inverse and exact verification up to a fixed constant factor.

\subsection{Expected gain analysis}
\begin{lemma}[Pointwise Certificate]\label{card:lemma:certificate}
Every output of the actual source, including inaccurate-score outcomes, satisfies
\[
|B|\le b,\qquad f_S(B)\ge(1-\nu)^2\Delta|B|/k,\qquad V\in\{0\}\cup[u,1],\quad u=(1-\nu)^2/k.
\]
On the event $\mathcal E_R=\{\max_{e\in R}|\widehat s_e-s_e|\le\nu/8\}$, the batch $B_0$ passes its density test. Moreover, every unselected $e\in R\setminus B_0$ has $s_e<1$.
\end{lemma}
\begin{proof}
The source chooses a set $A$ equal to either a certified pool $R$, a certified $B_0$, or the empty set. Its returned set is $B=A\setminus S$. Since $S\cup B=S\cup A$, the residual definition gives $f_S(B)=f_S(A)$. Also $|B|\le|A|\le b$, so
\[f_S(B)=f_S(A)\ge(1-\nu)^2\Delta|A|/k\ge(1-\nu)^2\Delta|B|/k.\]
If $V>0$, then $B\ne\varnothing$ because $f_S(\varnothing)=0$, and hence $V\ge(1-\nu)^2|B|/k\ge(1-\nu)^2/k$. The bound $f_S(B)\le\Delta$ gives $V\le1$.

On $\mathcal E_R$, a selected element satisfies
\[s_e\ge\widehat s_e-\nu/8\ge1-\nu/2-\nu/8=1-5\nu/8\ge1-\nu.\]
Lemma~\ref{card:lemma:density} therefore gives
\[f_S(B_0)\ge\tau\sum_{e\in B_0}s_e\ge(1-\nu)\tau|B_0|=(1-\nu)^2\Delta|B_0|/k.\]
Thus $B_f=B_0$ on this event. An unselected element satisfies $\widehat s_e<1-\nu/2$, and consequently
\[s_e\le\widehat s_e+\nu/8<1-3\nu/8<1.\]
\end{proof}

\begin{lemma}[Mean Gain of the Actual Source]\label{card:lemma:mean}
For every current $S$ and every correct target $0<M\le\OPT_k(g)$,
\begin{equation}\label{card:eq:mean}\E V\ge\mu,\qquad \mu=\frac{\nu b}{2n}=\Theta(\nu^2k/n).\end{equation}
The expectation includes all pool choices, score samples, and score failures.
\end{lemma}
\begin{proof}
Fix an order $e_1,\ldots,e_s$ on the comparison set $O$ in~\eqref{card:eq:comparator}, where $s=|O|\le k$. Write $O_{<e_i}=\{e_1,\ldots,e_{i-1}\}$ and assign credits
\[c_e=f_S(e\mid O_{<e}),\qquad w_e=(c_e-\tau)_+,\qquad W=\sum_{e\in O}w_e.\]
The credit sum is
\[\sum_{i=1}^s c_{e_i}=\sum_{i=1}^s\bigl(f_S(\{e_1,\ldots,e_i\})-f_S(\{e_1,\ldots,e_{i-1}\})\bigr)=f_S(O)=\Delta.\]
Since $(x)_+\ge x$, $|O|\le k$, and $k\tau=(1-\nu)\Delta$,
\[W\ge\sum_{e\in O}(c_e-\tau)=\Delta-|O|\tau\ge\Delta-k\tau=\nu\Delta.\]
For a pool $R$, define its excess credit and overlap loss by
\begin{align*}
X_R&=\sum_{e\in O\cap R}w_e,\\
L_R&=\sum_{e\in O\cap R}\left[c_e-f_S(e\mid O_{<e}\cup(R\setminus\{e\}))\right].
\end{align*}
Each summand of $L_R$ is nonnegative: enlarging the base $O_{<e}$ can only decrease the marginal of $e$.

\paragraph{Step 1: Account for Selected and Unselected Credit.}
For every $e\in R$, a random-order predecessor set is contained in $R\setminus\{e\}$. Diminishing returns and monotonicity of $x\mapsto\min\{x/\tau,1\}$ imply, for every order and then after averaging,
\[s_e\ge\min\{f_S(e\mid R\setminus\{e\})/\tau,1\}.\]
On $\mathcal E_R$, an unselected $e$ has $s_e<1$, so $f_S(e\mid R\setminus\{e\})<\tau$. If also $w_e>0$, then
\begin{align*}
w_e=c_e-\tau
&\le c_e-f_S(e\mid R\setminus\{e\})\\
&\le c_e-f_S(e\mid O_{<e}\cup(R\setminus\{e\})).
\end{align*}
If $w_e=0$, the same final upper bound holds because the loss is nonnegative. Thus every unselected comparison element's excess credit is covered by its loss.

For selected comparison elements, list $B_0\cap O$ in the order inherited from $O$. Its predecessor set before $e$ is $(B_0\cap O)\cap O_{<e}\subseteq O_{<e}$. Hence
\begin{align*}
f_S(B_0)&\ge f_S(B_0\cap O)\\
&=\sum_{e\in B_0\cap O}f_S(e\mid(B_0\cap O)\cap O_{<e})\\
&\ge\sum_{e\in B_0\cap O}c_e\ge\sum_{e\in B_0\cap O}w_e.
\end{align*}
Splitting $O\cap R$ into selected and unselected elements now gives, on $\mathcal E_R$,
\begin{align*}
X_R
&=\sum_{e\in O\cap B_0}w_e+\sum_{e\in O\cap(R\setminus B_0)}w_e\\
&\le f_S(B_0)+\sum_{e\in O\cap(R\setminus B_0)}
\left[c_e-f_S(e\mid O_{<e}\cup(R\setminus\{e\}))\right]\\
&\le f_S(B_f)+L_R.
\end{align*}
Here $B_f=B_0$ on $\mathcal E_R$; the last inequality also uses nonnegativity of the omitted selected losses. Consequently, on all outcomes,
\begin{equation}\label{card:eq:creditloss}
f_S(B_f)\ge X_R-L_R-X_R\one_{\mathcal E_R^c}.
\end{equation}
Indeed, on $\mathcal E_R$ this is the inequality just proved, and on its complement the right side is $-L_R\le0\le f_S(B_f)$.

\paragraph{Step 2: Compute the Average Overlap Loss.}
Put $t=b-1$ and $\overline v_j=\E_{|T|=j}f_S(T)$, where the expectation is uniform over $j$-subsets of $N$. Every element belongs to a uniform $b$-set with probability $b/n$, so linearity of expectation gives
\[\E X_R=\sum_{e\in O}w_e\Prb[e\in R]=(b/n)W.\]
Furthermore $\overline v_t\le \overline v_b$: sample a uniform $b$-set and delete one uniformly selected element; the remaining $t$-set is uniform and has no larger value.

Conditional on $e\in R$, the set $R\setminus\{e\}$ is uniform among $t$-subsets avoiding $e$. To relate this law to unconditional sampling, let $T$ be a uniform $t$-set. Since $\Prb[e\in T]=t/n$ and $f_S(e\mid O_{<e}\cup T)=0$ whenever $e\in T$, the law of total expectation gives
\begin{align*}
\E f_S(e\mid O_{<e}\cup T)
&=(1-t/n)\E[f_S(e\mid O_{<e}\cup T)\mid e\notin T],\\
\E[f_S(e\mid O_{<e}\cup(R\setminus\{e\}))\mid e\in R]
&=\frac{\E f_S(e\mid O_{<e}\cup T)}{1-t/n}.
\end{align*}
For every fixed $T$, summing along the order of $O$ telescopes:
\begin{align*}
\sum_{i=1}^s f_S(e_i\mid O_{<e_i}\cup T)
&=\sum_{i=1}^s\bigl(f_S(T\cup\{e_1,\ldots,e_i\})-f_S(T\cup\{e_1,\ldots,e_{i-1}\})\bigr)\\
&=f_S(T\cup O)-f_S(T)=\Delta-f_S(T).
\end{align*}
The last equality is~\eqref{card:eq:comparator}. Therefore
\begin{align*}
\E L_R
&=\sum_{e\in O}\Prb[e\in R]\left(c_e-\E[f_S(e\mid O_{<e}\cup(R\setminus\{e\}))\mid e\in R]\right)\\
&=\frac bn\left(\sum_{e\in O}c_e-
\frac{1}{1-t/n}\E\sum_{e\in O}f_S(e\mid O_{<e}\cup T)\right).
\end{align*}
Substituting the two telescoping sums yields
\[
\E L_R=\frac bn\left(\Delta-\frac{\Delta-\overline v_t}{1-t/n}\right)
=\frac bn\frac{\overline v_t-(t/n)\Delta}{1-t/n}
\le\frac{2b}{n}\overline v_b.
\]
For the final bound, the numerator is at most $\overline v_t$, the denominator is at least $1/2$ because $t/n\le b/n\le\nu/32<1/2$, and $\overline v_t\le \overline v_b$.

Lemma~\ref{card:lemma:score} gives $\Prb[\mathcal E_R^c\mid R]\le1/8$ for every fixed pool. Since $X_R\ge0$ is determined by $R$,
\begin{align*}
\E[X_R\one_{\mathcal E_R^c}]
&=\E_R\bigl[\E[X_R\one_{\mathcal E_R^c}\mid R]\bigr]\\
&=\E_R\bigl[X_R\Prb[\mathcal E_R^c\mid R]\bigr]
\le\tfrac18\E X_R.
\end{align*}
Taking expectations in~\eqref{card:eq:creditloss} and substituting the bounds above gives
\[\E f_S(B_f)\ge\tfrac78\E X_R-\E L_R\ge\tfrac78\frac bn\nu\Delta-\frac{2b}{n}\overline v_b.\]
Writing $r=b/n$ and $\overline v=\overline v_b/\Delta$, and dividing by $\Delta>0$, we obtain
\[\E f_S(B_f)/\Delta\ge r(7\nu/8-2\overline v).\]

\paragraph{Step 3: Combine the Filtered Batch with the Whole Pool.}
Put $Y=f_S(R)/\Delta$ and $\lambda=(1-\nu)^2 b/k$, so $\E Y=\overline v$. If $Y\ge\lambda$, the source returns $R\setminus S$, with gain $f_S(R)\ge f_S(B_f)$ because $B_f\subseteq R$. Otherwise it returns $B_f\setminus S$, with gain exactly $f_S(B_f)$. Thus $V\ge f_S(B_f)/\Delta$ pointwise. Also $V\ge Y\one_{Y\ge\lambda}$, since the gain on the complementary event is nonnegative. Consequently
\begin{align*}
\E V&\ge r(7\nu/8-2\overline v),\\
\E V&\ge\E[Y\one_{Y\ge\lambda}]
=\overline v-\E[Y\one_{Y<\lambda}]
\ge \overline v-\lambda\Prb[Y<\lambda]\ge \overline v-\lambda.
\end{align*}
Both bounds concern the same random output, so
\[\E V\ge\max\{\overline v-\lambda,\ r(7\nu/8-2\overline v)\}.\]
For a direct lower bound independent of $\overline v$, a maximum is at least any convex combination of its two arguments. Choose weights $2r/(1+2r)$ and $1/(1+2r)$. The terms involving $\overline v$ then cancel:
\begin{align*}
\max\{\overline v-\lambda,r(7\nu/8-2\overline v)\}
&\ge\frac{2r(\overline v-\lambda)+r(7\nu/8-2\overline v)}{1+2r}\\
&=\frac{r(7\nu/8-2\lambda)}{1+2r}.
\end{align*}
Finally $\lambda\le b/k\le\nu/32$ and $r\le\nu/32\le1/32$, so
\[\frac{r(7\nu/8-2\lambda)}{1+2r}
\ge\frac{r(7\nu/8-\nu/16)}{1+1/16}
=\frac{13}{17}r\nu\ge\frac12r\nu=\mu.\]
The estimate $\mu=\Theta(\nu^2k/n)$ follows by substituting $\nu k/64\le b\le\nu k/32$ into its definition. This proves~\eqref{card:eq:mean}.
\end{proof}

\section{Fixed-target quantum search and amortized analysis}\label{sec:card-search}\label{card:sec:trial}
\subsection{A fixed-target quantum algorithm}
A straightforward search bound would charge each greedy update the worst-case
cost of finding a useful element. We instead search for a certified batch above
a public gain threshold and charge the search cost to the gain that it returns.
The total normalized gain before acceptance is bounded, so this accounting
controls the entire trial.

The source being amplified is a fixed finite circuit: pool seeds, score outcomes, exact values and work registers are retained. An intermediate measurement in a computational basis can be replaced by copying its outcome into a fresh register and controlling subsequent operations on that register; measuring all retained records at the end gives the same classical outcome probabilities. Thus the complete source has an inverse. After choosing any synthesized implementation of a source, all forward calls use that same circuit and all inverse calls reverse its chosen elementary gates exactly; the inverse is never approximated independently. The marking reflection compares the stored true ratio $f_S(B)/\Delta$ with rational $v$, using exact classical reversible gates and a phase flip. Inaccurate scores can reduce the chance of success, but cannot invalidate an exactly verified density certificate.

Choose rational $\Gamma$ with
\[1-1/e-4\nu\le\Gamma\le1-1/e-3\nu.\]

The acceptance test is an exact rational comparison. The public query cutoff
is enforced before every oracle call, so an incorrect target cannot cause an
unbounded execution.

For a large cardinality budget use $b,d_0,\mu,u$ above; the small-budget replacement is specified in Section~\ref{card:sec:complete}. Set
\[v=\max\{u,\mu/2\},\qquad T=\left\lceil C_1d_0/\sqrt{\mu v}\right\rceil.\]
The absolute constant $C_1$ is chosen below to make truncation failure at most $1/4$.

\begin{pseudocode}{Fixed-Target Trial $\operatorname{Trial}(M)$}\label{card:alg:trial}
\textbf{Input:} Target $M>0$; parameters $b,\Gamma,v,T$ specified above; the source for the current $S$ (Algorithm~\ref{card:alg:batch}, or its singleton replacement in Section~\ref{card:sec:complete}).

\textbf{Output:} Accept with a feasible set $S$ satisfying $g(S)\ge\Gamma M$, or reject.
\begin{enumerate}
\item Initialize $S=\varnothing$ and a total query counter with limit $T$.
\item If $g(S)\ge\Gamma M$, accept and return $S$.
\item If $k-|S|<b$, stop and accept only if the same exact certificate holds.
\item Amplify the source for the current $S$ to search for $V\ge v$. When found, add the entire batch $B$ and return to step 2.
\item \textbf{Global cutoff (applies throughout):} Before any query that would exceed $T$, terminate. Accept only if a stored exact value certifies $g(S)\ge\Gamma M$; otherwise reject.
\end{enumerate}
\end{pseudocode}

Exact values needed for acceptance are cached when verifying a batch. The counter includes every value query, inverse query and verification query. Before a batch is added, step 3 of Algorithm~\ref{card:alg:trial} ensures $|S|+b\le k$; since $|B|\le b$ and $B\cap S=\varnothing$, the new set is feasible on every path.

\subsection{Approximation guarantee and amortized query complexity}
\begin{lemma}[Gain-Amortized Search]\label{card:lemma:amortized}
Suppose a reversible source of cost $d_0$ outputs $V\in\{0\}\cup[u,1]$ and $\E V\ge\mu>0$. Put $v=\max\{u,\mu/2\}$. Search for the exactly verified event $V\ge v$ using Lemma~\ref{card:tool:aa}. If $C$ is the query cost and $V_{\rm out}$ its returned gain, then
\[
\E C\le\frac{C_0d_0}{\sqrt{\mu v}}\E V_{\rm out}
\]
for an absolute constant $C_0$, and the search succeeds almost surely.
\end{lemma}
\begin{proof}
Let $p=\Prb[V\ge v]$ and $m=\E[V\one_{V\ge v}]$. If $v=u$, all nonzero outcomes are marked, so $m=\E V\ge\mu$. If $v=\mu/2>u$, then
\[\E[V\one_{V<v}]\le v\Prb[V<v]\le\mu/2,\qquad
m=\E V-\E[V\one_{V<v}]\ge\mu/2.\]
In either case $m\ge\mu/2>0$. Since $V\le1$, $m\le p$, and therefore $p>0$. On the marked event $V\ge v$, whence
\[m\ge\E[v\one_{V\ge v}]=pv,\qquad p\le m/v.\]
Lemma~\ref{card:tool:aa} gives $\E C\le C'd_0/\sqrt p$ and termination with probability one. Its conditional output law gives
\[\E V_{\rm out}=\E[V\mid V\ge v]
=\frac{\E[V\one_{V\ge v}]}{\Prb[V\ge v]}=\frac mp.\]
All quantities in these denominators are positive. Dividing the cost bound by this exact mean and using $p\le m/v$ gives
\begin{align*}
\frac{\E C}{\E V_{\rm out}}
&\le\frac{C'd_0/\sqrt p}{m/p}
=\frac{C'd_0\sqrt p}{m}\\
&\le\frac{C'd_0\sqrt{m/v}}m
=\frac{C'd_0}{\sqrt{mv}}
\le\frac{\sqrt2 C'd_0}{\sqrt{\mu v}}.
\end{align*}
Taking $C_0=\sqrt2C'$ proves the claim. The argument uses a ratio of expectations and requires no independence between cost and returned gain.
\end{proof}

\begin{lemma}\label{card:lemma:rational-threshold}
A rational $\Gamma$ in the stated interval can be computed without value
queries, and the comparison $g(S)\ge\Gamma M$ is exact in the digital model.
\end{lemma}
\begin{proof}
Such a choice can be made using only public computation. For example, compute a rational $a$ with $|a-(1-1/e)|\le\nu/8$ and put $\Gamma=a-7\nu/2$. Then
\[1-1/e-29\nu/8\le\Gamma\le1-1/e-27\nu/8,\]
which lies inside the required interval because $29/8<4$ and $27/8>3$. To construct $a$ explicitly, set $e_j=\sum_{r=0}^j1/r!$ for an integer $j\ge1$. For every $r\ge0$, the factors after $(j+1)!$ in $(j+1+r)!$ are each at least $j+2$, so
\begin{align*}
0<e-e_j
&=\sum_{r=0}^\infty\frac1{(j+1+r)!}
\le\frac1{(j+1)!}\sum_{r=0}^\infty\frac1{(j+2)^r}\\
&=\frac{j+2}{j+1}\frac1{(j+1)!}
\le\frac2{(j+1)!}.
\end{align*}
Choose $j$ large enough that $2/(j+1)!\le\nu/8$, and set $a=1-1/e_j$. This is rational, and $e_j\ge1$ gives
\[|a-(1-1/e)|=\frac{e-e_j}{ee_j}\le e-e_j\le\nu/8.\]
The choice is made once and uses no oracle queries. The interval also ensures $0<\Gamma<1$, since $\nu<1/64$.

Rationality makes the acceptance test exact in the digital model. If $g(S)=a_1/b_1$, $M=a_2/b_2$, and $\Gamma=a_3/b_3$, with positive denominators, then
\[g(S)\ge\Gamma M\quad\Longleftrightarrow\quad
 a_1b_2b_3\ge a_2a_3b_1.\]
This is a finite integer comparison. The approximation analysis below only uses the stated interval for $\Gamma$.

\end{proof}
\begin{lemma}[Correct Targets Succeed within the Public Budget]\label{card:lemma:trial}
If $M\le\OPT_k(g)$, Trial$(M)$ accepts with probability at least $3/4$. For any target, every accepted output is feasible and has $g(S)\ge\Gamma M$.
\end{lemma}
\begin{proof}
First remove the query cutoff. By Lemma~\ref{card:lemma:mean} for the batch source and Lemma~\ref{card:lemma:singleton} for the singleton source, each reached search has positive marked probability, so Lemma~\ref{card:tool:aa} makes it terminate almost surely. Every returned batch has $V\ge v>0$, hence is nonempty. Feasibility implies that at most $k$ updates occur. A finite sequence of almost-surely terminating searches terminates almost surely.

Write $S_0=\varnothing$ and $S_{i+1}=S_i\cup B_i$. The density certificate at step $i$ gives
\[\Delta_{i+1}=M-G_M(S_{i+1})=\Delta_i-f_{S_i}(B_i)
\le(1-(1-\nu)^2|B_i|/k)\Delta_i.\]
The multiplying factor is nonnegative because $(1-\nu)^2\le1$ and $|B_i|\le k$. Starting with $\Delta_0=M$, repeated substitution yields
\begin{align*}
\Delta_j
&\le M\prod_{i=0}^{j-1}(1-(1-\nu)^2|B_i|/k)\\
&\le M\prod_{i=0}^{j-1}\exp(-(1-\nu)^2|B_i|/k)
=M\exp\left(-\frac{(1-\nu)^2}{k}\sum_{i=0}^{j-1}|B_i|\right)\\
&=M\exp(-(1-\nu)^2|S_j|/k).
\end{align*}
The second line uses $1-x\le e^{-x}$, and the last uses disjointness of the added batches. Since $G_M(S_j)=M-\Delta_j$,
\[G_M(S_j)/M\ge1-\exp(-(1-\nu)^2|S_j|/k).\]
In the batch case a capacity stop means $k-|S_j|<b$, so $|S_j|>k-b$. Also
\begin{align*}
(1-\nu)^2(1-b/k)&\ge(1-2\nu)(1-\nu/32)\\
&=1-65\nu/32+\nu^2/16\ge1-65\nu/32.
\end{align*}
For $x\in[0,1]$, integration gives
\[(1-e^{-1})-(1-e^{-x})=\int_x^1e^{-t}\,dt\le1-x.\]
Apply this to $x=1-65\nu/32>0$, using that $1-e^{-x}$ increases with $x$. It follows that
\[G_M(S_j)/M\ge1-1/e-65\nu/32>1-1/e-3\nu\ge\Gamma.\]
Thus a capacity stop also accepts a correct target; the untruncated trial accepts almost surely. For the singleton replacement, the same argument is completed in Section~\ref{card:sec:complete}.

To bound its expected queries, at the start of every search the failed acceptance test gives $g(S_i)<\Gamma M<M$. Hence $G_M(S_i)=g(S_i)$ and $\Delta_i>(1-\Gamma)M$. Its normalized gain is
\[V_i=\frac{f_{S_i}(B_i)}{\Delta_i}=\frac{G_M(S_{i+1})-G_M(S_i)}{\Delta_i}.\]
If the path has $j$ updates, monotonicity, the lower bound on each denominator, and telescoping give
\[
\sum_{i=0}^{j-1}V_i
\le\frac{\sum_{i=0}^{j-1}(G_M(S_{i+1})-G_M(S_i))}{(1-\Gamma)M}
=\frac{G_M(S_j)-G_M(S_0)}{(1-\Gamma)M}
\le\frac1{1-\Gamma}<e.
\]
Here $G_M(S_0)=0$, $G_M(S_j)\le M$, and $\Gamma<1-1/e$. The final update is included in the sum even if it crosses the acceptance threshold.

Let $\mathcal H_i$ denote the history just before the $i$th search. Set its cost $C_i$ and gain $V_i$ to zero if the trial has already ended, and index all potential searches by $i=0,\ldots,k-1$. Conditional on a reached history, the current residual objective satisfies the hypotheses of Lemma~\ref{card:lemma:amortized} with the same public $\mu,v,d_0$. On an unreached history both sides vanish. Thus
\[\E[C_i\mid\mathcal H_i]\le\frac{C_0d_0}{\sqrt{\mu v}}\E[V_i\mid\mathcal H_i].\]
Taking expectations and using the tower property, then summing the finite set of indices, gives
\begin{align*}
\E\sum_{i=0}^{k-1}C_i
&=\sum_{i=0}^{k-1}\E\bigl[\E[C_i\mid\mathcal H_i]\bigr]\\
&\le\frac{C_0d_0}{\sqrt{\mu v}}\sum_{i=0}^{k-1}\E\bigl[\E[V_i\mid\mathcal H_i]\bigr]\\
&=\frac{C_0d_0}{\sqrt{\mu v}}\E\sum_{i=0}^{k-1}V_i
\le\frac{eC_0d_0}{\sqrt{\mu v}}.
\end{align*}
Exact acceptance checks use cached values; any fixed query overhead in initialization or output verification is absorbed into the cost constant. Thus for an absolute $C_2$,
\[\E Q_{\rm trial}\le C_2d_0/\sqrt{\mu v}.\]
Couple the truncated and untruncated trials using the same randomness. They agree whenever the untruncated query count is at most $T$. Since
\[\E Q_{\rm trial}\ge\E[Q_{\rm trial}\one_{Q_{\rm trial}>T}]\ge T\Prb[Q_{\rm trial}>T],\]
choosing $C_1\ge4C_2$ ensures
\[\Prb[Q_{\rm trial}>T]\le\frac{\E Q_{\rm trial}}T\le\frac{C_2}{C_1}\le\frac14.\]
The truncated trial therefore accepts with probability at least $3/4$. For any target, including $M>\OPT_k(g)$, acceptance is possible only upon an exact certificate, and the query cutoff remains the same deterministic bound.
\end{proof}

\section{The complete cardinality algorithm}\label{sec:card-complete}\label{card:sec:complete}
\subsection{Algorithm description}
For $k=n$ return $N$. Otherwise query $V_N=g(N)$; if $V_N=0$, monotonicity implies $g(A)=0$ for every $A$, so return $\varnothing$. Put $r_{\rm part}=\lceil n/k\rceil\ge2$ and partition $N$ into $r_{\rm part}$ sets $P_1,\ldots,P_{r_{\rm part}}$ of size at most $k$. The resulting interval
$V_N/r_{\rm part}\le\OPT_k(g)\le V_N$ is justified in
Lemma~\ref{card:lemma:target-interval}; it has multiplicative width at most
$r_{\rm part}$.

Construct the rational geometric grid
\[M_j=(V_N/r_{\rm part})(1+\nu)^j,\quad 0\le j\le J,\qquad
J=\left\lceil\frac{\log r_{\rm part}}{\log(1+\nu)}\right\rceil,
\quad H=\lceil\log_2(J+2)\rceil.\]
Equivalently, $J$ is the least integer for which $(1+\nu)^J\ge r_{\rm part}$; this version can be computed with exact rational comparisons. In particular $M_0=V_N/r_{\rm part}$ and $M_J\ge V_N$.
Let $L=\lceil\log(2H/\delta)/\log4\rceil$. A test of $M_j$ runs $L$ independent trials, accepts if any trial accepts, and retains a certified solution. For a correct target, each trial rejects with probability at most $1/4$, so
\[\Prb[\text{all }L\text{ trials reject}]\le(1/4)^L\le\delta/(2H).\]

Start with index endpoints $\ell=-1$, $r=J+1$ and no saved solution. While $r-\ell>1$, test $j=\lfloor(\ell+r)/2\rfloor$. On acceptance set $\ell=j$ and retain the best certified solution obtained; on rejection set $r=j$. If the old width is $w=r-\ell$, the new width is at most $\lceil w/2\rceil$. Starting at $w=J+2$, after $a$ tests it is at most $\lceil(J+2)/2^a\rceil$. Thus at most $H$ tests suffice to reach width one. Return the saved solution, or the empty set if none was saved.

\paragraph{Singleton source for small budgets.}\label{card:sec:small}
If $k<64/\nu$, replace the batch source by a singleton source. Sample uniform $e\in N$ and return $\{e\}\setminus S$ if $f_S(\{e\})\ge(1-\nu)^2\Delta/k$; otherwise return $\varnothing$. This uses $d_0=O(1)$ queries and has the same pointwise certificate with $b=1$ and $u=(1-\nu)^2/k$.

\begin{pseudocode}{Cardinality algorithm with certified target search}\label{card:alg:complete}
\textbf{Input:} $f$, $N$, $k$, accuracy $\eps$, failure probability $\delta$.
\begin{enumerate}
\item If $k=0$, return $\varnothing$; if $k=n$, return $N$. Query $f(\varnothing)$ and $f(N)$,
set $V_N=f(N)-f(\varnothing)$, and return $\varnothing$ if $V_N=0$.
\item Form the geometric targets $M_0,\ldots,M_J$ and set
$L=\lceil\log(2H/\delta)/\log4\rceil$.
Initialize $\ell=-1$, $r=J+1$, and $S_{\rm best}=\varnothing$.
\item While $r-\ell>1$, put $j=\lfloor(\ell+r)/2\rfloor$ and run
$L$ independent copies of $\operatorname{Trial}(M_j)$.
If any accepts, let $\mathcal A$ be its certified outputs and set
\[
 \ell\gets j,\qquad S_{\rm best}\gets
 \mathop{\arg\max}_{A\in\{S_{\rm best}\}\cup\mathcal A} f(A).
\]
Otherwise set $r\gets j$.
\Comment{Values are cached by the exact certificates.}
\item Return $S_{\rm best}$.
\end{enumerate}
\end{pseudocode}

\subsection{Approximation ratio and complexity analysis}
\begin{lemma}\label{card:lemma:target-interval}
If $V_N=g(N)>0$ and $r_{\rm part}=\lceil n/k\rceil$, then
$V_N/r_{\rm part}\le\OPT_k(g)\le V_N$.
\end{lemma}
\begin{proof}
Submodularity and normalization give, for disjoint $A,B$,
\[g(A)+g(B)\ge g(A\cup B)+g(A\cap B)=g(A\cup B).\]
Applying this successively to the partition gives
\[V_N=g(N)\le\sum_{i=1}^{r_{\rm part}}g(P_i)\le r_{\rm part}\OPT_k(g).\]
Also $\OPT_k(g)\le V_N$ by monotonicity, and hence $V_N/r_{\rm part}\le\OPT_k(g)\le V_N$.

\end{proof}
\begin{lemma}[One-Sided Certificates Suffice for Bisection]\label{card:lemma:grid}
With probability at least $1-\delta/2$, the target search returns $g(S)\ge(1-1/e-\eps)\OPT_k(g)$.
\end{lemma}
\begin{proof}
Conditioned on any previous history, an adaptively selected correct target $M_j\le\OPT_k(g)$ is incorrectly rejected with probability at most $\delta/(2H)$, because the new trials use fresh independent randomness. If the target is incorrect, the event ``a correct target is rejected at this test'' has probability zero. Averaging over histories and summing over at most $H$ tests therefore bounds the probability of any such error by $H\cdot\delta/(2H)=\delta/2$.

On the complementary event let $j^*=\max\{j:M_j\le\OPT_k(g)\}$. It exists since $M_0\le\OPT_k(g)$. The initial right sentinel satisfies $r=J+1>j^*$. If a test rejects, the absence of an error implies $M_j>\OPT_k(g)$, so its new right endpoint $j$ is also greater than $j^*$. An acceptance does not change the right endpoint. Thus the invariant $r>j^*$ holds at every iteration. At termination, the endpoints are integers with $r-\ell=1$, so $\ell=r-1\ge j^*\ge0$. The left endpoint must therefore have been set by an acceptance. Since the saved solution is at least as valuable as the certificate obtained at this endpoint,
\[g(S)\ge\Gamma M_\ell\ge\Gamma M_{j^*}.\]
If $j^*<J$, maximality gives $(1+\nu)M_{j^*}=M_{j^*+1}>\OPT_k(g)$, so $M_{j^*}>\OPT_k(g)/(1+\nu)$. If $j^*=J$, then $M_J\le\OPT_k(g)\le V_N\le M_J$, so $M_{j^*}=\OPT_k(g)$ and the same weak lower bound holds. Therefore
\[g(S)\ge\frac{\Gamma}{1+\nu}\OPT_k(g).\]
To evaluate this ratio explicitly, observe that
\begin{align*}
\frac{1-1/e-4\nu}{1+\nu}-(1-1/e-5\nu)
&=\frac{1-1/e-4\nu-(1-1/e-5\nu)(1+\nu)}{1+\nu}\\
&=\frac{\nu/e+5\nu^2}{1+\nu}\ge0.
\end{align*}
Thus
\[\frac{\Gamma}{1+\nu}\ge\frac{1-1/e-4\nu}{1+\nu}\ge1-1/e-5\nu\ge1-1/e-\eps,\]
where the last step uses $\nu\le\eps/16$. The proof never assumes monotonicity of acceptance for targets above $\OPT_k(g)$.
\end{proof}
\begin{lemma}\label{card:lemma:singleton}
For $k<64/\nu$, the singleton source has the same pointwise density certificate
as Algorithm~\ref{card:alg:batch}, has cost $d_0=O(1)$, and satisfies
$\E V\ge\nu/n$ for every correct target. With $\mu=\nu/n$, the trial and
target-search guarantees remain valid.
\end{lemma}
\begin{proof}
For a correct target, use the comparison credits $c_e$ defined above, whose sum is $\Delta$. Let $A=\{e\in O:c_e\ge(1-\nu)^2\Delta/k\}$. There are at most $k$ terms outside $A$, each smaller than $(1-\nu)^2\Delta/k$, so
\[\sum_{e\in O\setminus A}c_e\le(1-\nu)^2\Delta,\qquad
\sum_{e\in A}c_e=\Delta-\sum_{e\in O\setminus A}c_e\ge(1-(1-\nu)^2)\Delta.\]
Since $1-(1-\nu)^2=2\nu-\nu^2\ge\nu$, these accepted credits sum to at least $\nu\Delta$. Diminishing returns gives $f_S(\{e\})=f_S(e\mid\varnothing)\ge f_S(e\mid O_{<e})=c_e$, so every $e\in A$ passes the singleton test. Uniform sampling therefore gives
\[\E V\ge\sum_{e\in A}\frac1n\frac{f_S(\{e\})}{\Delta}
\ge\frac1{n\Delta}\sum_{e\in A}c_e\ge\nu/n.\]
Set $\mu=\nu/n$, $v=\max\{u,\mu/2\}$ and recompute the public trial budget. Because $v\ge u=(1-\nu)^2/k$ and $(1-\nu)^2\ge1/2$,
\[\frac{d_0}{\sqrt{\mu v}}\le\frac{d_0}{\sqrt{(\nu/n)((1-\nu)^2/k)}}=O\!\left(\sqrt{nk/\nu}\right).\]
A capacity stop now means $k-|S|<1$, hence $|S|=k$. The product estimate in the proof of Lemma~\ref{card:lemma:trial} gives $G_M(S)/M\ge1-e^{-(1-\nu)^2}$. Since $(1-\nu)^2\ge1-2\nu$, the same integral bound used there yields
\[G_M(S)/M\ge1-1/e-2\nu>1-1/e-3\nu\ge\Gamma.\]
Thus untruncated correct-target trials also accept in this case. All gain-sum, cutoff and target-search arguments apply with the new parameters.

\end{proof}

\begin{proof}[Proof of Theorem~\ref{card:thm:upper}]
For $k\ge64/\nu$, the bounds $\nu k/64\le b\le\nu k/32$ give
\[\mu=\Theta(\nu^2k/n),\qquad u=\Theta(1/k),\qquad d_0=O(\nu^{-2}\log(2k)).\]
Since $v=\max\{u,\mu/2\}$, reciprocating its square root gives
\[\frac{d_0}{\sqrt{\mu v}}=d_0\min\{(\mu u)^{-1/2},\sqrt2/\mu\}.\]
For the first expression,
\[\mu u=\Theta\!\left(\frac{\nu^2k}{n}\frac1k\right)=\Theta(\nu^2/n),\qquad
\frac{d_0}{\sqrt{\mu u}}=O(\nu^{-3}\sqrt n\log(2k)).\]
For the second,
\[\frac{d_0}{\mu}=O\!\left(\nu^{-2}\log(2k)\frac{n}{\nu^2k}\right)
=O(\nu^{-4}(n/k)\log(2k)).\]
Each trial is cut off at $T=O(1+d_0/\sqrt{\mu v})$ queries. There are at most $H$ tested targets and $L=O(\log(2H/\delta))$ trials per target. Including the constant initialization queries gives the deterministic bound
\begin{equation}\label{card:eq:explicitcost}
Q=O\!\left(1+H\log(2H/\delta)\log(2k)
\min\{\nu^{-3}\sqrt n,\nu^{-4}n/k\}\right).
\end{equation}
To bound $H$, use $0<\nu<1$ and
\[\log(1+\nu)=\int_0^\nu\frac{dt}{1+t}\ge\frac\nu{1+\nu}\ge\nu/2.\]
Hence $J\le1+2\log r_{\rm part}/\nu$, and its definition implies
\[H=O\!\left(\log\left(2+\frac{\log(2r_{\rm part})}\nu\right)\right).\]

For $k<64/\nu$, the singleton cost is bounded by each of the two principal terms. Indeed,
\[\sqrt{nk/\nu}=\sqrt n\sqrt{k/\nu}\le8\nu^{-1}\sqrt n\le8\nu^{-3}\sqrt n,\]
and, since $k\le n$,
\[\frac{\sqrt{nk/\nu}}{\nu^{-4}n/k}
=\nu^{7/2}k\sqrt{k/n}
\le\nu^{7/2}\frac{64}\nu=64\nu^{5/2}\le64.\]
A quantity bounded by a constant times each term is bounded by the larger of these constants times their minimum. Thus~\eqref{card:eq:explicitcost} covers the singleton case as well. Finally $\eps/32\le\nu\le\eps/16$ gives $\nu^{-3}=O(\eps^{-3})$ and $\nu^{-4}=O(\eps^{-4})$. All remaining factors in~\eqref{card:eq:explicitcost} are logarithmic in the parameters claimed in Theorem~\ref{card:thm:upper}.

The finite score construction in Section~\ref{card:sec:score} already includes frequency truncation, angle discretization, rational-score replacement and its own finite-gate bias. Keep that chosen score circuit and its inverse unchanged. First analyze it together with any remaining ideal public quantum preparations, such as uniform pool and search-seed preparations. Given the input encoding width and the query cutoff, defer measurements, retain all random seeds and work records, and use the exact controlled-query construction or harmless padding from the model subsection. This gives one finite unitary computation with a public bound $D\ge1$ on the number of public gates. Oracle calls, classical reversible arithmetic, address formation, certificate comparisons, control and counters remain exact. Apply Lemma~\ref{card:tool:gates} only to the remaining quantum rotations requiring synthesis, with error at most $\eta=\delta/(2D)$ per gate. A source and its inverse use the same synthesized gates in opposite order. Thus this final allowance concerns only operations not already covered by the score circuit's bias bound.

For completeness, include the exact oracle calls in the full sequence of $P$ operations, denoted $U_1,\ldots,U_P$, and write $\widetilde U_j$ for the implemented operation. At an oracle position $\widetilde U_j=U_j$. Inserting one replacement at a time gives the exact identity
\[
\widetilde U_P\cdots\widetilde U_1-U_P\cdots U_1
=\sum_{j=1}^P\widetilde U_P\cdots\widetilde U_{j+1}
(\widetilde U_j-U_j)U_{j-1}\cdots U_1.
\]
Every surrounding factor is unitary and has norm one. The oracle terms and all unchanged exact-gate terms vanish, and at most $D$ synthesized-gate terms each have norm at most $\eta$. Hence the operator norm of the difference is at most $D\eta\le\delta/2$. Applied to the same initial unit vector, the two final state vectors differ in norm by at most $\delta/2$.

If their amplitudes in the final measurement basis are $a_x$ and $b_x$, respectively, then Cauchy--Schwarz gives
\begin{align*}
\frac12\sum_x\bigl||a_x|^2-|b_x|^2\bigr|
&\le\frac12\sum_x|a_x-b_x|(|a_x|+|b_x|)\\
&\le\frac12\left(\sum_x|a_x-b_x|^2\right)^{1/2}
\left(\sum_x(|a_x|+|b_x|)^2\right)^{1/2}\\
&\le\left(\sum_x|a_x-b_x|^2\right)^{1/2}\le\delta/2.
\end{align*}
The penultimate step uses the triangle inequality for Euclidean norm and normalization: $(\sum_x(|a_x|+|b_x|)^2)^{1/2}\le\|a\|_2+\|b\|_2=2$. Thus the total variation distance between the two output distributions, and in particular the change in success probability, is at most $\delta/2$.

Public gate synthesis is uncharged in the value-query model; increasing its precision introduces no new calls to $O_f$ or $O_f^{-1}$. Exact control and certificate checks continue to enforce feasibility and the query cutoff on every path; approximation changes the probabilities of candidate outcomes, not these checks. As a final classical output rule, return the empty set if a measured candidate has size greater than $k$; this requires no query and leaves every successful output unchanged. Combining the $\delta/2$ implementation allowance with the $\delta/2$ target-search failure probability proves success at least $1-\delta$. The normalization argument transfers the approximation guarantee from $g$ to $f$.

\end{proof}

\section{Quantum query lower bounds}\label{sec:lower}
An exponential classical lower bound does not by itself imply an exponential quantum lower bound: coherent queries can assign amplitude to many addresses simultaneously. The hard families used below satisfy a stronger condition. For each fixed query address, the hidden oracle agrees exactly with a public reference oracle except on a small fraction of hidden instances. A single hybrid argument therefore yields both unconstrained and cardinality-constrained lower bounds in the exact digital value-oracle model.

Queries may address arbitrary subsets, including sets that violate the cardinality constraint; only the returned solution is constrained. All query lower bounds in this section concern algorithms with a uniform worst-case query cap. We distinguish a probability-of-success requirement from an expected-approximation requirement explicitly.

\subsection{A common hybrid argument}
The reference computation is independent of the hidden instance, which is the reason an address-wise disagreement estimate remains useful for an adaptive quantum algorithm. We give the needed version of the hybrid method of~\citep*{BBBV} in full.

\begin{lemma}[\citep*{BBBV}]\label{lower:lb:hybrid}
Let $z$ be a random hidden parameter, let $v_z$ be a digital value oracle, and let $v_0$ be a fixed reference oracle. Assume that for every fixed query address $S$,
\[
 \Prb_z[v_z(S)\ne v_0(S)]\le p.
\]
Suppose an event $A_z$ concerning the final output has probability at most $p_{\mathrm{out}}$ on average over $z$ when the algorithm is run on $v_0$. For a quantum algorithm using at most $T$ queries on every execution path, let $\varrho_z$ and $\varrho_0$ be its final states under the hidden and reference oracles. Then
\begin{equation}\label{lower:lb:distance}
 \E_z D(\varrho_z,\varrho_0)\le 2T\sqrt p.
\end{equation}
In particular, the algorithm satisfies
\[
 \E_z\Prb[A_z\mid v_z]\le p_{\mathrm{out}}+2T\sqrt p.
\]
The same state-distance bound controls the change in the expectation of any output payoff in $[0,1]$, even if the payoff depends on $z$. A classical randomized algorithm using at most $T$ queries on every execution path satisfies instead
\[
 \E_z\Prb[A_z\mid v_z]\le p_{\mathrm{out}}+Tp.
\]
\end{lemma}

\begin{proof}
We give the hybrid argument explicitly, following the query-replacement method of Bennett, Bernstein, Brassard, and Vazirani~\citep*{BBBV}. Replace classical randomness by a register containing a coherent copy of the random seed, and replace each intermediate measurement by an isometry that copies its outcome into an additional register. Subsequent operations can be controlled by that register. Keeping all these registers until the final measurement preserves the output distribution and makes the intervening computation unitary. Pad a computation using fewer than $T$ queries with harmless oracle calls whose answer register is uniform over all bit strings: XOR with any encoded value leaves that uniform state invariant.

Let $U_z$ and $U_0$ denote the two oracle unitaries. Write the complete computations as
\[
 |\psi_z\rangle=W_TU_zW_{T-1}U_z\cdots W_1U_zW_0|\mathrm{init}\rangle,
\]
and the same expression with $U_0$ in place of every $U_z$ for $|\psi_0\rangle$. The inter-query unitaries $W_t$ do not depend on the hidden parameter. Immediately before reference query $t$, let
\[
 |\phi_t^0\rangle
 =W_{t-1}U_0W_{t-2}\cdots U_0W_0|\mathrm{init}\rangle
 =\sum_S|S\rangle|a_{t,S}\rangle.
\]
For $t=1$, the expression is $W_0|\mathrm{init}\rangle$. The vectors $|a_{t,S}\rangle$ include all other registers and need not be normalized; they satisfy $\sum_S\|a_{t,S}\|^2=1$. Since this is the reference computation, none of them depends on $z$.

Let $D_z=\{S:v_z(S)\ne v_0(S)\}$ and define
\[
 P_z=\sum_{S\in D_z}|S\rangle\langle S|\otimes I.
\]
Orthogonality of distinct address states gives
\begin{align*}
 \E_z\|P_z|\phi_t^0\rangle\|^2
 &=\E_z\sum_S\mathbf1_{\{S\in D_z\}}\|a_{t,S}\|^2\\
 &=\sum_S\|a_{t,S}\|^2\Prb_z[S\in D_z]\\
 &\le p\sum_S\|a_{t,S}\|^2=p.
\end{align*}
The oracle unitaries agree on the complement of $P_z$, so
\begin{align*}
 \|(U_z-U_0)|\phi_t^0\rangle\|
 &=\|(U_z-U_0)P_z|\phi_t^0\rangle\|\\
 &\le\|U_zP_z|\phi_t^0\rangle\|+\|U_0P_z|\phi_t^0\rangle\|\\
 &=2\|P_z|\phi_t^0\rangle\|.
\end{align*}
The final equality uses preservation of vector norms under unitary maps.

To display the telescoping identity, set
\[
 L_{t,z}=W_TU_zW_{T-1}\cdots U_zW_t,
 \qquad L_{T,z}=W_T.
\]
Replacing oracle calls one at a time, starting with the first call, yields the exact identity
\[
 |\psi_z\rangle-|\psi_0\rangle
 =\sum_{t=1}^T L_{t,z}(U_z-U_0)|\phi_t^0\rangle.
\]
Indeed, the summand is the difference between two computations whose first $t-1$ queries both use $U_0$, whose queries after $t$ both use $U_z$, and which differ only at query $t$. Their adjacent terms cancel in the sum. All $L_{t,z}$ are unitary, and hence
\[
 \|\,|\psi_z\rangle-|\psi_0\rangle\|
 \le2\sum_{t=1}^T\|P_z|\phi_t^0\rangle\|.
\]
Cauchy--Schwarz, applied to the random variables $1$ and $\|P_z|\phi_t^0\rangle\|$, gives
\[
 \bigl(\E_z\|P_z|\phi_t^0\rangle\|\bigr)^2
 \le\E_z1^2\,\E_z\|P_z|\phi_t^0\rangle\|^2\le p.
\]
Averaging the preceding vector inequality therefore gives
\[
 \E_z\|\,|\psi_z\rangle-|\psi_0\rangle\|\le2T\sqrt p.
\]

For completeness, we also justify the passage from vector distance to an output-event probability. Put $a=\langle\psi_0|\psi_z\rangle$. The Hermitian operator
\[
 R_z=|\psi_z\rangle\langle\psi_z|-|\psi_0\rangle\langle\psi_0|
\]
vanishes outside the span of the two vectors. If $|a|=1$, the vectors differ only by a phase, so $R_z=0$ and the desired probability bound is immediate. Otherwise write
\[
 |\psi_z\rangle=a|\psi_0\rangle+\sqrt{1-|a|^2}|w\rangle,
 \qquad \langle\psi_0|w\rangle=0,\quad\|w\|=1.
\]
In the orthonormal basis $\{|\psi_0\rangle,|w\rangle\}$, its matrix is
\[
 R_z=
 \begin{pmatrix}
 |a|^2-1&a\sqrt{1-|a|^2}\\
 \overline a\sqrt{1-|a|^2}&1-|a|^2
 \end{pmatrix}.
\]
Its trace is zero and its determinant equals
\[
 -(1-|a|^2)^2-|a|^2(1-|a|^2)=-(1-|a|^2).
\]
Its characteristic polynomial is therefore $x^2-(1-|a|^2)$, and its eigenvalues are $\pm\sqrt{1-|a|^2}$. An output event is represented by an operator $E_z$ satisfying $0\le E_z\le I$, including when the event itself depends on $z$. In an eigenbasis of $R_z$, both diagonal entries of $E_z$ lie in $[0,1]$, and therefore
\[
 |\Prb[A_z\mid v_z]-\Prb[A_z\mid v_0]|
 =|\operatorname{Tr}(E_zR_z)|\le\sqrt{1-|a|^2}.
\]
Furthermore,
\begin{align*}
 1-|a|^2
 &=(1-|a|)(1+|a|)\le2(1-|a|)\\
 &\le2(1-\operatorname{Re}a)
 =\|\,|\psi_z\rangle-|\psi_0\rangle\|^2.
\end{align*}
The event-probability difference is consequently at most the vector distance. Averaging over $z$ and using the reference-event assumption now yields
\[
 \E_z\Prb[A_z\mid v_z]
 \le\E_z\Prb[A_z\mid v_0]+2T\sqrt p
 \le p_{\mathrm{out}}+2T\sqrt p.
\]

The pure-state trace distance is $\sqrt{1-|a|^2}$, as the preceding eigenvalue calculation shows. Discarding ancillary registers cannot increase trace distance, so averaging the vector-distance bound proves~\eqref{lower:lb:distance}. For a payoff $Y_z$ taking values in $[0,1]$, write $\E Y_z=\int_0^1\Prb[Y_z>t]\,dt$ and apply the event bound to each level set. This proves the payoff assertion.

For the classical claim, condition on the algorithm's random seed and couple the two runs. They have identical transcripts until the first query where their oracle answers differ. All query addresses on the reference transcript are independent of $z$, so the union bound gives probability at most $Tp$ for such a first disagreement. If no disagreement occurs, the output is the reference output. The remaining event $A_z$ has average probability at most $p_{\mathrm{out}}$, proving the bound $p_{\mathrm{out}}+Tp$ after averaging over the random seed.
\end{proof}

\subsection{Unconstrained maximization above one half}\label{lower:sec:usm}
\subsubsection{The hard family and its agreement region}
Fix \(0<\varepsilon\le1/2\), an even \(n\ge4/\varepsilon\), and
\(\tau=\lfloor\varepsilon n/2\rfloor\).
Choose a uniformly random equipartition \(N=C\mathbin{\dot\cup}D\).
For \(S\subseteq N\), write
\[
 s(S)=|S|,\qquad d_C(S)=|S\cap C|-|S\cap D|.
\]
We use the functions of \citep*[Section~4.2]{FeigeMirrokniVondrak2011}:
\begin{equation}\label{lower:eq:hidden-function}
 f_C(S)=s(S)(n-s(S))+
               \bigl(|d_C(S)|-\tau\bigr)_+^2,\qquad
 g_{\mathrm u,0}(S)=s(S)(n-s(S)).
\end{equation}
The first term depends only on the size of a query. The second
reveals the partition only when the query is sufficiently unbalanced
between its two hidden parts. For all other queries, the two value
oracles return the same integer, including every bit of its encoding.

\begin{lemma}\label{lower:lem:hidden-submodular}
Each \(f_C\) is nonnegative, normalized, symmetric under set
complementation, and submodular. Its optimum is independent of \(C\)
and equals
\begin{equation}\label{lower:eq:hidden-optimum}
 M_{\mathrm u}=n^2/4+(n/2-\tau)^2.
\end{equation}
The function \(g_{\mathrm u,0}\) is also nonnegative and submodular. The values of
both functions have exact \(O(\log n)\)-bit integer encodings.
\end{lemma}
\begin{proof}
Nonnegativity is immediate. At the empty set both terms vanish, and
complementation sends \(s\) to \(n-s\) and \(d_C\) to \(-d_C\),
so it preserves the value.
To verify submodularity, extend the count function to real variables
\(x,y\in[0,n/2]\):
\[
 Q(x,y)=(x+y)(n-x-y)+(|x-y|-\tau)_+^2.
\]
On \(|x-y|\le\tau\), all its second partial derivatives are \(-2\).
On \(x-y>\tau\), it is
\[
 Q(x,y)=n(x+y)-4xy-2\tau(x-y)+\tau^2;
\]
the pure second derivatives are zero and the mixed derivative is
\(-4\). On \(y-x>\tau\), the same conclusions hold with the sign
of the linear term reversed. The first derivatives agree at the
region boundaries. It follows that \(Q_x\) and \(Q_y\) are
nonincreasing in either variable throughout the domain.

For example, the discrete marginal of adding an element of \(C\)
to a set with counts \((x,y)\) is
\[
 Q(x+1,y)-Q(x,y)=\int_0^1 Q_x(x+t,y)\,dt.
\]
Increasing either context count can only decrease this integral.
The same reasoning applies to an element of \(D\). These are
exactly the diminishing-returns inequalities for \(f_C\).
For \(g_{\mathrm u,0}\), adding an element to a set of size \(s\) has marginal
\(n-2s-1\), which decreases with the context size.

For every set, \(s(n-s)\le n^2/4\) and
\(|d_C(S)|\le n/2\). Hence \(f_C(S)\le M_{\mathrm u}\).
Both bounds are attained by \(S=C\) and by \(S=D\), proving
\eqref{lower:eq:hidden-optimum}. Finally all values are integers
between zero and \(n^2\), which proves the encoding claim.
\end{proof}

Define the ratio available at an uninformative address by
\[
 b=\frac{n^2/4}{M_{\mathrm u}}
   =\frac1{2-4a+4a^2},\qquad a=\tau/n.
\]
\begin{lemma}\label{lower:lem:hidden-gap}
If \(f_C(S)=g_{\mathrm u,0}(S)\), then \(f_C(S)/M_{\mathrm u}\le b\). Moreover,
\begin{equation}\label{lower:eq:hidden-gap}
 b\le\frac1{2-2\varepsilon+\varepsilon^2}
   \le\frac12+\frac{2\varepsilon}{3}
   <\frac12+\varepsilon.
\end{equation}
\end{lemma}
\begin{proof}
At such an address the numerator is at most \(n^2/4\).
Since \(0\le a\le\varepsilon/2\le1/4\), the denominator
\(2-4a+4a^2\) decreases as \(a\) increases, which gives the first
inequality in \eqref{lower:eq:hidden-gap}.
For the second, multiply by the positive denominator and observe
that
\[
 \left(\frac12+\frac{2\varepsilon}{3}\right)
 (2-2\varepsilon+\varepsilon^2)-1
 =\frac{\varepsilon}{6}(2-5\varepsilon+4\varepsilon^2)>0.
\]
The polynomial in parentheses is positive, since its discriminant
is negative and its leading coefficient is positive.
The last inequality follows from \(\varepsilon>0\).
\end{proof}

\begin{lemma}\label{lower:lem:rare-address}
For every fixed \(S\subseteq N\), independently of the random
partition,
\begin{equation}\label{lower:eq:rare-address}
 \Prb_C[f_C(S)\ne g_{\mathrm u,0}(S)]\le p,\qquad
 p=\min\{1,2e^{-(\tau+1)^2/n}\}
   \le 2e^{-\varepsilon^2n/4}.
\end{equation}
\end{lemma}
\begin{proof}
Let \(s=|S|\). Choosing \(C\) uniformly among the \(n/2\)-subsets
makes \(Z=|S\cap C|\) the sum of a sample of \(n/2\) elements
without replacement from a population containing \(s\) ones and
\(n-s\) zeros. Thus \(\E Z=s/2\), and
\(d_C(S)=2(Z-\E Z)\).
Hoeffding's sampling-without-replacement inequality
\citep*[Theorem~4]{Hoeffding} gives, for \(u>0\),
\[
 \Prb_C[|d_C(S)|\ge u]
 \le 2\exp(-u^2/n).
\]
By \eqref{lower:eq:hidden-function}, disagreement requires
\(|d_C(S)|>\tau\). This is an integer-valued random variable, so
the event requires \(|d_C(S)|\ge\tau+1\).
The claimed bound follows; its last inequality uses
\(\tau+1>\varepsilon n/2\).
\end{proof}

\subsubsection{Success probability and expected approximation}
We give finite-parameter bounds before taking asymptotics.
The success-probability and expected-value formulations require
slightly different final arguments.

\begin{theorem}\label{lower:thm:finite-lower}
Consider the parameters and functions in
\eqref{lower:eq:hidden-function}, with \(p\) as in
\eqref{lower:eq:rare-address}.
If a quantum algorithm makes at most \(T\) queries on every
run and returns a \((1/2+\varepsilon)\)-approximate set with
probability at least \(2/3\) on every instance, then
\[
 T\ge\frac{(2/3-p)_+}{2\sqrt p}.
\]
If instead it guarantees
\(\E_{\mathrm{alg}}f_C(S)\ge(1/2+\varepsilon)M_{\mathrm u}\) on every
instance, then
\[
 T\ge\frac{(\varepsilon/3-p)_+}{2\sqrt p}.
\]
For classical randomized algorithms, the respective bounds are
\[
 T\ge\frac{2}{3p}-1,\qquad
 T\ge\frac{\varepsilon}{3p}-1.
\]
\end{theorem}
\begin{proof}
Run the algorithm against the reference oracle \(g_{\mathrm u,0}\).
Its output distribution is independent of \(C\).
Lemma~\ref{lower:lem:rare-address}, applied to each fixed output set,
therefore implies
\[
 \Prb_{C,\mathrm{alg}}[f_C(S)\ne g_{\mathrm u,0}(S)
                     \text{ under the reference oracle}]
 \le p.
\]
By Lemma~\ref{lower:lem:hidden-gap}, every
\((1/2+\varepsilon)\)-approximate output must belong to this
disagreement event. Lemma~\ref{lower:lb:hybrid} consequently bounds
its actual average success probability by \(p+2T\sqrt p\).
If success is at least \(2/3\) for every \(C\), it is at least
\(2/3\) on average. Rearranging proves the first quantum bound.
The classical event bound in Lemma~\ref{lower:lb:hybrid} gives
\(2/3\le(T+1)p\), proving its classical counterpart.

For the expected-value claim, consider the normalized payoff
\(f_C(S)/M_{\mathrm u}\), which lies in \([0,1]\) by
Lemma~\ref{lower:lem:hidden-submodular}. Under the reference oracle,
it is at most \(b\) on agreement addresses and at most one
elsewhere. Its average is therefore at most \(b+p\).
For each \(C\), the difference in the expectation of any
measurement payoff in \([0,1]\) is at most the trace distance.
This follows, for example, by integrating the probabilities
of the events that the payoff exceeds \(t\), for \(0\le t\le1\).
Thus \eqref{lower:lb:distance} bounds the actual average
normalized payoff by \(b+p+2T\sqrt p\).
The required guarantee and \eqref{lower:eq:hidden-gap} imply
\[
 \varepsilon/3\le p+2T\sqrt p.
\]
Rearranging proves the second quantum bound.
For classical algorithms, the transcript coupling in
Lemma~\ref{lower:lb:hybrid} changes the expected payoff by at
most \(Tp\), giving
\(\varepsilon/3\le(T+1)p\) and the stated bound.
\end{proof}

\begin{corollary}\label{lower:cor:usm-exponential}
Fix $0<\varepsilon\le1/2$ and an even $n$. If $\varepsilon^2n$ is at least a sufficiently large absolute constant, any quantum algorithm that returns a $(1/2+\varepsilon)$-approximation with probability at least $2/3$ on every nonnegative submodular function requires
\[
 \Omega\!\left(e^{\varepsilon^2n/8}\right)
\]
queries in the worst case. For expected $(1/2+\varepsilon)$-approximation, the lower bound is
\[
 \Omega\!\left(\varepsilon e^{\varepsilon^2n/8}\right),
\]
provided $\varepsilon^2n\ge C\log(2/\varepsilon)$ for a sufficiently large universal $C$. The corresponding classical bounds are $\Omega(e^{\varepsilon^2n/4})$ and $\Omega(\varepsilon e^{\varepsilon^2n/4})$ in the respective regimes. The hard functions have exact $O(\log n)$-bit integer encodings.
\end{corollary}
\begin{proof}
By~\eqref{lower:eq:rare-address}, $p\le2e^{-\varepsilon^2n/4}$. The first condition ensures $p\le1/3$, so Theorem~\ref{lower:thm:finite-lower} gives a quantum lower bound of $1/(6\sqrt p)$ and a classical lower bound of $2/(3p)-1$. The second condition ensures $p\le\varepsilon/6$, so the same theorem gives quantum and classical lower bounds of $\varepsilon/(12\sqrt p)$ and $\varepsilon/(3p)-1$, respectively. These estimates imply the stated asymptotics. Taking the constants sufficiently large also ensures $n\ge4/\varepsilon$, as required by the construction. The encoding claim follows from Lemma~\ref{lower:lem:hidden-submodular}.
\end{proof}

For fixed $\varepsilon$, both the classical and quantum query complexities above the one-half threshold are exponential in $n$. Exhaustive search uses $2^n$ classical queries, and quantum maximum finding over all subsets uses $O(2^{n/2})$ queries at constant success probability~\citep*{DurrHoyer1996}. These lower bounds rule out reducing the above-threshold task to polynomial quantum query complexity, while leaving the optimal exponential rates and possible quadratic improvements unresolved.

\subsection{Constant-factor approximation under a cardinality constraint}\label{lower:sec:constant-card}

\begin{theorem}[Quantum Lower Bound for Constant Approximation]\label{lower:thm:cardinality-lower}
Fix $a\in(0,1]$, and let $1\le k\le n/2$. Any quantum algorithm that, for every normalized nonnegative monotone submodular function $f$, outputs a set $S$ with
\[
 |S|\le k,\qquad f(S)\ge a\OPT_k(f)
\]
with probability at least $2/3$ must make
\[
 Q\ge\frac{1}{12\sqrt{\lceil2/a\rceil}}\sqrt{\frac nk}
 =\Omega_a\!\left(\sqrt{\frac nk}\right)
\]
value queries in the worst case. The bound already holds for capped modular functions with exact rational values encoded in $O_a(\log(k+1))$ bits.
\end{theorem}

\begin{proof}
Put $m=\lceil2/a\rceil$ and $\beta=1/(mk)$. For a hidden element $j\in N=[n]$, define
\begin{equation}\label{lower:lb:capped-family}
 f_j(S)=\min\{1,\mathbf1_{\{j\in S\}}+\beta|S\setminus\{j\}|\},
 \qquad
 f_0(S)=\min\{1,\beta|S|\}.
\end{equation}
Here $f_0$ is a public reference function. Since $m\ge2/a$, we have $1/m\le a/2<a$. We verify the properties of this family before applying Lemma~\ref{lower:lb:hybrid}.

For each $j$, assign weight $1$ to $j$ and weight $\beta$ to every other element, and write $W_j(S)$ for the sum of weights in $S$. Then $f_j(S)=\min\{1,W_j(S)\}$. Nonnegative weights imply normalization, nonnegativity, and monotonicity. For $e\notin S$, its marginal is
\[
 f_j(e\mid S)=\min\{w_j(e),(1-W_j(S))_+\}.
\]
Indeed, if $W_j(S)\ge1$, the function is already saturated and the marginal is zero; otherwise the marginal is the smaller of the new weight and the remaining distance to one. If $A\subseteq B$ and $e\notin B$, then $W_j(A)\le W_j(B)$, so $(1-W_j(A))_+\ge(1-W_j(B))_+$. Taking the minimum with the same number $w_j(e)$ preserves this inequality and gives $f_j(e\mid A)\ge f_j(e\mid B)$. This proves submodularity. The same argument with all weights equal to $\beta$ proves every required property of $f_0$.

The singleton $\{j\}$ is feasible and has value one, while no function value exceeds one. Hence $\OPT_k(f_j)=1$. If $|S|\le k$ and $j\notin S$, then
\[
 f_j(S)=\beta|S|\le\beta k=1/m<a.
\]
Therefore every successful $a$-approximation must contain $j$.

Saturation prevents arbitrarily large query sets from revealing the hidden element. Put $K=mk-1$. If $j\notin S$, the two functions in~\eqref{lower:lb:capped-family} are equal. If $j\in S$, then $f_j(S)=1$, whereas
\[
 f_0(S)<1
 \ \Longleftrightarrow\ \beta|S|<1
 \ \Longleftrightarrow\ |S|<mk
 \ \Longleftrightarrow\ |S|\le K.
\]
The last equivalence uses that $|S|$ and $mk$ are integers. We have thus proved the exact characterization
\[
 f_j(S)\ne f_0(S)
 \ \Longleftrightarrow\ j\in S\ \text{and}\ |S|\le K.
\]
In particular, every set of size at least $mk$ is saturated under both oracles.

Choose $j$ uniformly from $N$. For any fixed query set $S$, the preceding characterization gives
\[
 \Prb_j[f_j(S)\ne f_0(S)]
 =\mathbf1_{\{|S|\le K\}}\frac{|S|}{n}
 \le\frac{\min\{K,n\}}{n}=:p.
\]
This estimate holds for every query address, including infeasible sets. Define the output event
\[
 A_j=\{|S|\le k\ \text{and}\ j\in S\}.
\]
Run the algorithm on $f_0$, and denote the probability of outputting $S$ by $q_S$; a failure symbol may be replaced by the empty set. This output distribution is independent of $j$. Therefore
\begin{align*}
 \E_j\Prb[A_j\mid f_0]
 &=\sum_{\substack{S\subseteq N\\|S|\le k}}q_S\Prb_j[j\in S]\\
 &=\sum_{\substack{S\subseteq N\\|S|\le k}}q_S\frac{|S|}{n}\\
 &\le\frac{k}{n}\sum_{\substack{S\subseteq N\\|S|\le k}}q_S
 \le\frac{k}{n}.
\end{align*}
Thus the reference-output parameter in Lemma~\ref{lower:lb:hybrid} can be taken to be $p_{\mathrm{out}}=k/n$. Keeping it separate from the query-disagreement parameter $p$ accounts for the fact that the output is a list of at most $k$ elements rather than an identified hidden element.

The algorithm succeeds with probability at least $2/3$ on every $f_j$, and success implies $A_j$. Applying Lemma~\ref{lower:lb:hybrid} gives
\[
 \frac23\le\E_j\Prb[A_j\mid f_j]
 \le\frac{k}{n}+2Q\sqrt{\frac{\min\{K,n\}}{n}}
 \le\frac12+2Q\sqrt{\frac{mk}{n}}.
\]
The last inequality uses $k\le n/2$ and $\min\{K,n\}\le mk$. Subtracting $1/2$ gives $1/6\le2Q\sqrt{mk/n}$; division by $2\sqrt{mk/n}>0$ yields the stated bound. This argument also excludes a zero-query algorithm.

Finally, the family is already finite and exact, so no rounding lemma is needed. With the public denominator $mk$, encode its values by the integers
\[
 mk f_j(S)=\min\{mk,mk\mathbf1_{\{j\in S\}}+|S\setminus\{j\}|\},
 \qquad mk f_0(S)=\min\{mk,|S|\}.
\]
These integers belong to $\{0,1,\ldots,mk\}$ and require $\lceil\log_2(mk+1)\rceil=O_a(\log(k+1))$ bits. The denominator is independent of the hidden element. Thus both the exact agreement property and the lower bound hold for the reversible digital oracle specified in Section~\ref{sec:prelim}.
\end{proof}

For $k=1$, this gives $\Omega(\sqrt n)$ queries. More generally, for any fixed $k$ and any fixed $\varepsilon\in(0,1/4)$, it matches the $\sqrt n$ dependence of our cardinality-constrained algorithm up to logarithmic factors, by taking $a=1-1/e-\varepsilon$.

\subsection{Cardinality-constrained maximization above \texorpdfstring{$1-1/e$}{one minus one over e}}\label{lower:sec:card-exponential}

\begin{theorem}[Sparse-Budget Exponential Lower Bound]\label{lower:thm:lower}
Let $0<\varepsilon<1/e$, let $1\le k\le n$, and assume $k/n\le\varepsilon$. Any quantum algorithm that, for every normalized nonnegative monotone submodular function, returns a set $S$ satisfying
\[
 |S|\le k,\qquad f(S)\ge(1-1/e+\varepsilon)\OPT_k(f)
\]
with probability at least $2/3$ has worst-case value-query complexity $Q$ satisfying
\[
 Q+1\ \ge\ \frac16\exp\!\left(\frac{\varepsilon^2 k}{16384}\right).
\]
For classical randomized algorithms with the same guarantee, the number $C$ of value queries satisfies
\[
 C+1\ \ge\ \frac23\exp\!\left(\frac{\varepsilon^2 k}{8192}\right).
\]
The hard instances have exact dyadic values encoded in $O(\log n+\log(1/\varepsilon))$ bits. In particular, when $\varepsilon^2k\to\infty$, both lower bounds are $\exp(\Omega(\varepsilon^2k))$.
\end{theorem}

\subsubsection{The hard family and its exact agreement region}

The classical inapproximability theorem alone is insufficient for a quantum lower bound. We use the following explicit properties of its construction.

\begin{lemma}[Feldman--Kuhnle Hard Family]\label{lower:lb:fk}
Let $1\le k\le n/2$, put $c=k/n$, and choose $d\in(0,c/4)$. For every $k$-element set $O\subseteq[n]$, there is a normalized nonnegative monotone submodular function $f_O:2^{[n]}\to[0,1]$. The function
\[
 g_{\mathrm c,0}(S)=1-(1-|S|/n)^{1/c}
\]
is also normalized, nonnegative, monotone, and submodular, and these functions satisfy
\begin{align}
 \OPT_k(f_O)&\ge 1-2d/c,\label{lower:lb:opt}\\
 f_O(S)\ne g_{\mathrm c,0}(S)&\ \Longrightarrow\
 |O\cap S|>c|S|+dn(1-c),\qquad S\subseteq[n].\label{lower:lb:agreement}
\end{align}
\end{lemma}

\noindent
This is the construction in Section~4.2.2 of Feldman and Kuhnle~\citep*{FK}: its submodularity is Lemma~4.21, its optimum bound is Observation~4.22, and its equality region is Observation~4.15. To match their notation, set
\[
 x=|S\cap O|/n,\qquad y=|S\setminus O|/n,
 \qquad \delta'=d.
\]
Their equality condition is $x\le cy/(1-c)+d$. Its translation is
\begin{align*}
 x\le\frac{cy}{1-c}+d
 &\ \Longleftrightarrow\
 (1-c)|S\cap O|\le c|S\setminus O|+dn(1-c)\\
 &\ \Longleftrightarrow\
 (1-c)|S\cap O|\le c(|S|-|S\cap O|)+dn(1-c)\\
 &\ \Longleftrightarrow\
 |S\cap O|\le c|S|+dn(1-c).
\end{align*}
Taking the contrapositive of this agreement implication gives~\eqref{lower:lb:agreement}. Although their asymptotic inapproximability theorem fixes the rational density $c$, these properties are statements about their explicit functions for every $0<c\le1/2$ and $0<d<c/4$. They therefore apply separately to every pair $(n,k)$ here, without requiring $c$ to be bounded away from zero.

Choose $O$ uniformly among all $k$-element subsets. For a fixed, possibly infeasible set $S$, the random variable $X=|O\cap S|$ is the sum of a sample of $k$ elements drawn without replacement from the $n$ membership indicators of $S$. Thus $\E X=k|S|/n=c|S|$. The sampling-without-replacement Hoeffding bound in \citep*[Theorem~4]{Hoeffding} implies
\begin{align}
 \Prb_O[f_O(S)\ne g_{\mathrm c,0}(S)]
 &\le \Prb_O[X-\E X>dn(1-c)]\notag\\
 &\le \exp\!\left(-\frac{2d^2n^2(1-c)^2}{k}\right)\notag\\
 &=\exp\!\left(-2\left(\frac{d(1-c)}{c}\right)^2k\right).
 \label{lower:lb:rare}
\end{align}
The last equality uses $n=k/c$. This is a bound for every individual query address, not an assumption on the distribution of addresses chosen by the algorithm.

\subsubsection{Finite digital encoding that preserves exact agreement}

Directly rounding a submodular function can destroy submodularity. The following construction first inserts a public amount of slack into each relevant inequality.

\begin{lemma}[Digital Encoding]\label{lower:lb:digital}
Suppose $\{f_z\}_z$ and $g_0$ are normalized nonnegative monotone submodular functions on an $n$-element ground set, with values at most $M$. For any $0<\lambda<1$, choose a dyadic grid spacing $\tau$ with
\[
 \frac{\lambda}{20n^2}<\tau\le\frac{\lambda}{10n^2}.
\]
There exist functions $\widetilde f_z,\widetilde g_0$ with the same normalization, nonnegativity, monotonicity, and submodularity properties, taking values in $\tau\mathbb Z_{\ge0}$, such that for every $S$,
\[
 f_z(S)-\tau\le\widetilde f_z(S)\le f_z(S)+\lambda,
 \qquad
 g_0(S)-\tau\le\widetilde g_0(S)\le g_0(S)+\lambda.
\]
Moreover, $f_z(S)=g_0(S)$ implies $\widetilde f_z(S)=\widetilde g_0(S)$, and the grid indices require $O(\log n+\log((M+1)/\lambda))$ bits.
\end{lemma}

\begin{proof}
Define the public function
\[
 h(S)=\frac{\lambda}{n^2}(2n|S|-|S|^2),
 \qquad
 \widetilde f_z(S)=\tau\left\lfloor\frac{f_z(S)+h(S)}{\tau}\right\rfloor,
\]
and define $\widetilde g_0$ by the identical formula with $g_0$ in place of $f_z$. To obtain the required spacing, take
\[
 \ell=\left\lceil\log_2\frac{10n^2}{\lambda}\right\rceil,
 \qquad \tau=2^{-\ell}.
\]
Indeed, $10n^2/\lambda\le2^\ell<20n^2/\lambda$, and taking reciprocals gives the stated interval. For $s=|S|\in[0,n]$,
\[
 2ns-s^2=n^2-(n-s)^2\in[0,n^2],
\]
so $0\le h(S)\le\lambda$. Write the downward rounding error as
\[
 r_z(S)=f_z(S)+h(S)-\widetilde f_z(S)\in[0,\tau).
\]
Then $\widetilde f_z(S)=f_z(S)+h(S)-r_z(S)$ lies between $f_z(S)-\tau$ and $f_z(S)+\lambda$, proving the approximation inequalities; the same calculation applies to $g_0$.

For $e\notin S$, the marginal of $h$ equals
\begin{align*}
 h(S\cup\{e\})-h(S)
 &=\frac{\lambda}{n^2}\bigl(2n(s+1)-(s+1)^2-2ns+s^2\bigr)\\
 &=\frac{\lambda}{n^2}(2n-2s-1)
 \ge\frac{\lambda}{n^2}.
\end{align*}
Here $s\le n-1$ because $e\notin S$. Consequently
\begin{align*}
 \widetilde f_z(S\cup\{e\})-\widetilde f_z(S)
 &=f_z(e\mid S)+h(e\mid S)-r_z(S\cup\{e\})+r_z(S)\\
 &\ge0+\lambda/n^2-\tau>0.
\end{align*}
Thus monotonicity survives rounding. For distinct $e,u\notin S$, put $H(s)=\lambda(2ns-s^2)/n^2$. The elementary submodularity slack of $h$ equals
\begin{align*}
 2H(s+1)-H(s)-H(s+2)
 &=\frac{\lambda}{n^2}
 \left[-2(s+1)^2+s^2+(s+2)^2\right]\\
 &=\frac{2\lambda}{n^2};
\end{align*}
the terms linear in $s$ cancel. The corresponding four-term expression for $f_z$ is nonnegative by submodularity. The rounding contribution is
\[
 -r_z(S\cup\{e\})-r_z(S\cup\{u\})
 +r_z(S)+r_z(S\cup\{e,u\})>-2\tau.
\]
The four-term expression for $\widetilde f_z$ is therefore at least $2\lambda/n^2-2\tau>0$. Equivalently,
\[
 \widetilde f_z(e\mid S)\ge\widetilde f_z(e\mid S\cup\{u\}).
\]
To obtain the full diminishing-returns inequality, take $A\subseteq B$ and $e\notin B$, enumerate $B\setminus A=\{u_1,\ldots,u_j\}$, and apply the preceding inequality successively:
\[
 \widetilde f_z(e\mid A)
 \ge\widetilde f_z(e\mid A\cup\{u_1\})
 \ge\cdots\ge\widetilde f_z(e\mid B).
\]
This proves submodularity. At the empty set, both $f_z$ and $h$ vanish, so its rounded value is zero. Rounding a nonnegative number to the grid by the floor operation preserves nonnegativity. Every argument applies with $g_0$ in place of $f_z$.

The added function and rounding rule are independent of $z$, so they preserve every equality $f_z(S)=g_0(S)$. Finally, every grid index is at most $(M+\lambda)/\tau=O(n^2(M+1)/\lambda)$, proving the register-size bound.
\end{proof}

\subsubsection{Proof of the exponential bound}

\begin{proof}[Proof of Theorem~\ref{lower:thm:lower}]
Set $c=k/n$ and write
\[
 A(c)=1-(1-c)^{1/c},\qquad \eta=\varepsilon/2,
 \qquad d=c\eta/32,\qquad\lambda=\eta/100.
\]
The assumptions imply $0<c\le\varepsilon<1/e<1/2$, so Lemma~\ref{lower:lb:fk} applies. Also $0<d<c/4$ and $0<\lambda<1$.

First we compare the desired approximation with the reference value. For $0<c<1$, integrating the geometric series on $[0,c]$ gives
\[
 -\log(1-c)=\int_0^c\frac{dt}{1-t}
 =\sum_{j=0}^{\infty}\int_0^c t^j\,dt
 =\sum_{j=0}^{\infty}\frac{c^{j+1}}{j+1}.
\]
The interchange is justified by uniform convergence on $[0,c]$. Dividing by $c$ and removing the $j=0$ term, define
\[
 u=\frac{-\log(1-c)}c-1=\sum_{j=1}^{\infty}\frac{c^j}{j+1}.
\]
All summands are nonnegative, and $1/(j+1)\le1/2$ for $j\ge1$. Therefore
\[
 0\le u\le\frac12\sum_{j=1}^{\infty}c^j
 =\frac{c}{2(1-c)}\le c.
\]
The last inequality uses $c\le1/2$, which implies $2(1-c)\ge1$. Moreover,
\[
 (1-c)^{1/c}=\exp\!\left(\frac{\log(1-c)}c\right)=e^{-1-u},
 \qquad 1-e^{-u}=\int_0^u e^{-t}\,dt\le u.
\]
Consequently
\[
 0\le A(c)-(1-1/e)=e^{-1}(1-e^{-u})\le c/e\le\varepsilon/e.
\]
Subtracting $A(c)$ from the desired ratio now yields
\[
 (1-1/e+\varepsilon)-A(c)
 \ge\varepsilon-\varepsilon/e=\varepsilon(1-1/e)>\varepsilon/2=\eta.
\]
Thus $1-1/e+\varepsilon>A(c)+\eta$.

Apply Lemma~\ref{lower:lb:digital} to $\{f_O\}$ and $g_{\mathrm c,0}$ with this $\lambda$. For $|S|\le k=cn$, one has $1-|S|/n\ge1-c$; raising both nonnegative sides to the positive power $1/c$ and subtracting from one gives $g_{\mathrm c,0}(S)\le A(c)$. Hence
\[
 \widetilde g_{\mathrm c,0}(S)\le g_{\mathrm c,0}(S)+\lambda\le A(c)+\lambda.
\]
For every feasible $S$, Lemma~\ref{lower:lb:digital} gives $\widetilde f_O(S)\ge f_O(S)-\tau$. Maximizing this inequality over feasible sets and using~\eqref{lower:lb:opt}, the hidden optimum obeys
\[
 \OPT_k(\widetilde f_O)\ge\OPT_k(f_O)-\tau
 \ge1-2d/c-\tau=1-\eta/16-\tau.
\]
This lower bound is positive because $\eta=\varepsilon/2<1/(2e)$ and $\tau\le\eta/1000$. Since $A(c)+\eta<1-1/e+\varepsilon<1$, we have
\begin{align*}
 (A(c)+\eta)(1-\eta/16-\tau)-(A(c)+\lambda)
 &=\eta-\lambda-(A(c)+\eta)(\eta/16+\tau)\\
 &\ge\eta-\lambda-\eta/16-\tau\\
 &\ge\eta\left(1-\frac1{100}-\frac1{16}-\frac1{1000}\right)>0.
\end{align*}
Here $\tau\le\lambda/(10n^2)\le\eta/1000$. Combining these inequalities gives the explicit separation
\[
 (1-1/e+\varepsilon)\OPT_k(\widetilde f_O)
 >(A(c)+\eta)(1-\eta/16-\tau)>A(c)+\lambda.
\]
Consequently every successful feasible output must satisfy
\[
 A_O:=\{\widetilde f_O(S)\ne\widetilde g_{\mathrm c,0}(S)\}.
\]
Indeed, equality would force its value to be at most $A(c)+\lambda$, strictly below the value required by the target approximation.

Digital encoding cannot introduce a new disagreement. Substituting $d=c\eta/32$ into~\eqref{lower:lb:rare} and using $1-c\ge1/2$, for every query address $S$ we obtain
\begin{align*}
 \Prb_O[\widetilde f_O(S)\ne\widetilde g_{\mathrm c,0}(S)]
 &\le\exp\!\left(-\frac{\eta^2(1-c)^2k}{512}\right)\\
 &\le\exp\!\left(-\frac{\eta^2k}{2048}\right)
 =\exp\!\left(-\frac{\varepsilon^2k}{8192}\right)=:p.
\end{align*}
When run on $\widetilde g_{\mathrm c,0}$, the output distribution is independent of $O$. Denote its probability of outputting $S$ by $q_S$. Then
\begin{align*}
 \E_O\Prb[A_O\mid\widetilde g_{\mathrm c,0}]
 &=\sum_S q_S\Prb_O[\widetilde f_O(S)\ne\widetilde g_{\mathrm c,0}(S)]\\
 &\le p\sum_S q_S=p.
\end{align*}
One can replace a failure symbol by the empty set, so the output distribution is on subsets. The bound also includes infeasible outputs, because the address-wise estimate holds for every subset.

Success on every hidden instance implies average success at least $2/3$. Lemma~\ref{lower:lb:hybrid}, with $p_{\mathrm{out}}=p$, therefore yields
\[
 \frac23\le p+2Q\sqrt p,\qquad
 \frac23\le(C+1)p.
\]
The classical inequality gives the stated bound directly. If $p\le1/3$, the quantum inequality gives $Q\ge1/(6\sqrt p)$. If $p>1/3$, the trivial bound $Q+1\ge1$ implies $Q+1\ge1/(6\sqrt p)$. Thus in both cases
\[
 Q+1\ge\frac16p^{-1/2}
 =\frac16\exp\!\left(\frac{\varepsilon^2k}{16384}\right).
\]
Finally, the original functions have $M=1$ and $\lambda=\varepsilon/200$, so Lemma~\ref{lower:lb:digital} supplies the claimed number of output bits. All estimates are uniform in $n,k,\varepsilon$ under the stated assumptions.
\end{proof}

\section*{Acknowledgment}
X.S. and J. Z. are supported by the Quantum Science and Technology-National Science and Technology Major Project (Grant No.~2024ZD0300500), and the National Natural Science Foundation of China (Grants No.~62325210, 62272441).
P.Y. is supported by the National Natural Science Foundation of China (Grant Nos. 62332009 and 12347104), the Quantum Science and Technology—National Science and Technology Major Project (Grant No. 2021ZD0302901), the NSFC/RGC Joint Research Scheme (Grant No. 12461160276), the Natural Science Foundation of Jiangsu Province (No. BK20243060), the Fundamental and Interdisciplinary Disciplines Breakthrough Plan of the Ministry of Education of China (No. JYB2025XDXM118), the "111 Center"(No. B26023), and the Fundamental Research Funds for the Central Universities (Grant no. 2026300376). Z.Y. is supported by start-up research funding from Fuzhou University (Grant No. XRC-26004). Z.Z. is supported by the National Natural Science Foundation of China (Grants No.~62402110).

\section*{AI disclosure}
This work was developed through iterative interactions with
GPT-5.6 Sol/GPT-6 Astra, which assisted in exploring algorithmic ideas,
developing and refining proofs of the approximation and running-time
guarantees, examining connections to related work, and drafting and
revising the manuscript. The authors have verified the results
presented in this paper and take full responsibility for its content
and correctness.

\clearpage
\phantomsection
\addcontentsline{toc}{section}{References}
\bibliographystyle{plainnat}
\bibliography{quantum_submodular_maximization}
\end{document}